\documentclass[
a4paper,reqno]{amsart}
\usepackage[T1]{fontenc}
\usepackage[english]{babel}
\usepackage{amssymb,bbm,enumerate}
\usepackage{color}
\usepackage[linktocpage=true,colorlinks=true, linkcolor=blue, citecolor=red, urlcolor=green]{hyperref}

\usepackage{float}
\restylefloat{table}

\renewcommand{\phi}{\varphi}
\renewcommand{\ker}{\Ker}

\newcommand{\mc}[1]{\mathcal{#1}}
\newcommand{\mf}[1]{\mathfrak{#1}}
\newcommand{\mb}[1]{\mathbb{#1}}

\newcommand{\id}{\mathbbm{1}}
\newcommand{\quot}[2] {\ensuremath{\raisebox{.40ex}{\ensuremath{#1}}
\! \big / \! \raisebox{-.40ex}{\ensuremath{#2}}}}
\newcommand{\tint}{{\textstyle\int}}

\DeclareMathOperator{\Mat}{Mat}

\DeclareMathOperator{\diag}{diag}
\DeclareMathOperator{\tr}{Tr}

\DeclareMathOperator{\ad}{ad}
\DeclareMathOperator{\im}{Im}

\DeclareMathOperator{\Ker}{Ker}

\theoremstyle{plain}
\newtheorem{theorem}{Theorem}[section]
\newtheorem{lemma}[theorem]{Lemma}
\newtheorem{proposition}[theorem]{Proposition}
\newtheorem{corollary}[theorem]{Corollary}

\theoremstyle{definition}
\newtheorem{definition}[theorem]{Definition}
\newtheorem{example}[theorem]{Example}

\theoremstyle{remark}
\newtheorem{remark}[theorem]{Remark}

\numberwithin{equation}{section}

\definecolor{light}{gray}{.9}

\title{Classical \texorpdfstring{$\mc W$}{W}-algebras 
and generalized Drinfeld-Sokolov hierarchies
for minimal and short nilpotents}

\author{Alberto De Sole, Victor G. Kac, Daniele Valeri}

\address{Dipartimento di Matematica, Sapienza Universit\`a di Roma,
P.le Aldo Moro 2, 00185 Rome, Italy.}
\email{desole@mat.uniroma1.it}

\address{Department of Math., MIT,
77 Massachusetts Avenue, Cambridge, MA 02139, USA.}
\email{kac@math.mit.edu}

\address{SISSA, Via Bonomea 265, 34136 Trieste, Italy.}
\email{dvaleri@sissa.it}

\begin{document}

\pagestyle{plain}

\begin{abstract}
We derive explicit formulas for $\lambda$-brackets
of the affine classical $\mc W$-algebras attached to the minimal and short
nilpotent elements of any simple Lie algebra $\mf g$.
This is used to compute explicitly the first non-trivial PDE of the corresponding
integrable generalized Drinfeld-Sokolov hierarchies.
It turns out that a reduction of the equation corresponding to a short nilpotent
is Svinolupov's equation attached to a simple Jordan algebra,
while a reduction of the equation corresponding to a minimal nilpotent
is an integrable Hamiltonian equation
on $2h\,\check{}-3$ functions,
where $h\,\check{}$ is the dual Coxeter number of $\mf g$.
In the case when $\mf g$ is $\mf{sl}_2$ both these equations coincide with the KdV equation.
In the case when $\mf g$ is not of type $C_n$,
we associate to the minimal nilpotent element of $\mf g$
yet another generalized Drinfeld-Sokolov hierarchy.
\end{abstract}

\maketitle

\tableofcontents

\section{Introduction}\label{sec:1}

In our paper \cite{DSKV12}
we put the theory of Drinfeld-Sokolov Hamiltonian reduction \cite{DS85}
in the framework of Poisson vertex algebras (PVA).
Using this, we gave a simple construction of the affine classical $\mc W$-algebras $\mc W(\mf g,f)$,
attached to an arbitrary nilpotent element $f$ of a simple Lie algebra $\mf g$.
Furthermore, for a large number of nilpotents $f$ we constructed in this framework
the associated generalized Drinfeld-Sokolov hierarchy of bi-Hamiltonian equations
and proved its integrability (the case studied by Drinfeld and Sokolov corresponds to 
the principal nilpotent element $f$).

In the present paper we study in detail the cases when $f$ is a ``minimal''
and a ``short'' nilpotent element.
Let $\{e,2x,f\}\subset\mf g$ be an $\mf{sl}_2$ triple containing $f$.
The element $f$ is called \emph{minimal} if its adjoint orbit has minimal dimension
among all non-zero nilpotent orbits;
equivalently, if the $\ad(x)$-eigenspace decomposition of $\mf g$
has the form
$$
\mf g=\mf g_{-1}\oplus\mf g_{-\frac12}\oplus\mf g_0\oplus\mf g_{\frac12}\oplus\mf g_1\,,
$$
where $\dim(\mf g_{\pm1})=1$. 
If $\theta$ is the highest root of $\mf g$,
the corresponding root vector $f=e_{-\theta}$ is a minimal nilpotent element,
and its adjoint orbit consists of all minimal nilpotent elements of $\mf g$.

The element $f$ is called \emph{short} if the $\ad(x)$-eigenspace decomposition
has the form
$$
\mf g=\mf g_{-1}\oplus\mf g_0\oplus\mf g_1\,;
$$
equivalently, if the product $\circ$ on $\mf g_{-1}$
(or, respectively, the product $*$ on $\mf g_1$) defined by 
\begin{equation}\label{20130320:eq1}
a\circ b=[[e,a],b]
\,\,,\,\,\,\, a,b\in\mf g_{-1}\,,
\,\,\,\,
\Big(\text{ resp. } 
a* b=[[f,a],b]
\,\,,\,\,\,\, a,b\in\mf g_{1}
\Big)\,,
\end{equation}
gives $\mf g_{-1}$ (resp. $\mf g_1$) the structure of a Jordan algebra.
In fact, the classification of conjugacy classes of short nilpotent elements in simple Lie algebras
corresponds to the classification of simple Jordan algebras \cite{Jac81}.
Recall that a complete list of conjugacy classes of
short nilpotent elements in simple Lie algebras is as follows:
for $\mf g$ of type $A_n$, with odd $n$,
of type $B_n$ and $C_n$, with arbitrary $n$,
and of type $D_4$ and $E_7$,
there is a unique conjugacy class of short nilpotent elements,
while for $\mf g$ of type $D_n$, with $n\geq5$,
there are two conjugacy classes of short nilpotent elements.
In all other cases there are no short nilpotent elements.

In the case of minimal and short nilpotent elements,
we describe explicitly the PVA structure on $\mc W(\mf g,f)$
and the corresponding generalized Drinfeld-Sokolov hierarchies,
associated to a choice of $s\in\mf g_1$,
for which $f+s$ is a semisimple element of $\mf g$.

In the case when $f$ is a short nilpotent element and $s=e$,
the corresponding hierarchy admits a reduction.
It turns out that the reduction of the simplest equation of this hierarchy
coincides with the integrable equation associated to a simple Jordan algebra 
by Svinolupov \cite{Svi91}.
It follows that Svinolupov's equations are Hamiltonian.
In the case when $f$ is a minimal nilpotent element
we get after a reduction an integrable Hamiltonian equation
on $2h\,\check{}-3$ functions.
In a forthcoming publication \cite{DSKV13} we construct for both cases the second Poisson structure,
which is non-local, via an analogue of Dirac reduction \cite{Dir50}.

The generalized Drinfeld-Sokolov hierarchy depends also on the choice of an isotropic subspace $\mf l\subset\mf g_{\frac12}$ \cite{DSKV12}.
In the above examples we chose $\mf l=0$.
We show that if
$\mf g$ is not of type $C_n$ and $f$ is a minimal nilpotent,
one can choose a maximal isotropic subspace $\mf l\subset\mf g_{\frac12}$
and $s\in\mf l$ such that $f+s$ is semisimple,
which leads us to yet another integrable generalized Drinfeld-Sokolov hierarchy.

Throughout the paper, unless otherwise specified, all vector spaces, tensor products etc.,
are defined over an algebraically closed field $\mb F$ of characteristic 0.

\subsubsection*{Acknowledgments}
We wish to thank the IHP and IHES, France,
where part of this research was conducted.
We wish to thank Corrado De Concini for useful discussions
and Andrea Maffei for illuminating observations.
The first author was partially supported by the University grants 
C26F09X3JP, C26A09EFE7, C26A1078PE, C26A11JPEP,
and by the FIRB National grant RBFR12RA9W.
The second author was partially supported by the Simons Fellowship.
The third author was supported by the ERC grant
``FroM-PDE: Frobenius Manifolds and Hamiltonian Partial
Differential Equations''.

\section{Poisson vertex algebras and Hamiltonian equations}\label{sec:0.5}

Recall (see e.g. \cite{DSK06,BDSK09}) 
that a \emph{Poisson vertex algebra} (PVA) is a commutative associative differential algebra $\mc V$,
with derivation $\partial$, endowed with a $\lambda$-bracket 
$\mc V\otimes\mc V\to\mc V[\lambda]$,
denoted $g\otimes h\mapsto\{g_\lambda h\}$,
satisfying the following axioms:
\begin{enumerate}[(i)]
\item
sesquilinearity: 
$\{\partial g_\lambda h\}=-\lambda\{g_\lambda h\}$, 
$\{g_\lambda \partial h\}=(\partial+\lambda)\{g_\lambda h\}$,
\item
skew-symmetry: 
$\{g_\lambda h\}=-\{h_{-\lambda-\partial} g\}$ ($\partial$ should be moved to the left 
to act on the coefficients), 
\item
Jacobi identity:
$\{{g_1}_\lambda{\{{g_2}_\mu h\}}\}-\{{g_2}_\mu{\{{g_1}_\lambda h\}}\}
=\{{\{{g_1}_\lambda {g_2}\}}_{\lambda+\mu}h\}$,
\item
left Leibniz rule:
$\{g_\lambda h_1h_2\}=\{g_\lambda h_1\}h_2+\{g_\lambda h_2\}h_1$,
\item
right Leibniz rule:
$\{{g_1g_2}_\lambda h\}=\{{g_1}_{\partial+\lambda} h\}_\to g_2+\{{g_2}_{\partial+\lambda} h\}{g_1}$
(where the arrow means that $\partial$ should be moved to the right).
\end{enumerate}
Recall also that a \emph{Lie conformal algebra} is an $\mb F[\partial]$-module
endowed with a $\lambda$-bracket satisfying axioms (i), (ii) and (iii).

An element $L\in\mc V$ in a PVA $\mc V$ is called a \emph{Virasoro element} if
$$
\{L_\lambda L\}=(\partial+2\lambda)L+c\lambda^3+\alpha\lambda\,,
$$
where $c,\alpha\in\mb F$.
(The number $c$ is called the \emph{central charge} of $L$.)
An element $a\in\mc V$ is called an \emph{eigenvector} for $L$
of \emph{conformal weight} $\Delta_a\in\mb F$ if 
$$
\{L_\lambda a\}=(\partial+\Delta_a\lambda)a+O(\lambda^2)\,.
$$
It is called a \emph{primary element} of conformal weight $\Delta_a$
if $\{L_\lambda a\}=(\partial+\Delta_a\lambda)a$.
A Virasoro element is called an \emph{energy momentum element}
if there exists a basis of $\mc V$ consisting of eigenvectors of $L$.
Clearly, by sesquilinearity and the left Leibniz rule, 
a Virasoro element is an energy momentum element if and only if there exists
a set generating $\mc V$ as differential algebra consisting of eigenvectors of $L$.

The basic example is the algebra of differential polynomials in $\ell$ differential variables
$\mc V_\ell=\mb F[u_i^{(n)}\,|\,i=1,\dots,\ell,n\in\mb Z_+]$,
with the derivation $\partial u_i^{(n)}=u_i^{(n+1)}$.
A $\lambda$-bracket on $\mc V_\ell$ is introduced by letting
$$
\{{u_i}_\lambda{u_j}\}=H_{ji}(\lambda)\,\in\mc V[\lambda]
\,\,,\,\,\,\,
i,j=1,\dots,\ell\,,
$$
and extending (uniquely) to the whole space $\mc V_\ell$ by sesquilinearity and Leibniz rules.
Then we have, for arbitrary $h,g\in\mc V_\ell$, the following Master Formula \cite{DSK06}:
\begin{equation}\label{masterformula}
\{g_\lambda h\}=
\sum_{\substack{i,j\in I\\m,n\in\mb Z_+}}\frac{\partial g}{\partial u_j^{(n)}}(\partial+\lambda)^n
H_{ji}(\partial+\lambda)(-\lambda-\partial)^m\frac{\partial h}{\partial u_i^{(m)}}\,.
\end{equation}
It is proved in \cite{BDSK09} that the $\lambda$-bracket \eqref{masterformula}
defines a structure of a PVA on $\mc V_\ell$ if and only if
skew-symmetry and the Jacobi identity hold on the generators $u_i$.
In this case we say that the $\ell\times\ell$ matrix differential operator
$H=\big(H_{ij}(\partial)\big)_{i,j=1}^\ell$ is a \emph{Poisson structure} on $\mc V_\ell$.

As usual, for a PVA $\mc V$,
we call $\quot{\mc V}{\partial\mc V}$ the space of Hamiltonian functionals,
and we denote by $\tint:\,\mc V\to\quot{\mc V}{\partial\mc V}$ the canonical quotient map.
It follows from the sesquilinearity that
$$
\{\tint g,\tint h\}:=\tint\{g_\lambda h\}\big|_{\lambda=0}\in\quot{\mc V}{\partial\mc V}
\,\,\text{ and }\,\,
\{\tint g,h\}:=\{g_\lambda h\}\big|_{\lambda=0}\in\mc V
\,,
$$
are well defined.
It follows from skewsymmetry and the Jacobi identity axioms of a PVA
that $\{\tint f,\tint g\}$ is a Lie algebra bracket on $\quot{\mc V}{\partial\mc V}$,
and that $\{\tint f,g\}$ is a representation of  $\quot{\mc V}{\partial\mc V}$ on $\mc V$
by derivations of the associative product on $\mc V$.

The \emph{Hamiltonian equation} associated 
to a Hamiltonian functional $\tint h\in\quot{\mc V}{\partial\mc V}$
and the Poisson structure $H$ is, by definition,
\begin{equation}\label{20130314:eq1}
\frac{du_i}{dt}=\{\tint h,u_i\}
\,\,,\,\,\,\,
i=1,\dots,\ell\,.
\end{equation}
Note that the RHS of \eqref{20130314:eq1} is $H(\partial)\frac{\delta h}{\delta u}$,
where $\frac{\delta h}{\delta u}\in\mc V^\ell$ is the vector of variational derivatives
$\frac{\delta h}{\delta u_i}
=\sum_{n\in\mb Z_+}(-\partial)^n\frac{\partial h}{\partial u_i^{(n)}},\,n\in\mb Z_+$.
An \emph{integral of motion} of this equation is 
a ``Hamiltonian'' functional $\tint g\in\quot{\mc V}{\partial\mc V}$
such that $\{\tint h,\tint g\}=0$,
i.e. $\tint h,\tint g$ are \emph{in involution}.
Equation \eqref{20130314:eq1} is called \emph{integrable}
if there are infinitely many linearly independent integrals of motion in involution,
$\tint h_n,\,n\in\mb Z_+$, where $h_0=h$.
In this case, we have an integrable hierarchy of Hamiltonian equations
$$
\frac{du_i}{dt_n}=\{\tint h_n,u_i\}
\,\,,\,\,\,\,
i=1,\dots,\ell,\,n\in\mb Z_+\,.
$$

The main tool for proving integrability is the so called \emph{Lenard-Magri scheme}
(see e.g. \cite{BDSK09}).
It can be applied to a \emph{bi-Hamiltonian equation}, 
meaning an equation that can be written in the form \eqref{20130314:eq1}
in two different ways:
\begin{equation}\label{20130314:eq2}
\frac{du_i}{dt}
=\{\tint h_0,u_i\}_H
=\{\tint h_1,u_i\}_K
\,,
\end{equation}
where $\{\cdot\,,\,\cdot\}_H$ and $\{\cdot\,,\,\cdot\}_K$
are the Lie algebra brackets on $\quot{\mc V}{\partial\mc V}$
associated to the Poisson structures $H$ and $K$,
which are assumed to be compatible, in the sense that any their linear combination
$\alpha H+\beta K$, for $\alpha,\beta\in\mb F$, is also a Poisson structure.
The scheme consists in finding a sequence 
of Hamiltonian functionals $\tint h_n,\,n\in\mb Z_+$,
starting with the two given ones,
satisfying the following recursive equations, for all $g\in\mc V$,
\begin{equation}\label{20130314:eq3}
\{{\tint h_0},\tint g\}_{K}=0
\,\,,\,\,\,\,
\{{\tint h_n},\tint g\}_{H}=\{{\tint h_{n+1}},\tint g\}_{K}
\,\,\text{ for all } n\in\mb Z_+
\,.
\end{equation}
Then, automatically all Hamiltonian functionals $\tint h_n$ are in involution
with respect to both Poisson brackets $\{\cdot\,,\,\cdot\}_H$ and $\{\cdot\,,\,\cdot\}_K$
(see \cite{Mag78} or \cite[Lem.2.6]{BDSK09} for a proof of this simple fact),
and therefore equation \eqref{20130314:eq2} is integrable, provided that the $\tint h_n$'s 
span an infinite dimensional vector space, \cite{Mag78}.

\section{Structure of classical \texorpdfstring{$\mc W$}{W}-algebras}\label{sec:2}

In this section we remind the definition of classical $\mc W$-algebras in the language 
of Poisson vertex algebras, following \cite{DSKV12}.

%

\subsection{Setup and notation}
\label{sec:2.1}

Let $\mf g$ be a simple finite-dimensional Lie algebra over the field $\mb F$
with a non-degenerate symmetric invariant bilinear form $(\cdot\,|\,\cdot)$,
and let $\{f,2x,e\}\subset\mf g$ be an $\mf{sl}_2$-triple in $\mf g$.
We have the corresponding $\ad x$-eigenspace decomposition
\begin{equation}\label{dec}
\mf g=\bigoplus_{i\in\frac{1}{2}\mb Z}\mf g_{i}\,.
\end{equation}
Clearly, $f\in\mf g_{-1}$, $x\in\mf g_{0}$ and $e\in\mf g_{1}$.
Given an $\ad x$-invariant subspace $\mf a\subset\mf g$,
we denote $\mf a_j=\mf a\cap\mf g_j,\,j\in\frac12\mb Z$,
and by $\mf a_{\leq k}=\bigoplus_{j\leq k}\mf a_j$,
or $\mf a_{\geq k}=\bigoplus_{j\geq k}\mf a_j$,
for arbitrary $k\in\frac12\mb Z$.
We have the nilpotent subalgebras:
$\mf g_{\geq1}
\subset
\mf g_{\geq\frac12}
\subset\mf g$.
Fix an element $s\in\mf g_d$, where $d$ is the maximal $i\in\frac12\mb Z$
for which $\mf g_i\neq0$ in \eqref{dec},
called the \emph{depth} of the grading \eqref{dec}.
(It is easy to show that $\mf g_d$ is the center of $\mf g_{\geq\frac12}$.)

Let $\mf g^f$ be the centralizer of $f$ in $\mf g$.
By representation theory of $\mf{sl}_2$,
we have $\mf g^f\subset\mf g_{\leq0}$,
and we have the direct sum decomposition
\begin{equation}\label{20130402:eq1}
\mf g_{\leq\frac12}=\mf g^f\oplus[e,\mf g_{\leq-\frac12}]\,.
\end{equation}
We let $\pi:\,\mf g_{\leq\frac12}\mapsto\mf g^f$
be the quotient map with kernel $[e,\mf g_{\leq-\frac12}]$,
and, for $a\in\mf g_{\leq\frac12}$, 
we denote $\pi(a)=a^\sharp\in\mf g^f$.

Clearly, $\ad e$ and $\ad f$ restrict to bijective maps 
\begin{equation}\label{20130201:eq1}
\ad e:\,\mf g_{-\frac12}\stackrel{\sim}{\longrightarrow}\mf g_{\frac12}
\,\,\text{ and }\,\,
\ad f:\,\mf g_{\frac12}\stackrel{\sim}{\longrightarrow}\mf g_{-\frac12}
\,.
\end{equation}
Moreover, since $x=\frac12[e,f]$, 
it immediately follows by the invariance of the bilinear form $(\cdot\,|\,\cdot)$ that
\begin{equation}\label{20130201:eq2}
(x|x)=\frac12(e|f)\,.
\end{equation}

Note that the bilinear form $(\cdot\,|\,\cdot)$ restricts to a non-degenerate pairing
between $\mf g_{-j}$ and $\mf g_j$.
Fix $\{u_j\}_{j\in J},\,\{u^j\}_{j\in J}$, bases of $\mf g$ consisting of eigenvectors of $\ad x$
and dual with respect to $(\cdot\,|\,\cdot)$.
Let $\{q_i\}_{i\in J_{\leq\frac12}}$ be the (sub)basis of $\mf g_{\leq\frac12}$,
and let $\{q^i\}_{i\in J_{\leq\frac12}}$ be the corresponding dual
basis of $\mf g_{\geq-\frac12}$.
Within these bases, we have the dual bases 
$\{a_j\}_{j\in J_0}$ and $\{a^j\}_{j\in J_0}$ of $\mf g_0$.

Since the pairing between $\mf g_{-\frac12}$ and $\mf g_{\frac12}$ is non-degenerate,
using the bijective maps \eqref{20130201:eq1},
we obtain induced non-degenerate skew-symmetric inner products on $\mf g_{-\frac12}$, 
given by
\begin{equation}\label{20130201:eq4}
\omega_-(u,u_1):=(e|[u,u_1])\,\, \text{ for } u,u_1\in\mf g_{-\frac12}
\,,
\end{equation}
and on $\mf g_{\frac12}$, given by
\begin{equation}\label{20130201:eq5}
\omega_+(v,v_1):=(f|[v,v_1])\,\,\text{ for } v,v_1\in\mf g_{\frac12}
\,.
\end{equation}
\begin{lemma}\phantomsection\label{20130315:lem1}
\begin{enumerate}[(a)]
\item
For $u,u_1\in\mf g_{-\frac12}$ and $a\in\mf g_0$ we have
$$
\omega_-(u,[a,u_1])=\omega_-([u,a],u_1)+([e,a]|[u,u_1])\,.
$$
\item
For $v,v_1\in\mf g_{\frac12}$ and $a\in\mf g_0$ we have
$$
\omega_+(v,[a,v_1])=\omega_+([v,a],v_1)+([f,a]|[v,v_1])
\,.
$$ 
\end{enumerate}
In particular, the bilinear forms $\omega_-$ and $\omega_+$ are invariant with respect to $\mf g_0^f$.
\end{lemma}
\begin{proof}
It is straightforward, using 
the invariance of the bilinear form $(\cdot\,|\,\cdot)$.
\end{proof}
Let $\{v_k\}_{k\in J_{\frac12}}$ be a basis of $\mf g_{\frac12}$,
and let $\{v^k\}_{k\in J_{\frac12}}$ be the dual basis, again of $\mf g_{\frac12}$, 
with respect to the skew-symmetric inner product $\omega_+$ 
defined in \eqref{20130201:eq5}.
Equivalently, the collection of elements $\{[f,v_k]\}_{k\in J_{\frac12}}\subset\mf g_{-\frac12}$
and $\{v^k\}_{k\in J_{\frac12}}\subset\mf g_{\frac12}$
are dual bases of $\mf g_{-\frac12}$ and $\mf g_{\frac12}$ with respect to $(\cdot\,|\,\cdot)$.
In other words, we have the orthogonality conditions
\begin{equation}\label{omega+dual}
\omega_+(v_h,v^k)=(f|[v_h,v^k])=\delta_{h,k}\,,
\end{equation}
for all $h,k\in J_{\frac12}$.
Note that the orthogonality conditions \eqref{omega+dual}
are equivalent to the following completeness relations for elements of $\mf g_{\frac12}$:
\begin{equation}\label{completeness}
\sum_{k\in J_{\frac12}} \omega_+(v,v^k)v_k
=-\sum_{k\in J_{\frac12}} \omega_+(v,v_k) v^k=v
\quad \text{ for all } v\in\mf g_{\frac12}\,.
\end{equation}
\begin{lemma}\label{20130315:lem2}
Suppose that $\mf g_{\frac32}=0$. Then:
\begin{enumerate}[(a)]
\item
The bijective maps \eqref{20130201:eq1} are inverse to each other,
i.e. $[e,[f,v]]=v$ for all $v\in\mf g_{\frac12}$,
and $[f,[e,u]]=u$ for all $u\in\mf g_{-\frac12}$.
\item
The collections of elements $\{[f,v_k]\}_{k\in J_{\frac12}}$ and $\{[f,v^k]\}_{k\in J_{\frac12}}$ 
are bases of $\mf g_{-\frac12}$ dual with respect to $-\omega_-$:
\begin{equation}\label{omega-dual}
\omega_-([f,v_h],[f,v^k]):=(e|[[f,v_h],[f,v^k]])=-\delta_{h,k}\,.
\end{equation}
\item
For every $u\in\mf g_{-\frac12}$ we have the completeness relation
\begin{equation}\label{completeness2}
\sum_{k\in J_{\frac12}}(u|v^k)v_k
=-\sum_{k\in J_{\frac12}}(u|v_k)v^k
=[e,u]
\quad \text{ for all } u\in\mf g_{-\frac12}
\,.
\end{equation}
\end{enumerate}
\end{lemma}
\begin{proof}
By the Jacobi identity we have, for $v\in\mf g_{\frac12}$,
$[e,[f,v]]=[f,[e,v]]+[[e,f],v]$. 
Since, by assumption, $\mf g_{\frac32}=0$, we have $[e,v]=0$.
Moreover, $[[e,f],v]=2[x,v]=v$, proving part (a).
Equation \eqref{omega-dual} follows immediately by invariance of $(\cdot\,|\,\cdot)$
and part (a).
Finally, equation \eqref{completeness2} 
follows from \eqref{completeness} and part (a).
\end{proof}

\subsection{Definition of the classical \texorpdfstring{$\mc W$}{W}-algebra}
\label{sec:2.2}

Consider the algebra of differential polynomials $\mc V(\mf g)=S(\mb F[\partial]\mf g)$.
For $P\in\mc V(\mf g)$, we use the standard notation $P'=\partial P$
and $P^{(n)}=\partial^n P$ for $n\in\mb Z_+$.
We have a PVA structure on $\mc V(\mf g)$, with $\lambda$-bracket given on generators by 
\begin{equation}\label{lambda}
\{a_\lambda b\}_z=[a,b]+(a| b)\lambda+z(s|[a,b]),
\qquad a,b\in\mf g\,,
\end{equation}
and extended to $\mc V(\mf g)$ by the Master Formula \eqref{masterformula}.
Here $z$ can be viewed as an element of the field $\mb F$
or as a formal variable, in which case we should replace $\mc V(\mf g)$
by $\mc V(\mf g)[z]$.
Since, by assumption, $s\in Z(\mf g_{\geq\frac12})$,
the $\mb F[\partial]$-submodule
$\mb F[\partial]\mf g_{\geq\frac12}\subset\mc V(\mf g)$ 
is a Lie conformal subalgebra 
with the $\lambda$-bracket $\{a_\lambda b\}_z=[a,b]$, $a,b\in\mf g_{\geq\frac12}$
(it is independent on $z$).
Consider the differential subalgebra
$\mc V(\mf g_{\leq\frac12})=S(\mb F[\partial]\mf g_{\leq\frac12})$ of $\mc V(\mf g)$,
and denote by $\rho:\,\mc V(\mf g)\twoheadrightarrow\mc V(\mf g_{\leq\frac12})$,
the differential algebra homomorphism defined on generators by
\begin{equation}\label{rho}
\rho(a)=\pi_{\leq\frac12}(a)+(f| a),
\qquad a\in\mf g\,,
\end{equation}
where $\pi_{\leq\frac12}:\,\mf g\to\mf g_{\leq\frac12}$ denotes 
the projection with kernel $\mf g_{\geq1}$.
Recall from \cite{DSKV12} that
we have a representation of the Lie conformal algebra $\mb F[\partial]\mf g_{\geq\frac12}$ 
on the differential subalgebra $\mc V(\mf g_{\leq\frac12})\subset\mc V(\mf g)$ given by
($a\in\mf g_{\geq\frac12}$, $g\in\mc V(\mf g_{\leq\frac12})$):
\begin{equation}\label{20120511:eq1}
a\,^\rho_\lambda\,g=\rho\{a_\lambda g\}_z
\end{equation}
(note that the RHS is independent of $z$ since, by assumption, $s\in Z(\mf g_{\geq\frac12})$).

The \emph{classical} $\mc W$-\emph{algebra} is, by definition,
the differential algebra
\begin{equation}\label{20120511:eq2}
\mc W
=\big\{g\in\mc V(\mf g_{\leq\frac12})\,\big|\,a\,^\rho_\lambda\,g=0\,\text{ for all }a\in\mf g_{\geq\frac12}\}\,,
\end{equation}
endowed with the following PVA $\lambda$-bracket
\begin{equation}\label{20120511:eq3}
\{g_\lambda h\}_{z,\rho}=\rho\{g_\lambda h\}_z,
\qquad g,h\in\mc W\,.
\end{equation}
\begin{theorem}[\cite{DSKV12}]\phantomsection\label{daniele1}
\begin{enumerate}[(a)]
\item
Equation \eqref{20120511:eq1} defines a representation 
of the Lie conformal algebra $\mb F[\partial]\mf g_{\geq\frac12}$ on $\mc V(\mf g_{\leq\frac12})$
by conformal derivations
(i.e. $a\,^\rho_\lambda\,(gh)=(a\,^\rho_\lambda\,g)h+(a\,^\rho_\lambda\,h)g$).
\item
$\mc W\subset\mc V(\mf g_{\leq\frac12})$ is a differential subalgebra.
\item
We have
$\rho\{g_\lambda\rho(h)\}_z=\rho\{g_\lambda h\}_z$,
and $\rho\{\rho(h)_\lambda g\}_z=\rho\{h_\lambda g\}_z$
for every $g\in\mc W$ and $h\in\mc V(\mf g)$.
\item
For every $g,h\in\mc W$, we have $\rho\{g_\lambda h\}_z\in\mc W[\lambda]$.
\item
The $\lambda$-bracket 
$\{\cdot\,_\lambda\,\cdot\}_{z,\rho}:\,\mc W\otimes\mc W\to\mc W[\lambda]$
given by \eqref{20120511:eq3} defines a PVA structure on $\mc W$.
\end{enumerate}
\end{theorem}

We can find an explicit formula for the $\lambda$-bracket in $\mc W$ as follows.
Recalling the Master Formula \eqref{masterformula} and using \eqref{rho} 
and the definition \eqref{lambda} of the $\lambda$-bracket in $\mc V(\mf g)$, we get
($g,h\in\mc W$):
$$
\{g_{\lambda}h\}_{z,\rho}
=\{g_{\lambda}h\}_{H,\rho}-z\{g_{\lambda}h\}_{K,\rho}\,,
$$
where
\begin{equation}\label{wbrackX}
\{g_{\lambda}h\}_{X,\rho}
=\sum_{\substack{i,j\in J_{\leq\frac12}\\m,n\in\mb Z_+}}
\frac{\partial h}{\partial q_j^{(n)}}(\partial+\lambda)^n
X_{ji}(\partial+\lambda)
(-\lambda-\partial)^m\frac{\partial g}{\partial q_i^{(m)}}\,,
\end{equation}
for $X$ one of the two matrices 
$H,K\in\Mat_{k\times k}\mc V(\mf g_{\leq\frac12})[\lambda]$ ($k=\dim(\mf g_{\leq\frac12})$),
given by
\begin{equation}\label{HKds}
H_{ji}(\lambda)
=\pi_{\leq\frac12}[q_i,q_j]+(q_i| q_j)\lambda+(f|[q_i,q_j]),
\qquad
K_{ji}(\partial)=(s|[q_j,q_i])
\,,
\end{equation}
for $i,j\in J_{\leq\frac12}$.

\subsection{Structure of the classical \texorpdfstring{$\mc W$}{W}-algebra}

In the algebra of differential polynomials $\mc V(\mf g_{\leq\frac12})$
we introduce the grading by \emph{conformal weight},
denoted by $\Delta$, defined as follows.
For a monomial $g=a_1^{(m_1)}\dots a_s^{(m_s)}$,
product of derivatives of elements $a_i\in\mf g_{1-\Delta_i}\subset\mf g_{\leq\frac12}$ 
(i.e. $\Delta_i\geq\frac12$),
we define its conformal weight as
\begin{equation}\label{degcw}
\Delta(g)=\Delta_1+\dots+\Delta_s+m_1+\dots+m_s\,.
\end{equation}
Thus we get the conformal weight space decomposition
$$
\mc V(\mf g_{\leq\frac12})=\bigoplus_{i\in\frac12\mb Z_+}\mc V(\mf g_{\leq\frac12})\{i\}\,.
$$
For example $\mc V(\mf g_{\leq\frac12})\{0\}=\mb F$,
$\mc V(\mf g_{\leq\frac12})\{\frac12\}=\mf g_{\frac12}$,
and $\mc V(\mf g_{\leq\frac12})\{1\}=\mf g_{0}\oplus S^2\mf g_{\frac12}$.

\begin{theorem}[\cite{DSKV12}]\label{daniele2}
Consider the PVA $\mc W$ with $\lambda$-bracket $\{\cdot\,_\lambda\,\cdot\}_{z,\rho}$
defined by equation \eqref{20120511:eq3}.
\begin{enumerate}[(a)]
\item
For every element $v\in\mf g^f_{1-\Delta}$ there exists a (not necessarily unique)
element $w\in\mc W\{\Delta\}=\mc W\cap\mc V(\mf g_{\leq\frac12})\{\Delta\}$ 
of the form $w=v+g$, where 
$$
g=\sum b_1^{(m_1)}\dots b_s^{(m_s)}\in\mc V(\mf g_{\leq\frac12})\{\Delta\}\,,
$$
is a sum of products of $\ad x$-eigenvectors 
$b_i\in\mf g_{1-\Delta_i}\subset\mf g_{\leq\frac12}$,
such that
$$
\Delta_1+\dots+\Delta_s+m_1+\dots+m_s=\Delta\,,
$$
and $s+m_1+\dots+m_s>1$.
\item
Let $\{w_j=v_j+g_j\}_{j\in J_f}$ be any collection of elements in $\mc W$ as in part (a),
where $\{v_j\}_{j\in J_f}$ is a basis of $\mf g^f\subset\mf g_{\leq\frac12}$ 
consisting of $\ad x$-eigenvectors.
Then the differential subalgebra $\mc W\subset\mc V(\mf g_{\leq\frac12})$ 
is the algebra of differential polynomials in the variables $\{w_j\}_{j\in J_f}$.
The algebra $\mc W$ is graded by the conformal weights defined in \eqref{degcw}:
$\mc W=\mb F\oplus\mc W\{1\}\oplus\mc W\{\frac32\}\oplus\mc W\{2\}\oplus\dots$.
\item
Let $\Delta=1$ or $\frac32$.
Then, for $v\in\mf g^f_{1-\Delta}$, the corresponding element $w=v+g\in\mc W\{\Delta\}$
as in (a) is uniquely determined by $v$.
Moreover, 
the space $\mc W\{\Delta\}$ coincides with the space of all such elements
$\{w=v+g\,|\,v\in\mf g^f_{1-\Delta}\}$.
In other words, there are bijective linear maps
$\varphi:\,\mf g^f_0\stackrel{\sim}{\longrightarrow}\mc W\{1\}$,
such that $\varphi(v)= w =v+g$,
and $\psi:\,\mf g^f_{-\frac12}\stackrel{\sim}{\longrightarrow}\mc W\{\frac32\}$,
such that $\psi(v)=w=v+g$.
\item
We have a Virasoro element $L\in\mc W\{2\}$ given by
\begin{equation}\label{virL}
L=
f+x'+\frac12\sum_{j\in J_0}a_ja^j
+\sum_{k\in J_{\frac12}}v^k[f,v_k]
+\frac12\sum_{k\in J_\frac12}v^k\partial v_k
\,.
\end{equation}
The $\lambda$-bracket of $L$ with itself is
\begin{equation}\label{virasoro}
\{L_\lambda L\}_{z,\rho}
=(\partial+2\lambda)L-(x| x)\lambda^3+2(f| s)z\lambda\,.
\end{equation}
Furthermore, for $z=0$, $L$ is an energy-momentum element
and the conformal weight defined by \eqref{degcw} coincides 
with the $L$-conformal weight,
namely, for $w\in\mc W\{\Delta\}$ we have 
$\{L_\lambda w\}_{0,\rho}=(\partial+\Delta\lambda)w+O(\lambda^2)$.
\end{enumerate}
\end{theorem}
\begin{remark}\label{20120727:rem}
In \cite{DSKV12} the Virasoro element $L\in\mc W$ is constructed, 
via the so called Sugawara construction, as
$$
L=
\rho\Big(\frac12\sum_{j\in J}u^ju_j
+\frac12\sum_{j\in J_{\frac12}}v^j\partial v_j
+x^{\prime}
\Big)\,\in\mc W\,,
$$
where $\rho$ is the map \eqref{rho}.
The equality of this expression and the one in \eqref{virL} 
follows immediately from the definition of the map $\rho$.
\end{remark}

Later we will use the following result concerning the conformal weights defined by \eqref{degcw}.
\begin{lemma}\label{20130320:lem}
If $g\in\mc W\{\Delta_1\}$ and $h\in\mc W\{\Delta_2\}$,
we have
$$
\{g_\lambda h\}_{z,\rho}
=\sum_{n\in\mb Z_+} g_{(n,H)}h\, \lambda^n
-z\sum_{n\in\mb Z_+} g_{(n,K)}h\, \lambda^n\,,
$$
where
$$
g_{(n,H)}h\,\in\mc W\{\Delta_1+\Delta_2-n-1\}
\,\,\text{ and }\,\,
g_{(n,K)}h\,\in\mc W\{\Delta_1+\Delta_2-n-2-d\}
\,,
$$
are independent of $z$.
(Recall that $d$ denotes the depth of the grading \eqref{dec}.)
\end{lemma}
\begin{proof}
It immediately follows from equations \eqref{wbrackX} and \eqref{HKds}
defining the $\lambda$-bracket of elements in $\mc W$.
\end{proof}

\section{Classical \texorpdfstring{$\mc W$}{W}-algebras for minimal nilpotent elements}
\label{sec:3}

%

\subsection{Setup and preliminary computations}\label{sec:3.1}

Let $f\in\mf g$ be a minimal nilpotent element,
that is a lowest root vector of $\mf g$.
In this case, the $\ad x$-eigenspace decomposition \eqref{dec} is
$$
\mf g=\mb Ff\oplus\mf g_{-\frac12}\oplus\mf g_{0}\oplus\mf g_{\frac12}\oplus\mb Fe\,,
$$
%
\begin{remark}\label{20130417:rem1}
It is well known that $\dim\mf g_{\pm\frac12}=2h\,\check{}-4$,
where $h\,\check{}$ is the dual Coxeter number of $\mf g$ (which equals $\frac12$ of the eigenvalue
in the adjoint representation of the Casimir operator associated to the Killing form $\kappa$).
Also, $\kappa(x|x)=h\,\check{}$.
\end{remark}
Note that $(x|a)=0$ for all $a\in\mf g_0^f$.
Hence, the subalgebra $\mf g_0\subset\mf g$ admits the orthogonal decomposition
$\mf g_0=\mf g_0^f\oplus\mb Fx$.
For $a\in\mf g_0$ we denote, as in \eqref{20130402:eq1}, 
by $a^\sharp$ its component in $\mf g_0^f$.
In other words, an element $a\in\mf g_0=\mf g_0^f\oplus\mb Fx$ decomposes as
\begin{equation}\label{20130201:eq3}
a=a^\sharp+\frac{(a|x)}{(x|x)}x\,.
\end{equation}
Moreover, since $\mf g_1=\mb Fe$ and $\mf g_{-1}=\mb Ff$, we have, using \eqref{20130201:eq2},
\begin{equation}\label{20130315:eq5}
[u,u_1]=\frac{\omega_-(u,u_1)}{2(x|x)}f
\,\,,\,\,\,\,
\text{ for all } u,u_1\in\mf g_{-\frac12}
\,,
\end{equation}
and 
\begin{equation}\label{20130315:eq6}
[v,v_1]=\frac{\omega_+(v,v_1)}{2(x|x)}e
\,\,,\,\,\,\,
\text{ for all } v,v_1\in\mf g_{\frac12}
\,,
\end{equation}
where $\omega_{\pm}$ are as in \eqref{20130201:eq4} and \eqref{20130201:eq5}.

Since $\mf g_1=\mb Fe$, we can choose, without loss of generality, $s=e$.
We can compute, using the definitions \eqref{lambda} and \eqref{rho}, 
all $\lambda$-brackets $\rho\{x_\lambda y\}_z$ for arbitrary $x,y\in\mf g$.
The results are given in Table \ref{table1}:

\begin{table}[H]
\caption{$\rho\{x_\lambda y\}_z$ for $x,y\in\mf g$} \label{table1} 
\begin{center}
\begin{tabular}{c||c|c|c|c|c} 
\phantom{$\Bigg($} 
$\rho\{\cdot\,_\lambda\,\cdot\}_z$
& $f\in\mf g_{-1}$ & $u_1\in\mf g_{-\frac12}$ & $b\in\mf g_{0}$ & $v_1\in\mf g_{\frac12}$ & $e\in\mf g_1$\\
\hline
\hline \phantom{$\Bigg($} 
$f\in\mf g_{-1}$ 
& 0 & 0 
& $\begin{array}{l}\frac{(x|b)}{(x|x)}f+ \\ 2z(x|b)\end{array}$ 
& $[f,v_1]$ & $\begin{array}{l} -2x+ \\ 2(x|x)\lambda \end{array}$ \\
\hline \phantom{$\Bigg($}
$u\in\mf g_{-\frac12}$ 
& 0 & $\begin{array}{l} \frac{\omega_-(u,u_1)}{2(x|x)}f+ \\ z\omega_-(u,u_1)\end{array}$ 
& $[u,b]$ & $\begin{array}{l} [u,v_1]+ \\ (u|v_1)\lambda\end{array}$ & $-[e,u]$ \\
\hline \phantom{$\Bigg($}
$a\in\mf g_{0}$ 
& $\begin{array}{l} -\frac{(x|a)}{(x|x)}f \\ -2z(x|a) \end{array}$ 
& $[a,u_1]$ & $\begin{array}{l} [a,b]+ \\ (a|b)\lambda \end{array}$ & $[a,v_1]$ & $2(x|a)$ \\
\hline \phantom{$\Bigg($}
$v\in\mf g_{\frac12}$ 
& $-[f,v]$ & $\begin{array}{l} [v,u_1]+ \\ (v|u_1)\lambda \end{array}$ & $[v,b]$ & $\omega_+(v,v_1)$ & 0 \\
\hline \phantom{$\Bigg($}
$e\in\mf g_1$ 
& $\begin{array}{l} 2x+ \\ 2(x|x)\lambda \end{array}$ & $[e,u_1]$ & $-2(x|b)$ & 0 & 0
\end{tabular}
\end{center}
\end{table}

\subsection{Generators of \texorpdfstring{$\mc W$}{W} for minimal nilpotent
\texorpdfstring{$f$}{f}}
\label{sec:3.2}

In this section we construct an explicit set of generators $\{w_j\}_{j\in J_f}$
for the $\mc W$-algebra.
\begin{theorem}\label{20120727:thm1}
The $\mc W$-algebra associated to a minimal nilpotent element $f\in\mf g$
is the algebra of differential polynomials 
with the following generators:
the energy-momentum element $L$ defined by \eqref{virL},
and elements of conformal weight $1$ and $\frac32$,
given by the following bijective maps:
\begin{equation}\label{phi}
\varphi:\,\mf g_{0}^f\to\mc W\{1\},
\quad
\varphi(a)=a+\frac12\sum_{k\in J_{\frac12}}[a,v_k]v^k
\,,
\end{equation}
and
\begin{equation}\label{psi}
\psi:\,\mf g_{-\frac12}\to\mc W\{\frac32\},
\quad
\psi(u)=
u+\frac13\sum_{h,k\in J_{\frac12}}[[u,v_h],v_k]v^hv^k+\sum_{k\in J_{\frac12}}[u,v_k]v^k+\partial[e,u]\,.
\end{equation}
\end{theorem}
\begin{proof}
By Theorem \ref{daniele2}, we only need to prove that the images of the maps $\varphi$ and $\psi$
lie in $\mc W$.
In other words, recalling the definition \eqref{20120511:eq2} of $\mc W$, we need to prove
that 
\begin{enumerate}[(i)]
\item
$\rho\{e_\lambda \varphi(a)\}_z=0$ for every $a\in\mf g^f_0$,
\item
$\rho\{v_\lambda \varphi(a)\}_z=0$ for every $v\in\mf g_{\frac12}$ and $a\in\mf g^f_0$,
\item
$\rho\{e_\lambda \psi(u)\}_z=0$ for every $u\in\mf g_{-\frac12}$,
\item
$\rho\{v_\lambda \psi(u)\}_z=0$ for every $u\in\mf g_{-\frac12}$ and $v\in\mf g_{\frac12}$.
\end{enumerate}
We immediately get from equation \eqref{phi}, the left Leibniz rule, and Table \ref{table1}, that
$$
\rho\{e_\lambda \varphi(a)\}_z
=
\rho\{e_\lambda a\}_z+\frac12\sum_{k\in J_{\frac12}} \rho\{e_\lambda[a,v_k]\}_zv^k
+\frac12\sum_{k\in J_{\frac12}} \rho\{e_\lambda v^k\}_z[a,v_k]
=
-2(x|a)\,,
$$
which is zero since $x$ is orthogonal to $\mf g_0^f$.
This proves (i).
Similarly, for (ii) we have
$$
\begin{array}{l}
\displaystyle{
\rho\{v_\lambda \varphi(a)\}_z
=\rho\{v_\lambda a\}_z+\frac12\sum_{k\in J_{\frac12}}\rho\{v_\lambda [a,v_k]\}_z v^k
+\frac12\sum_{k\in J_{\frac12}}\rho\{v_\lambda v^k\}_z [a,v_k]
} \\
\displaystyle{
=[v,a]+\frac12\sum_{k\in J_{\frac12}}\omega_+(v,[a,v_k])v^k
+\frac12\sum_{k\in J_{\frac12}}\omega_+(v,v^k)[a,v_k]
\,,} 
\end{array}
$$
and this is zero by Lemma \ref{20130315:lem1}(b) and the completeness relation \eqref{completeness}.
Similarly, by the definition \eqref{psi} of $\psi(u)$, the Leibniz rule for the $\lambda$-bracket,
and Table \ref{table1}, we have
$$
\begin{array}{l}
\displaystyle{
\rho\{e_\lambda \psi(u)\}_z
} \\
\displaystyle{
=
\rho\{e_\lambda u\}_z
+\frac13\sum_{h,k\in J_{\frac12}} \Big(
\rho\{e_\lambda[[u,v_h],v_k]\}_z v^hv^k
+\rho\{e_\lambda v^h\}_z [[u,v_h],v_k] v^k
} \\
\displaystyle{
+\rho\{e_\lambda v^k\}_z [[u,v_h],v_k] v^h
\Big)
+\sum_{k\in J_{\frac12}}\Big(
\rho\{e_\lambda [u,v_k]\}_z v^k
+\rho\{e_\lambda v^k\}_z [u,v_k]
\Big)
} \\
\displaystyle{
+(\partial+\lambda)\rho\{e_\lambda[e,u]\}_z
=
[e,u]-2\sum_{k\in J_{\frac12}}(x|[u,v_k])v^k
=
[e,u]+\sum_{k\in J_{\frac12}}(u|v_k)v^k
\,.}
\end{array}
$$
This is zero by Lemma \ref{20130315:lem2}(c), proving (iii).
Finally, for part (iv) we have, using Table \ref{table1},
\begin{equation}\label{20120729:eq1}
\begin{array}{l}
\displaystyle{
\rho\{v_\lambda \psi(u)\}_z
=
[v,u]+(v|u)\lambda
+\frac13\sum_{h,k\in J_{\frac12}} \Big(
\omega_+(v,[[u,v_h],v_k]) v^hv^k
} \\
\displaystyle{
+\omega_+(v,v^h) [[u,v_h],v_k] v^k
+\omega_+(v,v^k) [[u,v_h],v_k] v^h
\Big)
+\sum_{k\in J_{\frac12}}
[v,[u,v_k]] v^k
} \\
\displaystyle{
+\sum_{k\in J_{\frac12}} \omega_+(v,v^k) [u,v_k]
+\omega_+(v,[e,u])\lambda
\,.} 
\end{array}
\end{equation}
First, it follows by the definition \eqref{20130201:eq5} of $\omega_+$
and Lemma \ref{20130315:lem2}(a) that $\omega_+(v,[e,u])=-(v|u)$,
so that the second term and the last term in the RHS of \eqref{20120729:eq1} cancel.
Moreover, the first term and the second last term in the RHS of \eqref{20120729:eq1} 
cancel thanks to the completeness relation \eqref{completeness}.
Furthermore, again using the completeness relation \eqref{completeness}, we have
\begin{equation}\label{20130315:eq1}
\sum_{h,k\in J_{\frac12}}\omega_+(v,v^k) [[u,v_h],v_k] v^h
=
\sum_{h\in J_{\frac12}}[[u,v_h],v] v^h
=
-\sum_{h\in J_{\frac12}}[v,[u,v_h]] v^h
\,,
\end{equation}
and
\begin{equation}\label{20130315:eq2}
\begin{array}{l}
\displaystyle{
\sum_{h,k\in J_{\frac12}}\omega_+(v,v^h) [[u,v_h],v_k] v^k
=
\sum_{k\in J_{\frac12}}[[u,v],v_k] v^k
} \\
\displaystyle{
=
\sum_{k\in J_{\frac12}}[u,[v,v_k]] v^k
-\sum_{k\in J_{\frac12}}[v,[u,v_k]] v^k
\,,
}
\end{array}
\end{equation}
while, by Lemma \ref{20130315:lem1}(b) and 
the completeness relation \eqref{completeness}
we have
\begin{equation}\label{20130315:eq3}
\begin{array}{l}
\displaystyle{
\sum_{h,k\in J_{\frac12}}
\omega_+(v,[[u,v_h],v_k]) v^hv^k
} \\
\displaystyle{
=\sum_{h,k\in J_{\frac12}}
\Big(\omega_+([v,[u,v_h]],v_k)+([f,[u,v_h]]|[v,v_k])\Big)v^hv^k
} \\
\displaystyle{
=-\sum_{h\in J_{\frac12}} [v,[u,v_h]] v^h
+\sum_{h,k\in J_{\frac12}}
([f,[u,v_h]]|[v,v_k])v^hv^k
\,.
}
\end{array}
\end{equation}
Combining equations \eqref{20120729:eq1}, \eqref{20130315:eq1}, 
\eqref{20130315:eq2} and \eqref{20130315:eq3}, we get
\begin{equation}\label{20130315:eq4}
\rho\{v_\lambda \psi(u)\}_z
=
\frac13\sum_{h,k\in J_{\frac12}}([f,[u,v_h]]|[v,v_k])v^hv^k
+\frac13\sum_{k\in J_{\frac12}}[u,[v,v_k]] v^k
\,.
\end{equation}
By \eqref{20130315:eq6} we have $[v,v_k]=\frac{\omega_+(v,v_k)}{2(x|x)}e$.
Moreover, by the decomposition \eqref{20130201:eq3}, we have
$$
[f,[u,v_h]]=\frac{(x|[u,v_h])}{(x|x)}[f,x]=-\frac{(u|v_h)}{2(x|x)}f\,.
$$
Therefore, equation \eqref{20130315:eq4} becomes,
by the completeness relation \eqref{completeness},
$$
\rho\{v_\lambda \psi(u)\}_z
=
\frac{1}{6(x|x)}\sum_{h\in J_{\frac12}}
(u|v_h)v^hv
-\frac{1}{6(x|x)}
[u,e] v
\,,
$$
which is zero by the completeness relation \eqref{completeness2}.
\end{proof}

Let, as before, $\mc V(\mf g_{\leq\frac12})=S(\mb F[\partial]\mf g_{\leq\frac12})$ 
be the algebra of differential polynomials over $\mf g_{\leq\frac12}$,
and let $\mc V(\mf g^f)$ be the algebra of differential polynomials over $\mf g^f$.
We extend the quotient map $\pi:\,\mf g_{\leq\frac12}\to\mf g^f$ (defined by \eqref{20130402:eq1})
to a differential algebra homomorphism $\pi:\,\mc V(\mf g_{\leq\frac12})\to\mc V(\mf g^f)$.
\begin{corollary}\label{20130402:cor1}
The quotient map 
$\pi:\,\mc V(\mf g_{\leq\frac12})\twoheadrightarrow\mc V(\mf g^f)$
restricts to a differential algebra isomorphism
$\pi:\,\mc W\stackrel{\sim}{\longrightarrow}\mc V(\mf g^f)$,
and the inverse map $\pi^{-1}:\,\mc V(\mf g^f)\stackrel{\sim}{\longrightarrow}\mc W$
is defined on generators by 
$$
\begin{array}{l}
\displaystyle{
\vphantom{\Big(}
\pi^{-1}(a)=\varphi(a)
\,\,\,\, \text{ for } a\in\mf g_0^f
\,,} \\
\displaystyle{
\vphantom{\Big(}
\pi^{-1}(u)=\psi(u)
\,\,\,\, \text{ for } u\in\mf g_{-\frac12}
\,,} \\
\displaystyle{
\pi^{-1}(f)=L-\frac12\sum_{i\in J_0^f}\varphi(a_i)\varphi(a^i)=:\tilde{L}\,,
}
\end{array}
$$
where $\{a_i\}_{i\in J_0^f}$ and $\{a^i\}_{i\in J_0^f}$ are dual bases of $\mf g_0^f$
with respect to $(\cdot\,|\,\cdot)$.
\end{corollary}
\begin{proof}
By equations \eqref{phi} and \eqref{psi}
we have, respectively, that $\pi(\varphi(a))=a$ for every $a\in\mf g_0^f$
and $\pi(\psi(u))=u$ for every $u\in\mf g_{-\frac12}$,
since $\mf g_{\frac12}\subset\ker(\pi)$.
Moreover, by equation \eqref{virL} we have
$\pi(L)=f+\frac12\sum_{i\in J_0^f}a_ia^i$,
so that $L=\pi^{-1}(f)+\frac12\sum_{i\in J_0^f}\varphi(a_i)\varphi(a^i)$.
The statement follows.
\end{proof}

\subsection{\texorpdfstring{$\lambda$}{lambda}-brackets in \texorpdfstring{$\mc W$}{W}
for minimal nilpotent \texorpdfstring{$f$}{f}}

\begin{theorem}\label{prop:minimal}
The multiplication table for $\mc W$ in the case of a minimal nilpotent $f$ is 
given by Table \ref{table2}
(where $L$ is as in \eqref{virL}, $\varphi(a), \varphi(b)$ are as in \eqref{phi} for $a,b\in\mf g_0^f$,
and $\psi(u), \psi(u_1)$ are as in \eqref{psi} for $u,u_1\in\mf g_{-\frac12}$):

\begin {table}[H]
\caption{$\lambda$-brackets among generators of $\mc W$ for minimal nilpotent $f$} \label{table2} 
\begin{center}
\begin{tabular}{c||c|c|c}
\phantom{$\Bigg($} 
$\{\cdot\,_\lambda\,\cdot\}_{z,\rho}$
& $L$ & $\varphi(b)$ & $\psi(u_1)$ \\
\hline
\hline \phantom{$\Bigg($} 
$L$ & $\begin{array}{l} (\partial+2\lambda)L \\ -(x| x)\lambda^3+4(x|x)z\lambda \end{array}$
& $(\partial+\lambda)\varphi(b)$ & $\big(\partial+\frac32\lambda\big)\psi(u_1)$ \\
\hline \phantom{$\Bigg($}
$\varphi(a)$ & $\lambda\varphi(a)$ & $\varphi([a,b])+(a|b)\lambda$ & $\psi([a,u_1])$ \\
\hline \phantom{$\Bigg($}
$\psi(u)$ & $\big(\frac12\partial+\frac32\lambda\big)\psi(u)$ & $\psi([u,b])$ & eq.\eqref{20130320:eq2} 
\end{tabular}
\end{center}
\end{table}

The $\lambda$-bracket of $\psi(u)$ and $\psi(u_1)$ is
\begin{equation}\label{20130320:eq2}
\begin{array}{c}
\displaystyle{
\{\psi(u)_\lambda\psi(u_1)\}_{z,\rho}=
\sum_{k\in J_{\frac12}}\phi([u,v^k]^\sharp)\phi([u_1,v_k]^\sharp)
+\big(\partial+2\lambda\big)\phi([u,[e,u_1]]^\sharp)
} \\
\displaystyle{
+\frac{\omega_-(u,u_1)}{2(x|x)}\tilde{L}
-\lambda^2\omega_-(u,u_1)+z\omega_-(u,u_1)\,,
}
\end{array}
\end{equation}
where, for $a\in\mf g_0$, $a^\sharp$ was defined in \eqref{20130201:eq3}.

In particular, $L$ is an energy momentum element for all $z\in\mb F$,
$\varphi(a)$ is a primary element of conformal weight $1$ for every $a\in\mf g_0^f$,
and $\psi(u)$ is a primary element of conformal weight $\frac32$ for every $u\in\mf g_{-\frac12}$.
\end{theorem}
\begin{proof}
The $\lambda$-bracket $\{L_\lambda L\}_{z,\rho}$ was given by \eqref{virasoro}.
Next, let us prove the formula for the $\lambda$-bracket $\{L_\lambda\varphi(a)\}_{z,\rho}$,
$a\in\mf g_0^f$.
By Theorem \ref{daniele2}(d),
we already know that, for $z=0$, $\varphi(a)$ is an $L$-eigenvector of eigenvalue $1$:
$\{L_\lambda\varphi(a)\}_{z=0,\rho}=(\partial+\lambda)\varphi(a)+O(\lambda^2)$.
Moreover, since $\varphi(a)\in\mc W\{1\}$, $L\in\mc W\{2\}$, and $s=e\in\mf g_1$, 
by Lemma \ref{20130320:lem} we have
$$
\{L_\lambda\varphi(a)\}_{z,\rho}=(\partial+\lambda)\varphi(a)
+L_{(2,H)}\varphi(a)\lambda^2-zL_{(0,K)}\varphi(a)
\,,
$$
where $L_{(2,H)}\varphi(a),\,L_{(0,K)}\varphi(a)\in\mc W\{0\}=\mb F$.
We need to prove that $L_{(2,H)}\varphi(a)=L_{(0,K)}\varphi(a)=0$.
Note that the expression \eqref{phi} of $\varphi(a)$ involves only $a\in\mf g_0^f$
and elements of $\mf g_{\frac12}$.
But it is clear from Table \ref{table1}
that $\lambda$-brackets of elements in $\mf g_0^f$ or of elements in $\mf g_{\frac12}$
with any other element do not involve any $z$.
Therefore, we automatically have that $L_{(0,K)}\varphi(a)=0$.
Furthermore,
we see from Table \ref{table1} that the highest power of $\lambda$
which appears in the $\lambda$-bracket $\rho\{\cdot\,_\lambda\,\cdot\}_z$
of two elements in $\mf g$ is $1$.
Therefore, by the Leibniz rule and sesquilinearity, 
we can get contributions to the power $\lambda^2$
only from terms involving derivatives in the expressions \eqref{virL} of $L$.
Such terms are
$$
-\lambda\rho\{x_\lambda\varphi(a)\}_z
-\frac12\sum_{k\in J_{\frac12}}\rho{\{{v_k}_{\partial+\lambda} \varphi(a)\}_z}_{\to}
(\partial+\lambda)v^k
\,.
$$
On the other hand,
recalling the expression \eqref{phi} of $\varphi(a)$,
it is immediate to check, using Table \ref{table1},
that $\rho\{x_\lambda \varphi(a)\}_z$ and $\rho\{{v_k}_\lambda\varphi(a)\}_z$
are independent of $\lambda$.
Therefore $\lambda^2$ never appears, proving that $L_{(2,H)}\varphi(a)=0$.

We use a similar argument for the formula of the $\lambda$-bracket $\{L_\lambda\psi(u)\}_{z,\rho}$,
for $u\in\mf g_{-\frac12}$.
By Theorem \ref{daniele2}(d),
we already know that, for $z=0$, $\psi(u)$ is an $L$-eigenvector of eigenvalue $\frac32$:
$\{L_\lambda\psi(u)\}_{z=0,\rho}=\big(\partial+\frac32\lambda\big)\psi(u)+O(\lambda^2)$.
Moreover, since $\psi(u)\in\mc W\{\frac32\}$ and $L\in\mc W\{2\}$,
we get by Lemma \ref{20130320:lem} that
the coefficients of $\lambda^2$ and $z$ in $\{L_\lambda\psi(u)\}_{z,\rho}$
lie in $\mc W\{\frac12\}=0$.
Therefore, 
$\{L_\lambda\psi(u)\}_{z,\rho}=\big(\partial+\frac32\lambda\big)\psi(u)$.

Next, let us prove the formula for $\{\varphi(a)_\lambda\varphi(b)\}_{z,\rho}$, for $a,b\in\mf g_0^f$.
It turns out that it is much more convenient to compute, instead,
$\pi\{\varphi(a)_\lambda\varphi(b)\}_{z,\rho}$,
and then apply the inverse map $\pi^{-1}$, using Corollary \ref{20130402:cor1}.
Since $\mf g_{\frac12}\subset\ker(\pi)$, and since $\pi$ is a differential algebra homomorphism,
by the Leibniz rule all quadratic terms in $\mf g_{\frac12}$
in the expression \eqref{phi} of $\varphi(a)$ and $\varphi(b)$
give zero contribution in the computation of $\pi\{\varphi(a)_\lambda\varphi(b)\}_{z,\rho}$.
Therefore, we have
$$
\pi\{\varphi(a)_\lambda\varphi(b)\}_{z,\rho}
=\pi\rho\{a_\lambda b\}_z
=\pi([a,b]+(a|b)\lambda)=[a,b]+(a|b)\lambda\,.
$$
In the second identity we used Table \ref{table1},
while in the last identity we used the definition \eqref{20130402:eq1} of $\pi$ and the fact that 
$[\mf g^f,\mf g^f]\subset\mf g^f$.
Applying $\pi^{-1}$ to the above equation, we get,
according to Corollary \ref{20130402:cor1},
$\{\varphi(a)_\lambda\varphi(b)\}_{z,\rho}=\varphi([a,b])+(a|b)\lambda$,
as stated in Table \ref{table2}.

We use a similar argument to prove the formula for $\{\varphi(a)_\lambda\psi(u)\}_{z,\rho}$, 
for $a\in\mf g_0^f$ and $u\in\mf g_{-\frac12}$.
Again, we start by computing the projection
$\pi\{\varphi(a)_\lambda\psi(u)\}_{z,\rho}$.
Since $\mf g_{\frac12}\subset\ker(\pi)$, by the Leibniz rule we have
$$
\pi\{\varphi(a)_\lambda\psi(u)\}_{z,\rho}
=\pi\Big(
\rho\{a_\lambda u\}_z
+\sum_{k\in J_{\frac12}}\rho\{a_\lambda v^k\}_z[u,v_k]
+(\partial+\lambda)\rho\{a_\lambda [e,u]\}_z\Big)\,.
$$
By Table \ref{table1}, we have
$\rho\{a_\lambda u\}_z=[a,u]\in\mf g_{-\frac12}\subset\mf g^f$,
$\rho\{a_\lambda v^k\}_z=[a,v^k]\in\mf g_{\frac12}\subset\ker(\pi)$,
and $\rho\{a_\lambda [e,u]\}_z=[a,[e,u]]\in\mf g_{\frac12}\subset\ker(\pi)$.
Therefore, 
$\pi\{\varphi(a)_\lambda\psi(u)\}_{z,\rho}=[a,u]$,
and, by Corollary \ref{20130402:cor1},
$\{\varphi(a)_\lambda\psi(u)\}_{z,\rho}=\psi([a,u])$,
as stated in Table \ref{table2}.

We are left to prove equation \eqref{20130320:eq2}.
As before, we start by computing the projection 
$\pi\{\psi(u)_\lambda\psi(u_1)\}_{z,\rho}$.
By the definition \eqref{psi} of $\psi(u)$ and $\psi(u_1)$,
and by the Leibniz rule, we have,
using the fact that $\mf g_{\frac12}\subset\ker(\pi)$,
\begin{equation}\label{20130402:eq2}
\begin{array}{l}
\displaystyle{
\vphantom{\Big(}
\pi\{\psi(u)_\lambda\psi(u_1)\}_{z,\rho}
=
\pi\bigg(
\rho\{u _\lambda u_1\}_z
+\sum_{h,k\in J_{\frac12}}[u_1,v_k]^\sharp\rho{\{{v^h}{} _{\partial+\lambda} v^k\}_z}_{\to} [u,v_h]^\sharp
} \\
\displaystyle{
+\sum_{k\in J_{\frac12}}[u_1,v_k]^\sharp\rho\{u _\lambda v^k\}_z
+\sum_{h\in J_{\frac12}}\rho{\{{v^h}{} _{\partial+\lambda} u_1\}_z}_{\to} [u,v_h]^\sharp
} \\
\displaystyle{
+\sum_{h\in J_{\frac12}}(\partial+\lambda)\rho{\{{v^h}{} _{\partial+\lambda} [e,u_1]\}_z}_{\to} [u,v_h]^\sharp
-\sum_{k\in J_{\frac12}}\lambda[u_1,v_k]^\sharp\rho\{[e,u] _\lambda v^k\}_z
} \\
\displaystyle{
+(\partial+\lambda)\rho\{u _\lambda [e,u_1]\}_z
-\lambda\rho\{[e,u] _\lambda u_1\}_z
-\lambda(\partial+\lambda)\rho\{[e,u] _\lambda [e,u_1]\}_z
\bigg)
}
\end{array}
\end{equation}
Using Table \ref{table1} and the completeness relations \eqref{completeness}
and \eqref{completeness2}, equation \eqref{20130402:eq2}
gives the following
\begin{equation}\label{20130402:eq3}
\begin{array}{l}
\displaystyle{
\vphantom{\Big(}
\pi\{\psi(u)_\lambda\psi(u_1)\}_{z,\rho}
=
\pi\bigg(
\frac{\omega_-(u,u_1)}{2(x|x)}f+2z\omega_-(u,u_1)
-\sum_{k\in J_{\frac12}}[u_1,v_k]^\sharp[u,v^k]^\sharp
} \\
\displaystyle{
+\sum_{k\in J_{\frac12}}[u_1,v_k]^\sharp[u,v^k]
+\sum_{h\in J_{\frac12}}[v^h,u_1][u,v_h]^\sharp
+(\partial+\lambda)[u,[e,u_1]]
-[[e,u],u_1]\lambda
} \\
\displaystyle{
+(u|[e,u_1])\lambda^2
-([e,u]|u_1)\lambda^2
-\omega_+([e,u],[e,u_1])\lambda^2
\bigg)
\,.}
\end{array}
\end{equation}
Recall that $\pi(f)=f$, since $f\in\mf g^f$.
By skewsymmetry of $\omega_+$ and the definition of $\pi$, we have
$$
\pi\Big(\sum_{k\in J_{\frac12}}[u_1,v_k]^\sharp[u,v^k]\Big)
=\pi\Big(\sum_{h\in J_{\frac12}}[v^h,u_1][u,v_h]^\sharp\Big)
=\sum_{k\in J_{\frac12}}[u_1,v_k]^\sharp[u,v^k]^\sharp\,.
$$
Furthermore, we have, 
$$
\pi([u,[e,u_1]])=-\pi([[e,u],u_1])=[u,[e,u_1]]^\sharp
\,,
$$
and, by the definition \eqref{20130201:eq5} of $\omega_+$ and Lemma \ref{20130315:lem2}(a),
we also have
$$
-\omega_+([e,u],[e,u_1])=-(u|[e,u_1])=([e,u]|u_1)=\omega_-(u,u_1)\,.
$$
Therefore, equation \eqref{20130402:eq3} gives
\begin{equation}\label{20130402:eq4}
\begin{array}{l}
\displaystyle{
\vphantom{\Big(}
\pi\{\psi(u)_\lambda\psi(u_1)\}_{z,\rho}
=
\frac{\omega_-(u,u_1)}{2(x|x)}f+2z\omega_-(u,u_1)
+\sum_{k\in J_{\frac12}}[u_1,v_k]^\sharp[u,v^k]^\sharp
} \\
\displaystyle{
+(\partial+2\lambda)[u,[e,u_1]]^\sharp
-\omega_-(u,u_1)\lambda^2
\,.}
\end{array}
\end{equation}
Applying the bijective map $\pi^{-1}:\,\mc V(\mf g^f)\to\mc W$ to both sides of equation 
\eqref{20130402:eq4} and using Corollary \ref{20130402:cor1}, we get
equation \eqref{20130320:eq2}.
\end{proof}
\begin{remark}\label{20130321:rem}
The $\lambda$-brackets for the $\mc W$-algebra associated to a minimal nilpotent element
were computed via the cohomological quantum (resp. classical) Hamiltonian reduction
in \cite{KW04} (resp. \cite{Suh13}).
\end{remark}

\section{Classical \texorpdfstring{$\mc W$}{W}-algebras for short nilpotent elements}
\label{sec:3.5}

\subsection{Setup and preliminary computations}

By definition, a nilpotent element $f\in\mf g$ is called \emph{short} if the $\ad x$-eigenvalues
are $-1,0,1$, namely the $\ad x$-eigenspace decomposition \eqref{dec} is
$$
\mf g=\mf g_{-1}\oplus\mf g_{0}\oplus\mf g_{1}\,.
$$

According to Theorem \ref{daniele2},
a set of generators for $\mc W$ is in bijective correspondence with a basis
of $\mf g^f=\mf g_{-1}\oplus\mf g^f_0$.
Note that, by representation theory of $\mf{sl}_2$, we have $\mf g_0^f=\mf g_0^e$,
$[f,\mf g_1]=[e,\mf g_{-1}]=(\mf g_0^f)^\perp$.
The subspace $\mf g_0^f$, being the centralizer of $\mf{sl}_2$,
is a reductive subalgebra of the simple Lie algebra $\mf g$,
hence the bilinear form $(\cdot\,|\,\cdot)$ restricts to a
non-degenerate symmetric invariant  
bilinear form on $\mf g_0^f$, and $[f,\mf g_1]$ is its orthocomplement.
Hence,
we have the direct sum decomposition
\begin{equation}\label{20130322:eq4}
\mf g_0=\mf g_0^f\oplus[f,\mf g_1]\,,
\end{equation}
and we denote by ${}^\sharp:\,\mf g_0\to\mf g_0^f$
and ${}^\perp:\,\mf g_0\to[f,\mf g_1]$ the corresponding orthogonal projections.
In fact, the decomposition \eqref{20130322:eq4}
is a $\mb Z/2\mb Z$-grading of the Lie algebra $\mf g_0$,
namely we have the following
\begin{lemma}\phantomsection\label{20130417:lem1}
\begin{enumerate}[(a)]
\item
$[\mf g_0^f,\mf g_0^f]\subset\mf g_0^f$,
\item
$[\mf g_0^f,[f,\mf g_1]]\subset[f,\mf g_1]$,
\item
$[[f,\mf g_1],[f,\mf g_1]]\subset\mf g_0^f$.
\end{enumerate}
\end{lemma}
\begin{proof}
Parts (a) and (b) are immediate, by invariance of the bilinear form and
by the Jacobi identity. For part (c) we have, for $v,v_1\in\mf g_1$,
$$
\begin{array}{l}
\displaystyle{
\vphantom{\Big(}
[e,[[f,v],[f,v_1]]]=[[e,[f,v]],[f,v_1]]+[[f,v],[e,[f,v_1]]]
} \\
\displaystyle{
\vphantom{\Big(}
=2[v,[f,v_1]]+2[[f,v],v_1]=0
\,.}
\end{array}
$$
Hence, $[[f,v],[f,v_1]]\in\mf g_0^e=\mf g_0^f$.
\end{proof}

Recall that we have a commutative Jordan product on $\mf g_{-1}$ given by
\eqref{20130320:eq1}.
%
We fix dual bases $\{a_i\}_{i\in J_0^f}$ and $\{a^i\}_{i\in J_0^f}$ of $\mf g_0^f$:
$(a_i|a^j)=\delta_{ij}$.
They are equivalently defined by the following completeness relation:
\begin{equation}\label{complete1}
\sum_{i\in J_0^f}(a|a^i)a_i=a^\sharp
\quad\text{ for all } a\in\mf g_0\,.
\end{equation}
Let also $\{u_k\}_{k\in J_1}$ be a basis of $\mf g_{-1}$,
and let $\{u^k\}_{k\in J_1}$ be the dual (with respect to $(\cdot\,|\,\cdot)$)
basis of $\mf g_1$.
Then, it is easy to check that
$\big\{-\frac12[e,u_k]\big\}_{k\in J_1}$, $\big\{[f,u^k]\big\}_{k\in J_1}$,
are dual bases of $[f,\mf g_1]\subset\mf g_0$.
In other words, we have the completeness relations:
\begin{equation}\label{complete2}
\begin{array}{l}
\displaystyle{
\sum_{k\in J_1}(u|u^k)u_k=u
\quad\text{ for all } u\in\mf g_{-1}
\,\,,\,\,\,\,
\sum_{k\in J_1}(v|u_k)u^k=v
\quad\text{ for all } v\in\mf g_{1}
\,,} \\
\displaystyle{
-\frac12\sum_{k\in J_1}([e,u_k]|a)[f,u^k]
=-\frac12\sum_{k\in J_1}([f,u^k]|a)[e,u_k]=a^\perp
\quad\text{ for all } a\in\mf g_0
\,.}
\end{array}
\end{equation}

Recall that, by assumption, $s\in\mf g_1$.
We can compute, using the definitions \eqref{lambda} and \eqref{rho}, 
all $\lambda$-brackets $\rho\{x_\lambda y\}_z$ for arbitrary $x,y\in\mf g$.
The results are given in Table \ref{table3}:

\begin {table}[H]
\caption{$\rho\{x_\lambda y\}_z$ for $x,y\in\mf g$} \label{table3} 
\begin{center}
\begin{tabular}{c||c|c|c} 
\phantom{$\Bigg($} 
$\rho\{\cdot\,_\lambda\,\cdot\}_z$
& $u_1\in\mf g_{-1}$ & $b\in\mf g_{0}$ & $v_1\in\mf g_1$ \\
\hline
\hline \phantom{$\Bigg($} 
$u\in\mf g_{-1}$ & 0 & $[u,b]+z(s|[u,b])$ & $[u,v_1]+(u|v_1)\lambda$ \\
\hline \phantom{$\Bigg($}
$a\in\mf g_0$ & $[a,u_1]+z(s|[a,u_1])$ & $[a,b]+(a|b)\lambda$ & $(f|[a,v_1])$ \\
\hline \phantom{$\Bigg($}
$v\in\mf g_1$ & $[v,u_1]+(v|u_1)\lambda$ & $(f|[v,b])$ & 0 
\end{tabular}
\end{center}
\end{table}

\subsection{Generators of \texorpdfstring{$\mc W$}{W} for short nilpotent
\texorpdfstring{$f$}{f}}

\begin{theorem}\label{20120727:thm1b}
Let $\mc W$ be the $\mc W$-algebra associated to a short nilpotent element $f\in\mf g$.
As a differential algebra, $\mc W$ is the algebra of differential polynomials 
with the following generators:
$a_i$, for $i\in J_0^f$, and $\psi(u_k)$, for $k\in J_1$,
where $\psi:\,\mf g_{-1}\to\mc W\{2\}$ is the following injective map
\begin{equation}\label{psiu}
\psi(u)
=u
-\frac12\sum_{k\in J_1}[u,u^k][e,u_k]
-\frac18\sum_{k\in J_1}[f,u^k][e,u\circ u_k]
+\frac12\partial[e,u]
\,,
\end{equation}
where we are using the notation \eqref{20130320:eq1}.
The subspace of elements of conformal weight $1$ is $\mc W\{1\}=\mf g_0^f$,
while the subspace of elements of conformal weight $2$ is 
$$
\mc W\{2\}=\psi(\mf g_{-1})\oplus\partial\mf g_0^f
\oplus S^2\mf g_0^f\,.
$$
\end{theorem}
\begin{proof}
Since $\mf g_{\frac12}=0$, it is clear from Theorem \ref{daniele2} that $\mc W\{1\}=\mf g_0^f$.
Hence, 
according to Theorem \ref{daniele2}, 
we only need to prove that all the elements $\psi(u),\,u\in\mf g_{-1}$, lie in $\mc W$.
In other words, recalling the definition \eqref{20120511:eq2} of $\mc W$, we need to prove
that, for every $u\in\mf g_{-1}$ and $v\in\mf g_1$, we have
$$
\rho\{v_\lambda \psi(u)\}_z=0\,.
$$
%
By Table \ref{table3} we have
\begin{equation}\label{20120730:eq1}
\begin{array}{l}
\displaystyle{
\rho\{v_\lambda \psi(u)\}_z
=
\rho\{v_\lambda u\}_z
+\frac12(\partial+\lambda)\rho\{v_\lambda [e,u]\}_z
-\frac12\sum_{k\in J_1}\rho\{v_\lambda [u,u^k]\}_z[e,u_k]
} \\
\displaystyle{
-\frac12\sum_{k\in J_1}\rho\{v_\lambda [e,u_k]\}_z[u,u^hk]
-\frac18\sum_{k\in J_1}
\rho\{v_\lambda [f,u^k]\}_z[e,u\circ u_k]
} \\
\displaystyle{
-\frac18\sum_{k\in J_1}
\rho\{v_\lambda [e,u\circ u_k]\}_z[f,u^k]
} \\
\displaystyle{
=
[v,u]+(v|u)\lambda
+\frac12\lambda (f|[v,[e,u]])
-\frac12\sum_{k\in J_1}(f|[v,[u,u^k]])[e,u_k]
} \\
\displaystyle{
-\frac12\sum_{k\in J_1}(f|[v,[e,u_k]])[u,u^k]
-\frac18\sum_{k\in J_1}
(f|[v,[f,u^k]])
[e,u\circ u_k]
} \\
\displaystyle{
-\frac18\sum_{k\in J_1}
(f|[v,[e,u\circ u_k]])
[f,u^k]
\,.}
\end{array}
\end{equation}
By invariance of $(\cdot\,|\,\cdot)$ and representation theory of $\mf{sl}_2$, we have
$(f|[v,[e,u]])=-2(v|u)$.
Hence, the linear terms in $\lambda$ in the RHS of \eqref{20120730:eq1} vanish.
Moreover, by invariance of $(\cdot\,|\,\cdot)$, by the Jacobi identity, 
and the completeness relations \eqref{complete2},
we can rewrite the RHS of \eqref{20120730:eq1} as
\begin{equation}\label{20130322:eq2}
\begin{array}{l}
\displaystyle{
[v,u]
-\frac12[e,[[f,v],u]]
-\frac12[u,[[f,v],e]]
} \\
\displaystyle{
-\frac18
[e,u\circ [[f,v],f]]
-\frac18
[f,[[[f,v],e],[e,u]]]
\,.}
\end{array}
\end{equation}
By the Jacobi identity, we have
$$
\begin{array}{l}
[e,[[f,v],u]]
=2[v,u]+[[f,v],[e,u]]
\,, \\
{[u,[[f,v],e]]}
=-2[u,v]
\,, \\
{[e,u\circ [[f,v],f]]}
=2[[e,u],[f,v]]+4[u,v]
\,,\\
{[f,[[[f,v],e],[e,u]]]}
=2[[e,u],[f,v]]-4[v,u]
\,.
\end{array}
$$
Hence, \eqref{20130322:eq2} is equal to $0$.
\end{proof}

As in the Section \ref{sec:3.2}, 
we denote by $\pi:\,\mc V(\mf g_{\leq\frac12})\to\mc V(\mf g^f)$
the differential algebra homomorphism induced by the quotient map 
$\mf g_{\leq\frac12}\to\mf g^f$ defined by \eqref{20130402:eq1}.
\begin{corollary}\label{20130402:cor2}
The quotient map 
$\pi:\,\mc V(\mf g_{\leq\frac12})\twoheadrightarrow\mc V(\mf g^f)$
restricts to a differential algebra isomorphism
$\pi:\,\mc W\stackrel{\sim}{\longrightarrow}\mc V(\mf g^f)$,
and the inverse map $\pi^{-1}:\,\mc V(\mf g^f)\stackrel{\sim}{\longrightarrow}\mc W$
is defined, on generators, by 
$$
\pi^{-1}(a)=a
\,\,\,\, \text{ for } a\in\mf g_0^f
\,\,,\,\,\,\,
\pi^{-1}(u)=\psi(u)
\,\,\,\, \text{ for } u\in\mf g_{-1}
\,.
$$
\end{corollary}
\begin{proof}
The same as for Corollary \ref{20130402:cor1}.
\end{proof}

\subsection{\texorpdfstring{$\lambda$}{lambda}-bracket in \texorpdfstring{$\mc W$}{W}
for short nilpotent \texorpdfstring{$f$}{f}}

\begin{theorem}\label{prop:short}
The multiplication table for $\mc W$ in the case of a short nilpotent element $f$ is 
given by Table \ref{table4} ($a,b\in\mf g_0^f$, $u,u_1\in\mf g_{-1}$):
%
\begin {table}[h]
\caption{$\lambda$-brackets among generators of $\mc W$ for short nilpotent $f$} \label{table4}
\begin{center}
\begin{tabular}{c||c|c}
\phantom{$\Bigg($} 
$\{\cdot\,_\lambda\,\cdot\}_{z,\rho}$
& $b$ & $\psi(u_1)$ \\
\hline
\hline \phantom{$\Bigg($} 
$a$ & $[a,b]+(a|b)\lambda$ & $\psi([a,u_1])+z(s|[a,u_1])$  \\
\hline \phantom{$\Bigg($}
$\psi(u)$ & $\psi([u,b])+z(s|[u,b])$ & eq.\eqref{20130320:eq2b}
\end{tabular}
\end{center}
\end{table}

where the $\lambda$-bracket of $\psi(u)$ and $\psi(u_1)$ is
\begin{equation}\label{20130320:eq2b}
\begin{array}{l}
\displaystyle{
\{\psi(u)_\lambda\psi(u_1)\}_{z,\rho}
=\frac12\sum_{k\in J_1}\psi(u\circ u_k)[u_1,u^k]^\sharp
-\frac12\sum_{k\in J_1}\psi(u_1\circ u_k)[u,u^k]^\sharp
}\\
\displaystyle{
+\frac14\sum_{h,k\in J_1}[[e,u_h],[e,u_k]]
[u,u^h]^\sharp[u_1,u^k]^\sharp
-\frac12(\partial+2\lambda)\psi(u\circ u_1)
}\\
\displaystyle{
+\frac14(\partial+2\lambda)\sum_{k\in J_1}[[e,u],[e,u_k]][u_1,u^k]^\sharp
+\frac14\sum_{k\in J_1} [[e,u_1],[e,u_k]] (\partial+\lambda) [u,u^k]^\sharp
}\\
\displaystyle{
-\frac14\left(3\lambda^2+3\lambda\partial+\partial^2\right)[[e,u],[e,u_1]]
+\frac14(e|u\circ u_1)\lambda^3
}\\
\displaystyle{
+\frac12z\left([[e,u],[s,u_1]]^\sharp-[[e,u_1],[s,u]]^\sharp\right)
-(s|u\circ u_1)z\lambda
\,.
}
\end{array}
\end{equation}
The Virasoro element \eqref{virL} can be expressed in terms of the generators of $\mc W$
as follows:
\begin{equation}\label{20130322:eq3}
L=\psi(f)+\frac12\sum_{i\in J_0^f}a_ia^i\,,
\end{equation}
and we have the following $\lambda$-brackets of $L$ with the generators of $\mc W$
($a\in\mf g_0^f$, $u\in\mf g_{-1}$):
\begin{equation}\label{20130322:eq9}
\begin{array}{l}
\displaystyle{
\vphantom{\Big)}
\{L_\lambda a\}_{z,\rho}=(\partial+\lambda)a
\,,}\\
\displaystyle{
\vphantom{\Big)}
\{L_\lambda \psi(u)\}_{z,\rho}=(\partial+2\lambda)\psi(u)-\frac12(e|u)\lambda^3+z(s|u)\lambda
\,.}
\end{array}
\end{equation}
\end{theorem}
\begin{remark}\label{20130403:rem}
Note that the formula \eqref{20130320:eq2b} for $\{\psi(u)_\lambda\psi(u_1)\}_{z,\rho}$
is indeed skewsymmetric w.r.t. exchanging $u$ with $u_1$ and $\lambda$ with $-\lambda-\partial$.
This follows form 
the following identity in $(\mf g_0^f)^{\otimes2}$ ($u,u_1\in\mf g_{-1}$):
\begin{equation}\label{20130403:eq6}
\sum_{k\in J_1} [[e,u],[e,u_k]] \otimes [u_1,u^k]^\sharp 
=
\sum_{k\in J_1} [u,u^k]^\sharp\otimes [[e,u_1],[e,u_k]] 
\,,
\end{equation}
which can be easily checked,
by taking inner product with an element $a\otimes b\in (\mf g_0^f)^{\otimes2}$.
\end{remark}
\begin{proof}[{Proof of Theorem \ref{prop:short}}]
The $\lambda$-bracket $\{a_\lambda b\}_{z,\rho}$ is given by Table \ref{table3}.
Next, we compute $\{a_\lambda\psi(u)\}_{z,\rho}$ for $a\in\mf g_0^f$ and $u\in\mf g_{-1}$.
As in the proof of Theorem \ref{prop:minimal},
we compute, instead,
$\pi\{a_\lambda\psi(u)\}_{z,\rho}$,
and then apply the inverse map $\pi^{-1}$, using Corollary \ref{20130402:cor2}.
Recalling that $[e,\mf g_{-1}]=[f,\mf g_1]\subset\ker(\pi)$ we have, 
by the formula \eqref{psiu} for $\psi(u)$ and the Leibniz rule,
\begin{equation}\label{20130403:eq1}
\begin{array}{c}
\displaystyle{
\pi\{a_\lambda\psi(u)\}_{z,\rho}
=
\pi\big(\rho\{a_\lambda u\}_z\big)
-\frac12\sum_{k\in J_1}[u,u^k]^\sharp\pi\big(\rho\{a_\lambda [e,u_k]\}_z\big)
} \\
\displaystyle{
+\frac12(\partial+\lambda)\pi\big(\rho\{a_\lambda [e,u]\}_z\big)
\,.}
\end{array}
\end{equation}
Using Table \ref{table3} and the definition of $\pi$, equation \eqref{20130403:eq1}
gives
\begin{equation}\label{20130403:eq2}
\begin{array}{c}
\displaystyle{
\pi\{a_\lambda\psi(u)\}_{z,\rho}
=
[a,u]+z(s|[a,u])
-\frac12\sum_{k\in J_1}[u,u^k]^\sharp
\big([a,[e,u_k]]^\sharp+(a|[e,u_k])\lambda\big)
} \\
\displaystyle{
+\frac12(\partial+\lambda)
\big([a,[e,u]]^\sharp+(a|[e,u])\lambda\big)
\,.}
\end{array}
\end{equation}
By assumption, $a\in\mf g_0^f=\mf g_0^e$. Hence, 
by invariance of the bilinear form we have $(a|[e,u])=0$ for all $u\in\mf g_{-1}$,
and by the Jacobi identity and the definition of $\sharp$ we have $[a,[e,u]]^\sharp=0$ for all $u\in\mf g_{-1}$.
Therefore, equation \eqref{20130403:eq2} reduces to
\begin{equation}\label{20130403:eq3}
\pi\{a_\lambda\psi(u)\}_{z,\rho}
=
[a,u]+z(s|[a,u])\,.
\end{equation}
Applying $\pi^{-1}$ to equation \eqref{20130403:eq3}, we get,
according to Corollary \ref{20130402:cor2},
$\{a_\lambda\psi(u)\}_{z,\rho}=\psi([a,u])+z(s|[a,u])$,
as stated in Table \ref{table4}.

Next, we want to compute the formula for the $\lambda$-bracket $\{\psi(u)_\lambda\psi(u_1)\}_{z,\rho}$.
As before, we first compute $\pi\{\psi(u)_\lambda\psi(u_1)\}_{z,\rho}$,
and then apply the inverse map $\pi^{-1}$.
Recalling that $[e,\mf g_{-1}]=[f,\mf g_1]\subset\ker(\pi)$ we have, 
by the formula \eqref{psiu} for $\psi(u)$ and the Leibniz rule,
\begin{equation}\label{20130403:eq4}
\begin{array}{l}
\displaystyle{
\pi\{\psi(u)_\lambda\psi(u_1)\}_{z,\rho}
=
\pi\big(\rho\{ u _\lambda u_1 \}_z\big)
-\frac12\sum_{k\in J_1}[u_1,u^k]^\sharp
\pi\big(\rho\{ u _\lambda [e,u_k] \}_z\big)
} \\
\displaystyle{
-\frac12\sum_{k\in J_1}
\pi\big(\rho\{ [e,u_k] _{\partial+\lambda} u_1 \}_z\big)_\to [u,u^k]^\sharp
+\frac12(\partial+\lambda) \pi\big(\rho\{ u _\lambda [e,u_1] \}_z\big)
} \\
\displaystyle{
-\frac12\lambda
\pi\big(\rho\{ [e,u] _\lambda u_1 \}_z\big)
+\frac14\sum_{h,k\in J_1}
[u_1,u^k]^\sharp \pi\big(\rho\{ [e,u_h] _{\partial+\lambda} [e,u_k] \}_z\big)_\to [u,u^h]^\sharp
} \\
\displaystyle{
-\frac14\sum_{k\in J_1}
(\partial+\lambda) \pi\big(\rho\{ [e,u_k] _{\partial+\lambda} [e,u_1] \}_z\big)_\to [u,u^k]^\sharp
} \\
\displaystyle{
+\frac14\sum_{k\in J_1}\lambda
[u_1,u^k]^\sharp \pi\big(\rho\{ [e,u] _\lambda [e,u_k] \}_z\big)
} \\
\displaystyle{
-\frac14\lambda(\partial+\lambda)
\pi\big(\rho\{ [e,u] _\lambda [e,u_1] \}_z\big)
\,.}
\end{array}
\end{equation}
According to Table \ref{table3} we have, by definition of $\pi$,
for arbitrary $u,u_1\in\mf g_{-1}$:
$$
\begin{array}{l}
\vphantom{\Big(}
\displaystyle{
\pi\big(\rho\{ u _\lambda u_1 \}_z\big)=0\,,
} \\
\vphantom{\Big(}
\displaystyle{
\pi\big(\rho\{ [e,u] _\lambda u_1 \}_z\big)
=-\pi\big(\rho\{ u _\lambda [e,u_1] \}_z\big)
=u\circ u_1+z(s|u\circ u_1)
\,,} \\
\vphantom{\Big(}
\displaystyle{
\pi\big(\rho\{ [e,u] _\lambda [e,u_1] \}_z\big)
=
[[e,u],[e,u_1]]-(e|u\circ u_1)\lambda
\,.
}
\end{array}
$$
In the last identity we used the fact that $[[e,\mf g_{-1}],[e,\mf g_{-1}]]\subset\mf g_0^f$
(cf. Lemma \ref{20130417:lem1}).
Using the above identities and the completeness relations \eqref{complete2}, 
equation \eqref{20130403:eq4} gives
\begin{equation}\label{20130403:eq5}
\begin{array}{l}
\displaystyle{
\pi\{\psi(u)_\lambda\psi(u_1)\}_{z,\rho}
=
\frac12\sum_{k\in J_1}(u\circ u_k) [u_1,u^k]^\sharp
-\frac12\sum_{k\in J_1} (u_1\circ u_k) [u,u^k]^\sharp
} \\
\displaystyle{
+\frac12z [[s,u],[e,u_1]]^\sharp
+\frac12 z [[e,u],[s,u_1]]^\sharp
} \\
\displaystyle{
-\frac12(\partial+\lambda) u\circ u_1
-\frac12\lambda u\circ u_1
-z \lambda (s|u\circ u_1)
} \\
\displaystyle{
+\frac14\sum_{h,k\in J_1} [[e,u_h],[e,u_k]] [u,u^h]^\sharp [u_1,u^k]^\sharp 
+\frac14\sum_{k\in J_1} [u_1,u^k]^\sharp (\partial+\lambda) [[e,u],[e,u_k]]
} \\
\displaystyle{
-\frac14\sum_{k\in J_1} (\partial+\lambda) [[e,u_k],[e,u_1]] [u,u^k]^\sharp
-\frac14 (\partial+\lambda)^2 [[e,u],[e,u_1]]
} \\
\displaystyle{
+\frac14 \lambda \sum_{k\in J_1} [[e,u],[e,u_k]] [u_1,u^k]^\sharp 
-\frac14 \lambda^2 [[e,u],[e,u_1]]
} \\
\displaystyle{
-\frac14\lambda(\partial+\lambda)[[e,u],[e,u_1]]
+\frac14\lambda^3(e|u\circ u_1)
\,.}
\end{array}
\end{equation}
Applying the bijective map $\pi^{-1}:\,\mc V(\mf g^f)\to\mc W$
to both sides of equation \eqref{20130403:eq5}
and using Corollary \ref{20130402:cor2} and equation \eqref{20130403:eq6}, 
we get equation \eqref{20130320:eq2b}.

The formula \eqref{virL} for $L$ reduces, in the special case of a short nilpotent element $f$, to
\begin{equation}\label{20130403:eq7}
L=
f+x'+\frac12\sum_{j\in J_0}a_ja^j
\,,
\end{equation}
while equation \eqref{psiu} reduces, for $u=f$, to
\begin{equation}\label{20130403:eq8}
\psi(f)
=f+x'-\frac14\sum_{k\in J_1}[f,u^k][e,u_k]
\,.
\end{equation}
Here we used the fact that $f\circ u=-2u$ for every $u\in\mf g_{-1}$.
Equation \eqref{20130322:eq3} follows from equations \eqref{20130403:eq7} and \eqref{20130403:eq8}
and the fact that 
$\big\{-\frac12[e,u_k]\big\}_{k\in J_1}$, $\big\{[f,u^k]\big\}_{k\in J_1}$,
are dual bases of $[f,\mf g_1]\subset\mf g_0$,
and $\mf g_0^f$ is the orthocomplement to $[f,\mf g_1]$ in $\mf g_0$.

Finally, we prove equations \eqref{20130322:eq9}.
If $a\in\mf g_0^f$ we have, from Table \ref{table4}, 
that $\{a_\lambda\psi(f)\}_{z,\rho}=0$,
and 
$$
\{a_\lambda\frac12\sum_{i\in J_0^f}a_ia^i\}_{z,\rho}
=\sum_{i\in J_0^f}\big([a,a^i]+(a|a^i)\lambda\big)a_i
=a\lambda\,.
$$
Here we used the fact that $\sum_{i\in J_0^f}a_ia^i$ is invariant with respect
to the adjoint action of $\mf g_0^f$,
and the completeness relation \eqref{complete1}.
Therefore, 
by equation \eqref{20130322:eq3}, we have
$$
\{a_\lambda L\}_{z,\rho}
=a\lambda\,.
$$
The first equation in \eqref{20130322:eq9} follows by skew-symmetry.
Similarly, for the second equation in \eqref{20130322:eq9} we have,
from equation \eqref{20130320:eq2b},
\begin{equation}\label{20130403:eq9}
\begin{array}{l}
\displaystyle{
\{\psi(u)_\lambda\psi(f)\}_{z,\rho}
=
\sum_{k\in J_1}\psi(u_k)[u,u^k]^\sharp
+(\partial+2\lambda)\psi(u)
-\frac12(e|u)\lambda^3
}\\
\displaystyle{
+\frac12z[[e,u],[s,f]]^\sharp
+2(s|u)z\lambda
\,.}
\end{array}
\end{equation}
Here we used the facts that $[f,\mf g_1]\subset\ker\pi$, so that $[f,v]^\sharp=0$ for every $v\in\mf g_1$,
that $f\circ u=u\circ f=-2u$ for every $u\in\mf g_{-1}$,
and that $[e,f]=2x$ commutes with $\mf g_0$.
Furthermore, by Table \ref{table4}, the left Leibniz rule, 
and the completeness relation \eqref{complete1}, 
we have
\begin{equation}\label{20130403:eq10}
\{\psi(u)_\lambda\frac12\sum_{i\in J_0^f}a_ia^i\}_{z,\rho}
=
\sum_{i\in J_0^f}\psi([u,a^i])a_i
+z[s,u]^\sharp
\,.
\end{equation}
It is not hard to check, using the completeness relations \eqref{complete1} and \eqref{complete2}, that
$$
\sum_{k\in J_1}\psi(u_k)[u,u^k]^\sharp
=-\sum_{i\in J_0^f}\psi([u,a^i])a_i
\,,
$$
and, by the Jacobi identity, that
$$
[[e,u],[s,f]]^\sharp
=-2[s,u]^\sharp
\,.
$$
Therefore, combining equations \eqref{20130403:eq9} and \eqref{20130403:eq10}, we get,
by \eqref{20130322:eq3}, that
$$
\{\psi(u)_\lambda L\}_{z,\rho}
=
(\partial+2\lambda)\psi(u)
-\frac12(e|u)\lambda^3
+2(s|u)z\lambda
\,.
$$
The second equation in \eqref{20130322:eq9} follows by skew-symmetry.
\end{proof}

\section{Generalized Drinfeld-Sokolov integrable bi-Hamiltonian hierarchies}
\label{sec:4}

In this section we remind the construction of generalized Drinfeld-Sokolov integrable
bi-Hamiltonian hierarchies, following \cite{DSKV12}.

According to the Lenard-Magri scheme of integrability \eqref{20130314:eq3},
in order to construct an integrable hierarchy of bi-Hamiltonian equations in $\mc W$,
we need to find a sequence of local functionals
$\tint g_n\in\quot{\mc W}{\partial\mc W},\,n\in\mb Z_+$, 
satisfying the following recursive equations ($w\in\mc W$):
\begin{equation}\label{eq:lenard2}
\tint \{{g_0}_\lambda w\}_{K,\rho}\big|_{\lambda=0}=0
\,\,\,\text{ and }\,\,\,
\tint \{{g_n}_\lambda w\}_{H,\rho}\big|_{\lambda=0}
=\tint \{{g_{n+1}}_\lambda w\}_{K,\rho}\big|_{\lambda=0}
\,.
\end{equation}
In this case it is well known and easy to prove (see e.g. \cite[Lem.2.6]{BDSK09})
that the $\tint g_n$'s are in involution with respect to both $H$ and $K$:
$$
\{\tint g_m,\tint g_n\}_{H,\rho}
=\{\tint g_m,\tint g_n\}_{K,\rho}=0
\qquad\text{for all }m,n\in\mb Z_+\,,
$$
and we get the corresponding integrable hierarchy of Hamiltonian equations:
\begin{equation}\label{eq:hierarchy}
\frac{dw}{dt_n}
=\{{g_n}_\lambda w\}_{H,\rho}\big|_{\lambda=0}
\,\,,\,\,\,\,
n\in\mb Z_+\,,
\end{equation}
provided that the $\tint g_n$'s span an infinite dimensional subspace of 
$\quot{\mc W}{\partial\mc W}$.

Consider the Lie algebra $\mf g((z^{-1}))=\mf g\otimes\mb F((z^{-1}))$,
where $\mb F((z^{-1}))$ is the field of formal Laurent series in the indeterminate $z^{-1}$.
Introduce the $\mb Z$-grading of $\mf g((z^{-1}))$ by letting 
$$
\deg(a\otimes z^k)=i-(d+1)k
\,\,,\,\,\,\,
\text{ for } 
a\in\mf g_i
\text{ and }
k\in\mb Z
\,,
$$
where $d$ is the depth of the grading \eqref{dec}.
Then, for $s\in\mf g_d$, the element $f+zs$ is homogeneous of degree $-1$.
We denote by $\mf g((z^{-1}))_i\subset\mf g((z^{-1}))$
the space of homogeneous elements of degree $i$.
Then we have the grading
\begin{equation}\label{decz}
\mf g((z^{-1}))=\widehat\bigoplus_{i\in\frac12\mb Z}\mf g((z^{-1}))_i\,,
\end{equation}
where the direct sum is completed by allowing infinite series in positive degrees.

Recall the following well known facts about semisimple elements in a simple,
or more generally reductive, Lie algebra:
\begin{lemma}\label{20130506:lem}
The following conditions are equivalent for an element $\Lambda$ 
of a reductive Lie algebra $\mf g$:
\begin{enumerate}[(i)]
\item
$\ad\Lambda$ is a semisimple endomorphism of $\mf g$;
\item
the adjoint orbit of $\Lambda$ is closed;
\item
$\ker(\ad\Lambda)\cap\im(\ad\Lambda)=0$;
\item
$\mf g=\ker(\ad\Lambda)+\im(\ad\Lambda)$;
\item
$\mf g=\ker(\ad\Lambda)\oplus\im(\ad\Lambda)$.
\end{enumerate}
\end{lemma}
\begin{proof}
The equivalence of (i) and (ii) is well known, \cite{OV89}.
Clearly, conditions (iii), (iv) and (v) are equivalent by linear algebra,
since $\mf g$ is finite dimensional.
Moreover, condition (i) obviously implies (iii), again by linear algebra.
To conclude, we shall prove that condition (iii) implies condition (i).
Consider the Jordan decomposition $\Lambda=s+n$,
where $s,n\in\mf g$ are respectively the semisimple and nilpotent parts of $\Lambda$.
Since $s$ and $n$ commute,
$n$ lies in $\mf g^s$, the centralizer of $s$,
and so $\ad n|_{\mf g^s}$ is a nilpotent endomorphism of $\mf g^s$.
By assumption, $\ker(\ad\Lambda)\cap\im(\ad\Lambda)=0$,
and, a fortiori, $\ker(\ad\Lambda|_{\mf g^s})\cap\im(\ad\Lambda|_{\mf g^s})=0$,
which is the same as saying that $\ker(\ad n|_{\mf g^s})\cap\im(\ad n|_{\mf g^s})=0$.
But for a nilpotent element, kernel and image having zero intersection is the same
as the element being zero. Therefore, $\ad n|_{\mf g^s}=0$,
i.e. $n$ is a central element of $\mf g^s$.
But $\mf g^s\subset\mf g$ is a reductive subalgebra,
hence its center consists of semisimple elements of $\mf g$, \cite{OV89}.
It follows that $n=0$.
\end{proof}
We shall assume that $f+s$ is a semisimple element of $\mf g$.
Then $f+zs$ is a semisimple element of $\mf g((z^{-1}))$.
(Indeed, for any scalar $t$ we have a Lie algebra automorphism acting as $t^i$ in $\mf g_i$.
On the other hand, $f\in\mf g_{-1}$, and $s$ is a homogeneous element of $\mf g_{\geq0}$.
Hence, considering $z$ as an element of the field $\mb F((z^{-1}))$,
$f+s$ is semisimple if and only if $f+zs$ is semisimple.)
Therefore we have the following direct sum decomposition 
\begin{equation}\label{eq:dech}
\mf g((z^{-1}))=\mf h\oplus\mf h^\perp
\,,\,\,\text{ where }
\mf h:=\Ker\ad(f+zs)
\,\,\text{ and }\,\,
\mf h^\perp:=\im\ad(f+zs)
\,.
\end{equation}
Since $f+zs$ is homogeneous,
the decomposition \eqref{decz} induces the corresponding decompositions
of $\mf h$ and $\mf h^\perp$:
\begin{equation}\label{20130404:eq6}
\mf h=\widehat\bigoplus_{i\in\frac12\mb Z}\mf h_i
\,\,\text{ and }\,\,
\mf h^\perp=\widehat\bigoplus_{i\in\frac12\mb Z}\mf h^\perp_i
\,.
\end{equation}

Recall that $\mc V(\mf g_{\leq\frac12})$ is a commutative associative algebra
with derivation $\partial$.
Hence we may consider the Lie algebra 
$$
\mb F\partial\ltimes\big(\mf g((z^{-1}))\otimes\mc V(\mf g_{\leq\frac12})\big)\,,
$$
where $\partial$ acts on the second factor.
Note that $\mf g((z^{-1}))_{>0}\otimes\mc V(\mf g_{\leq\frac12})$
is a pro-nilpotent subalgebra.
Hence,
for $U(z)\in \mf g((z^{-1}))_{>0}\otimes\mc V(\mf g_{\leq\frac12})$
we have a well defined automorphism
$e^{\ad U(z)}$ of the Lie algebra 
$\mb F\partial\ltimes\big(\mf g((z^{-1}))\otimes\mc V(\mf g_{\leq\frac12})\big)$.

As in Section \ref{sec:2.1},
we fix a basis $\{q_i\}_{i\in J_{\leq\frac12}}$ of $\mf g_{\leq\frac12}$,
and we let $\{q^i\}_{i\in J_{\leq\frac12}}$ be the dual basis of $\mf g_{\geq-\frac12}$,
w.r.t. the bilinear form $(\cdot\,|\,\cdot)$.
We denote
\begin{equation}\label{20130404:eq1}
q=\sum_{i\in J_{\leq\frac12}}q^i\otimes q_i
\,\in\mf g\otimes\mc V(\mf g_{\leq\frac12})
\,.
\end{equation}
\begin{proposition}[{\cite[Prop.4.5]{DSKV12}}]
\label{int_hier2_ds}
There exist formal Laurent series $U(z)\in\mf{g}((z^{-1}))_{>0}\otimes\mc V(\mf g_{\leq\frac12})$
and $h(z)\in\mf{h}_{>-1}\otimes\mc V(\mf g_{\leq\frac12})$ such that
\begin{equation}\label{L0_dsr}
e^{\ad U(z)}(\partial+(f+zs)\otimes1+q)=\partial+(f+zs)\otimes1+h(z)\,.
\end{equation}
The elements $U(z)$ and $h(z)$ solving \eqref{L0_dsr} are uniquely determined
if we require that
$U(z)\in\mf h^\perp_{>0}\otimes\mc V(\mf g_{\leq\frac12})$
\end{proposition}

The key result of this section is the following theorem,
which allows us to construct an integrable hierarchy of bi-Hamiltonian equations.
\begin{theorem}[{\cite[Thm.4.18]{DSKV12}}]\label{final}
Assume that $s\in\mf g_d$ is such that $f+s$ is a semisimple element of $\mf g$.
Let $U(z)\in\mf g((z^{-1}))_{>0}\otimes\mc V(\mf g_{\leq\frac12})$ 
and $h(z)\in\mf h_{>-1}\otimes\mc V(\mf g_{\leq\frac12})$
be a solution of equation \eqref{L0_dsr}.
Let $0\neq a(z)\in Z(\mf h)$, the center of $\mf h\subset\mf g((z^{-1}))$.
Then, the coefficients $\tint g_n,\,n\in\mb Z_+$,
of the Laurent series $\tint g(z)=\sum_{n\in\mb Z_+}\tint g_nz^{N-n}$ 
defined by
\begin{equation}\label{gz}
\tint g(z)=\tint (a(z)\otimes1 | h(z))\,,
\end{equation}
span an infinite-dimensional subspace of $\quot{\mc W}{\partial\mc W}$
and they satisfy the Lenard-Magri recursion conditions \eqref{eq:lenard2}. 
Hence, 
we get an integrable hierarchy of bi-Hamiltonian equations \eqref{eq:hierarchy},
called a \emph{generalized Drinfeld-Sokolov hierarchy}.
\end{theorem}
\begin{remark}\label{20130404:rem}
It follows from \cite[Prop.2.10]{BDSK09}
that all the $\tint g_n$'s obtained by taking all possible $a(z)\in Z(\mf h)$
are in involution.
Thus we attach an integrable bi-Hamiltonian hierarchy to a nilpotent element $f$
of a simple Lie algebra $\mf g$ and a choice of $s\in\mf g_d$ such that $f+s$
is a semisimple element of $\mf g$.
\end{remark}

\section{Generalized Drinfeld-Sokolov hierarchy for a minimal nilpotent}
\label{sec:5a}

\subsection{Preliminary computations}

As in Section \ref{sec:3} we assume, without loss of generality, that $s=e$.
We start by finding the direct sum decomposition \eqref{eq:dech}.
\begin{lemma}\label{minimal_semisimple}
The element $f+ze\in\mf g((z^{-1}))$ is semisimple, and we have the decomposition
$\mf g((z^{-1}))=\mf h\oplus\mf h^\perp$, where
\begin{equation}\label{20130404:eq4}
\mf h=\Ker\ad(f+ze)=
\mf g_0^f((z^{-1}))\oplus\mb F(f+ze)((z^{-1})) 
\end{equation}
and
\begin{equation}\label{20130404:eq5}
\mf h^\perp=\im\ad(f+ze)=
\mf g_{-\frac12}((z^{-1}))\oplus\mb Fx((z^{-1}))\oplus\mf g_{\frac12}((z^{-1}))
\oplus\mb F(f-ze)((z^{-1}))\,.
\end{equation}
\end{lemma}
\begin{proof}
Since $\mf g_0^f=\mf g_0^e$, we have
$\mf g_0^f((z^{-1}))\subset\Ker\ad(f+ze)$. 
Moreover, obviously $f+ze\in\Ker\ad(f+ze)$.
Therefore
\begin{equation}\label{20130404:eq2}
\mf g_0^f((z^{-1}))\oplus\mb F(f+ze)((z^{-1})) 
\subset\ker\ad(f+ze)\,.
\end{equation}
On the other hand, we have
$$
\begin{array}{l}
\displaystyle{
\vphantom{\Big(}
[f+ze,\mf g_{\frac12}((z^{-1}))]=[f,\mf g_{\frac12}((z^{-1}))]=\mf g_{-\frac12}((z^{-1}))
\,,} \\
\displaystyle{
\vphantom{\Big(}
[f+ze,\mf g_{-\frac12}((z^{-1}))]=z[e,\mf g_{-\frac12}((z^{-1}))]=\mf g_{\frac12}((z^{-1}))
\,,}
\end{array}
$$
and
\begin{equation}\label{20130405:eq5}
[f+ze,\big(\frac14(f-ze)z^{-1}\big)\big]=x
\,\,\text{ and }\,\,
[f+ze,x]=f-ze\,.
\end{equation}
Therefore
\begin{equation}\label{20130404:eq3}
\mf g_{-\frac12}((z^{-1}))\oplus\mb Fx((z^{-1}))\oplus\mf g_{\frac12}((z^{-1}))
\oplus\mb F(f-ze)((z^{-1}))
\subset\im\ad(f+ze)
\,.
\end{equation}
Equalities \eqref{20130404:eq4} and \eqref{20130404:eq5}
immediately follow from the inclusions \eqref{20130404:eq2} and \eqref{20130404:eq3}.
\end{proof}

Since $d=1$, the degree of $z$ equals $-2$.
It is then easy to find each piece $\mf h_i$ and $\mf h^\perp_i$, $i\in\frac12\mb Z$, 
of the decompositions \eqref{20130404:eq6}.
We have
\begin{enumerate}[(i)]
\item
$\mf h_i=0$ for $i\in\mb Z+\frac12$,
\item
$\mf h_i=\mf g_0^fz^{-\frac{i}{2}}$ for $i\in2\mb Z$,
\item
$\mf h_i=\mb F (f+ze)z^{-\frac{i+1}{2}}$ for $i\in2\mb Z-1$,
\item
$\mf h^\perp_i=\mf g_{-\frac12}z^{-\frac{2i+1}{4}}$ for $i\in2\mb Z-\frac12$,
\item
$\mf h^\perp_i=\mf g_{\frac12}z^{-\frac{2i-1}{4}}$ for $i\in2\mb Z+\frac12$,
\item
$\mf h^\perp_i=\mb Fxz^{-\frac{i}{2}}$ for $i\in2\mb Z$,
\item
$\mf h^\perp_i=\mb F (f-ze)z^{-\frac{i+1}{2}}$ for $i\in2\mb Z-1$.
\end{enumerate}

Recall from Section \ref{sec:2.1}
that 
$\{a_j\}_{j\in J_0^f}$ and $\{a^j\}_{j\in J_0^f}$ denote dual bases of $\mf g_0^f$
with respect to the inner product $(\cdot\,|\,\cdot)$,
and that $\{v^k\}_{k\in J_{\frac12}}$ and $\{v_k\}_{k\in J_{\frac12}}$
denote bases of $\mf g_{\frac12}$ 
dual with respect to the skewsymmetric bilinear form $\omega_+$
defined in \eqref{20130201:eq5}.
Therefore, $\{[f,v_k]\}_{k\in J_{\frac12}}\subset\mf g_{-\frac12}$ 
and $\{v^k\}_{k\in J_{\frac12}}\subset\mf g_{\frac12}$
are dual bases with respect to $(\cdot\,|\,\cdot)$.
Then, the element $q\in\mf g_{\geq-\frac12}\otimes\mc V(\mf g_{\leq\frac12})$ defined
in \eqref{20130404:eq1} is the following
\begin{equation}\label{20130404:eq7}
q=
\sum_{k\in J_\frac12}[f,v_k]\otimes v^k
+\sum_{k\in J_\frac12} v^k\otimes[f,v_k]
+\sum_{i\in J_0^f}a_i\otimes a^i
+\frac x{(x|x)}\otimes x
+\frac e{2(x|x)}\otimes f
\,.
\end{equation}

In order to apply Theorem \ref{final}
we need to find (unique)
$U(z)\in\mf h^\perp_{>0}\otimes\mc V(\mf g_{\leq\frac12})$ 
and $h(z)\in\mf h_{>-1}\otimes\mc V(\mf g_{\leq\frac12})$
solving equation \eqref{L0_dsr}.
We will do it recursively by degree.
We extend the degree \eqref{decz} of $\mf g((z^{-1}))$ 
to $\mf g((z^{-1}))\otimes\mc V(\mf g_{\leq\frac12})$
by letting the degree of elements of $\mc V(\mf g_{\leq\frac12})$ be zero.
We can then expand $U(z)$ and $h(z)$ 
according to the decompositions \eqref{20130404:eq6}
for $\mf h$ and $\mf h^\perp$:
$$
\begin{array}{l}
\displaystyle{
\vphantom{\Big(}
U(z)=U(z)_{\frac12}+U(z)_1+U(z)_{\frac32}+U(z)_2+\dots
\,\,\text{ with } U(z)_i\in\mf h^\perp_i\otimes\mc V(\mf g_{\leq\frac12})
\,,} \\
\displaystyle{
\vphantom{\Big(}
h(z)=h(z)_0+h(z)_1+h(z)_2+\dots
\,\,\text{ with } h(z)_i\in\mf h_i\otimes\mc V(\mf g_{\leq\frac12})
\,,}
\end{array}
$$
(here we are using the fact that $\mf h_i=0$ for semi-integer $i$),
and we can also expand accordingly the element $q$ in \eqref{20130404:eq7}
as $q=q_{-\frac12}+q_0+q_{\frac12}+q_1$, where
$$
\begin{array}{l}
\displaystyle{
q_{-\frac12}=
\sum_{k\in J_\frac12}[f,v_k]\otimes v^k
\,\,,\,\,\,\,
q_{0}=
\sum_{i\in J_0^f}a_i\otimes a^i
+\frac x{(x|x)}\otimes x
\,,} \\
\displaystyle{
q_{\frac12}=
\sum_{k\in J_\frac12} v^k\otimes[f,v_k]
\,\,,\,\,\,\,
q_{1}=
\frac e{2(x|x)}\otimes f
\,.}
\end{array}
$$

Clearly, both sides of equation \eqref{L0_dsr} have the form $\partial+(f+ze)\otimes1+$
some expression in $\mf g((z^{-1}))_{>-1}\otimes\mc V(\mf g_{\leq\frac12})$,
and we could solve it, degree by degree, by comparing the terms in
$\mf g((z^{-1}))_i\otimes\mc V(\mf g_{\leq\frac12})$ for each $i\geq-\frac12$
in both sides of the equation.
Instead, we will use a trick similar to the one used for the proof of Theorem \ref{prop:minimal}.
Recall the definition of the quotient map $\pi:\,\mc V(\mf g_{\leq\frac12})\to\mc V(\mf g^f)$,
which, by Corollary \ref{20130402:cor1}, restricts to a bijection on $\mc W$,
and the inverse map $\pi^{-1}:\,\mc V(\mf g^f)\to\mc W$.
We extend these maps to
\begin{equation}\label{20130406:eq4}
\pi:\,\mf g((z^{-1}))\otimes\mc V(\mf g_{\leq\frac12})\twoheadrightarrow\mf g((z^{-1}))\otimes\mc V(\mf g^f)
\,\,,\,\,\,\,
\pi^{-1}:\,\mf g((z^{-1}))\mc V(\mf g^f)\stackrel{\sim}{\to}\mf g((z^{-1}))\otimes\mc W\,,
\end{equation}
by acting as the identity on the first factors.
Then, instead of computing $U(z)$ and $h(z)$ in low degrees (which 
becomes quite lengthy even at low degrees),
we will compute their projections $\pi U(z)$ and $\pi h(z)$
in $\mf h^\perp_{>0}\otimes\mc V(\mf g^f)$ and $\mf h_{>-1}\otimes\mc V(\mf g^f)$
respectively, degree by degree.
Using equation \eqref{gz} we then get $\tint \pi g(z)$,
and therefore, applying the inverse map $\pi^{-1}$ by use of Corollary \ref{20130402:cor1}
(which is a differential algebra homomorphism, and therefore it commutes with taking $\tint$),
we get $\tint g(z)\in\quot{\mc W}{\partial\mc W}((z^{-1}))$.

Since $\pi:\,\mc V(\mf g_{\leq\frac12})\to\mc V(\mf g^f)$
is a homomorphism of differential algebras,
it follows that $\pi:\,\mf g((z^{-1}))\otimes\mc V(\mf g_{\leq\frac12})\to\mf g((z^{-1}))\otimes\mc V(\mf g^f)$
is a homomorphism of Lie algebras.
Therefore, applying $\pi$ to both sides of equation \eqref{L0_dsr},
we get
\begin{equation}\label{20130405:eq6}
e^{\ad \pi U(z)}(\partial+(f+ze)\otimes1+\pi q)=\partial+(f+ze)\otimes1+\pi h(z)\,,
\end{equation}
and, by the formula \eqref{20130404:eq7} for $q$ we get
$\pi q=\pi q_0+\pi q_{\frac12}+\pi q_1$, where
\begin{equation}\label{20130405:eq7}
\pi q_{0}=
\sum_{i\in J_0^f}a_i\otimes a^i
\,\,,\,\,\,\,
\pi q_{\frac12}=
\sum_{k\in J_\frac12} v^k\otimes[f,v_k]
\,\,,\,\,\,\,
\pi q_{1}=
\frac e{2(x|x)}\otimes f
\,.
\end{equation}

We start by looking at the homogeneous components of degree $-\frac12$
in both sides of equation \eqref{20130405:eq6}. We get
the following equation
$$
[\pi U(z)_{\frac12},(f+ze)\otimes1]=0\,,
$$
which implies $\pi U(z)_{\frac12}=0$,
since $\ad(f+ze)$ is bijective on $\mf h^\perp$.
Similarly, taking the homogeneous components of degree $0$ in both sides 
of equation \eqref{20130405:eq6}, we get
$$
\pi h(z)_0
=
\pi q_0+
[\pi U(z)_1,(f+ze)\!\otimes\!1]
\,.
$$
By equation \eqref{20130405:eq7} we have $\pi q_0\in\mf g^f_0\otimes\mc V(\mf g^f)$.
Therefore, 
by looking at the components in $\mf h_0=\mf g_0^f$
and $\mf h^\perp_0=\mb Fx$ separately, we get
that $\pi U(z)_1=0$, and
\begin{equation}\label{20130405:eq3}
\pi h(z)_0
=
\sum_{i\in J_0^f}a_i\otimes a^i
\,.
\end{equation}
Next, we take the homogeneous components of degree $\frac12$ in both sides 
of equation \eqref{20130405:eq6}:
$$
\pi q_{\frac12}+[\pi U(z)_{\frac32},(f+ze)\otimes1]=0\,.
$$
Recalling the expression \eqref{20130405:eq7} for $\pi q_{\frac12}$,
and using the fact that $v^k\!=\![f+ze,[f,v^k]z^{-1}]$,
we deduce that
\begin{equation}\label{20130405:eq8}
\pi U(z)_{\frac32}
=
\sum_{k\in J_\frac12} [f,v^k]z^{-1}\otimes[f,v_k]
\,.
\end{equation}
We then take the homogeneous components of degree $1$ in both sides of equation \eqref{20130405:eq6}:
$$
\pi h(z)_1=
\pi q_1+[\pi U(z)_2,(f+ze)\otimes1]\,.
$$
Recalling the expression \eqref{20130405:eq7} of $\pi q_1$,
using the obvious decomposition
$$
e=\frac12(f+ze)z^{-1}-\frac12(f-ze)z^{-1}\in\mf h_1\oplus\mf h^\perp_1\,,
$$
and the second equation in \eqref{20130405:eq5}, we get
\begin{equation}\label{20130405:eq9}
\begin{array}{l}
\displaystyle{
\vphantom{\Big(}
\pi h(z)_1
=
\frac1{4(x|x)}
(f+ze)z^{-1}
\otimes f
\,,} \\
\displaystyle{
\pi U(z)_2
=
-\frac1{4(x|x)}xz^{-1}\otimes f
\,.}
\end{array}
\end{equation}
Next, we take the terms of degree $\frac32$ in both sides of equation \eqref{20130405:eq6}:
$$
-[\partial,\pi U(z)_{\frac32}]+[\pi U(z)_{\frac52},(f+ze)\otimes1]
+[\pi U(z)_{\frac32},\pi q_0]=0
\,.
$$
Recalling the formulas \eqref{20130405:eq7} for $\pi q_0$ 
and \eqref{20130405:eq8} for $\pi U(z)_{\frac32}$,
we get
$$
[(f+ze)\otimes1,\pi U(z)_{\frac52}]
=
-\sum_{k\in J_\frac12} [f,v^k]z^{-1}\otimes\partial[f,v_k]
+\!\!
\sum_{i\in J_0^f,k\in J_\frac12}\!\!
[ [f,v^k], a^i ]z^{-1} \otimes a_i[f,v_k]
\,.
$$
Note that $[f,v^k]=[f+ze,v^k]$, and, since $a^i\in\mf g_0^f$, we also have,
by the Jacobi identity, $[ [f,v^k], a^i ]=-[f+ze,[a^i,v^k]]$.
Hence,
\begin{equation}\label{20130406:eq1}
\pi U(z)_{\frac52}
=
-\sum_{k\in J_\frac12} v^k z^{-1}\otimes\partial[f,v_k]
-\sum_{i\in J_0^f,k\in J_\frac12}
[a^i,v^k]z^{-1} \otimes a_i[f,v_k]
\,.
\end{equation}
Next, taking the terms of degree $2$ in both sides of equation \eqref{20130405:eq6}, we get:
$$
\begin{array}{l}
\displaystyle{
\pi h(z)_2+[(f+ze)\otimes1,\pi U(z)_3]
=
-[\partial,\pi U(z)_2]
+[\pi U(z)_2,\pi q_0]
} \\
\displaystyle{
+[\pi U(z)_{\frac32},\pi q_{\frac12}]
+\frac12[\pi U(z)_{\frac32},[\pi U(z)_{\frac32},(f+ze)\otimes1]]
\,.}
\end{array}
$$
Substituting the expressions 
\eqref{20130405:eq7} for $\pi q_0$ and $\pi q_{\frac12}$,
\eqref{20130405:eq8} for $\pi U(z)_{\frac32}$,
and \eqref{20130405:eq9} for $\pi U(z)_2$,
we get
$$
\begin{array}{l}
\displaystyle{
\pi h(z)_2+[(f+ze)\otimes1,\pi U(z)_3]
=
\frac1{4(x|x)} xz^{-1}\otimes \partial f
} \\
\displaystyle{
+\frac12
\sum_{h,k\in J_\frac12} [[f,v^h],v^k] z^{-1}\otimes[f,v_h][f,v_k]
\,.}
\end{array}
$$
Here we used the fact that, by Lemma \ref{20130315:lem2}(a), $[[f,v^k] ,f+ze]=-v^k$.
To find $\pi h(z)_2$ and $\pi U(z)_3$,
we find the components in 
$\mf h_2\otimes\mc V(\mf g^f)=\mf g_0^fz^{-1}\otimes\mc V(\mf g^f)$
and in $\mf h^\perp_2\otimes\mc V(\mf g^f)=\mb Fxz^{-1}\otimes\mc V(\mf g^f)$
in the RHS above.
By the definition \eqref{20130201:eq5} of the skew-symmetric form $\omega_+$,
we have the decomposition 
$$
[[f,v^h],v^k ]=[[f,v^h],v^k]^\sharp+\frac12\frac{\omega_+(v^k,v^h)}{(x|x)}x
\in\mf g_0^f\oplus\mb Fx
\,.
$$
Moreover, by the completeness relation \eqref{completeness}, we have
$$
\sum_{h,k\in J_\frac12} \omega_+(v^k,v^h) [f,v_h][f,v_k]
=\sum_{k\in J_\frac12} [f,v^k][f,v_k]\,,
$$
which is zero by skewsymmetry of the bilinear form $\omega_+$.
Therefore, by the first equation in \eqref{20130405:eq5}, we get
\begin{equation}\label{20130406:eq2}
\begin{array}{l}
\displaystyle{
\pi h(z)_2
=
\frac12
\sum_{h,k\in J_\frac12} [[f,v^h],v^k]^\sharp z^{-1}\otimes[f,v_h][f,v_k]
} \\
\displaystyle{
\pi U(z)_3
=
\frac1{16(x|x)} (f-ze)z^{-2}\otimes\partial f
\,.}
\end{array}
\end{equation}
It turns out that, in order to compute $\pi h(z)_3$, we do not need to compute $\pi U(z)_{\frac72}$.
Therefore, we look at the terms of degree $3$ in both sides of equation \eqref{20130405:eq6}.
We get
$$
\begin{array}{l}
\displaystyle{
\vphantom{\Big(}
\pi h(z)_3+[(f+ze)\otimes1,\pi U(z)_4]
=
[\pi U(z)_3,\partial]
+[\pi U(z)_3,\pi q_0]
+[\pi U(z)_{\frac52},\pi q_{\frac12}]
\,} \\
\displaystyle{
\vphantom{\Big(}
+[\pi U(z)_2,\pi q_1]
+\frac12[\pi U(z)_{\frac32},[\pi U(z)_{\frac32},\partial]]
+\frac12[\pi U(z)_{\frac32},[\pi U(z)_{\frac52},(f+ze)\otimes1]]
\,} \\
\displaystyle{
\vphantom{\Big(}
+\frac12[\pi U(z)_2,[\pi U(z)_2,(f+ze)\otimes1]]
+\frac12[\pi U(z)_{\frac52},[\pi U(z)_{\frac32},(f+ze)\otimes1]]
\,} \\
\displaystyle{
\vphantom{\Big(}
+\frac12[\pi U(z)_{\frac32},[\pi U(z)_{\frac32},\pi q_0]]
\,.}
\end{array}
$$
By equations \eqref{20130405:eq7}, \eqref{20130405:eq8},
\eqref{20130405:eq9}, \eqref{20130406:eq1}, and \eqref{20130406:eq2},
we get, after some manipulations using equation \eqref{20130315:eq5} and \eqref{20130315:eq6},
and the completeness relation \eqref{completeness},
$$
\begin{array}{l}
\displaystyle{
\vphantom{\Big(}
\pi h(z)_3+[(f+ze)\otimes1,\pi U(z)_4]
\,} \\
\displaystyle{
\vphantom{\Big(}
=
-\frac1{16(x|x)} (f-ze)z^{-2}\otimes\partial^2 f
+\frac1{32(x|x)^2}(f-3ze)z^{-2}\otimes f^2
} \\
\displaystyle{
\vphantom{\Big(}
-\frac1{4(x|x)} e z^{-1}\otimes \sum_{k\in J_\frac12}[f,v^k] \partial[f,v_k]
} \\
\displaystyle{
\vphantom{\Big(}
-\frac1{4(x|x)} e z^{-1}\otimes \sum_{i\in J_0^f,k\in J_\frac12}[f,[a^i,v^k]] a_i[f,v_k]
\,.}
\end{array}
$$
Taking the component of both sides in 
$\mf h_3\otimes\mc V(\mf g^f)=\mb F(f+ze)z^{-2}\otimes\mc V(\mf g^f)$,
i.e. replacing $f$ and $ze$ by $\frac12(f+ze)$ in the first factors,
we get the formula for $\pi h(z)_3$:
\begin{equation}\label{20130406:eq3}
\begin{array}{l}
\displaystyle{
\vphantom{\Big(}
\pi h(z)_3
=
(f+ze)z^{-2}\otimes
\Big(
-\frac1{32(x|x)^2} f^2
-\frac1{8(x|x)} \sum_{k\in J_\frac12}[f,v^k] \partial[f,v_k]
} \\
\displaystyle{
\vphantom{\Big(}
-\frac1{8(x|x)} \sum_{i\in J_0^f,k\in J_\frac12}[a^i,[f,v^k]] a_i[f,v_k]
\Big)
\,.}
\end{array}
\end{equation}

\subsection{First few equations of the hierarchies}\label{sec:6.2}

According to Theorem \ref{final},
for each non-zero element $a(z)\in Z(\mf h)$
there is an associated integrable bi-Hamiltonian hierarchy
corresponding to the Laurent series $\tint g(z)\in\quot{\mc W}{\partial\mc W}((z^{-1}))$
defined by \eqref{gz}.
By equation \eqref{20130404:eq4},
the center of $\mf h$ is spanned over $\mb F((z^{-1}))$
by $Z(\mf g_0^f)$ and the element $f+ze$.
(Note that $Z(\mf g_0^f)$ is 1-dimensional for $\mf g=\mf{sl}_n$, $n\geq3$,
and it is zero in all other cases.)
Hence, we will consider separately the following two choices:
\begin{enumerate}[(a)]
\item $a(z)=c\in Z(\mf g_0^f)$,
\item $a(z)=f+ze$.
\end{enumerate}

\subsubsection*{Case (a)}

Let $a(z)=c\in Z(\mf g_0^f)$.
By equation \eqref{gz} and the description (i)-(iii) of the subspaces $\mf h_i$, $i\in\frac12\mb Z$,
we obtain $\tint g(z)=\sum_{n\in\mb Z}\tint g_n z^{-n}$, where
\begin{equation}\label{20130406:eq5}
\tint g_n z^{-n}=
\tint (c\otimes1 | h(z)_{2n-1}+h(z)_{2n})
\,\,\text{ for all }\,\,
n\in\mb Z_+
\,.
\end{equation}
Equations \eqref{20130405:eq3}, \eqref{20130405:eq9} and \eqref{20130406:eq2}
give us the value of $\pi h(z)_n$ for $n=0,1,2$.
On the other hand, the quotient map 
$\pi:\,\mf g((z^{-1}))\times\mc V(\mf g_{\leq\frac12})\to\mf g((z^{-1}))\times\mc V(\mf g^f)$ 
defined in \eqref{20130406:eq4}
acts as the identity on the first factor,
and it is a differential algebra homomorphism on the second factor.
Therefore, it commutes with taking inner product $(\cdot\,|\,\cdot)$,
and with the $\tint$ sign.
Therefore, applying $\pi$ to both sides of equations \eqref{20130406:eq5}, we get
$$
\tint \pi g_n z^{-n}=
\tint (c\otimes1 | \pi h(z)_{2n-1}+\pi h(z)_{2n})
\,\,\text{ for all }\,\,
n\in\mb Z_+
\,,
$$
and in order to reconstruct $\tint g_n$ from this equation,
we just apply the inverse map $\pi^{-1}$ using Corollary \ref{20130402:cor1}.
Therefore,
$$
\tint g_0=\tint \phi(c)
\quad\text{and}\quad
\tint g_1=\frac12\sum_{k\in J_{\frac12}} \tint \psi([f,v_k])\psi([c,[f,v^k]])
\,.
$$
%
We can use these first two integrals of motion to write down the corresponding first two equations
of the hierarchy:
$\frac{dw}{dt_n}=\{{g_n}_\lambda w\}_{H,\rho}\big|_{\lambda=0}$, $n\in\mb Z_+$,
where $w$ is a generator of the $\mc W$-algebra.
They are ($a\in\mf g_0^f,\,u\in\mf g_{-\frac12}$):
\begin{equation}\label{eq:6.17}
\frac{d\phi(a)}{dt_{\tilde0}}=0
\,\,,\,\,\,\,
\frac{d\psi(u)}{dt_{\tilde0}}=\psi([c,u])
\,\,,\,\,\,\,
\frac{d L}{dt_{\tilde0}}=0,
\end{equation}
and
\begin{equation}\label{eq:6.18}
\begin{array}{l}
\displaystyle{
\frac{d\phi(a)}{dt_{\tilde1}}=0
\,\,,\,\,\,\,
\frac{d L}{dt_{\tilde1}}
=\sum_{k\in J_{\frac12}}\left(\psi([f,v_k])\psi([c,[f,v^k]])
\right)^\prime
\,,
}\\
\displaystyle{
\frac{d\psi(u)}{dt_{\tilde1}}=-\frac1{2(x|x)}\tilde{L}\psi([c,u])
+\sum_{i,j\in J_0^f}\phi(a_i)\phi(a_j)\psi([c,[a^i,[a^j,u]]])}\\
\displaystyle{
-\sum_{i\in J_0^f}\psi([c,[a^i,u]])\phi(a_i)'
-2\sum_{i\in J_0^f}\phi(a_i)\psi([c,[a^i,u]])'
+\psi([c,u])''
}
\end{array}
\end{equation}
where $\tilde{L}$ was defined in Corollary \ref{20130402:cor1}.
\begin{remark}\label{20130406:rem}
Note that
elements $\varphi(a)$, $a\in\mf g_0^f$, are central with respect to the $\lambda$-bracket 
$\{\cdot\,_\lambda\,\cdot\}_{K,\rho}$.
Therefore, $\frac{d\varphi(a)}{dt_{\tilde n}}=0$ for all $n\in\mb Z_+$.
\end{remark}

\subsubsection*{Case (b)}

Let $a(z)=f+ze\in Z(\mf h)$.
As before, we write the integrals of motion $\tint g_n$, $n\in\mb Z_+$,
in terms of the various components of $h(z)\in\mf h\otimes\mc V(\mf g_{\leq\frac12})$.
Recall that $\mf h_{2n}=\mf g_0^fz^{-n}$, which is orthogonal to $e$ and $f$ w.r.t. $(\cdot\,|\,\cdot)$.
Therefore the even components of $h(z)$ do not contribute to $\tint g(z)$.
Furthermore, 
we have $\mf h_{2n-1}=\mb F(f+ze)z^{-n}$,
so we can write $h(z)_{2n-1}=(f+ze)z^{-n}\otimes H_{2n-1}$,
for some $H_{2n-1}\in\mc V(\mf g_{\leq\frac12})$.
Since $(f+ze|f+ze)=4(x|x)z$, we thus  get that
$$
\tint g_n=
\tint 4(x|x) H_{2n+1}
\,\,\text{ for all }\,\,
n\in\mb Z_+
\,.
$$
%
Using the same trick as before, we obtain the values of $\tint g_n$ for $n=0,1$,
by the formulas \eqref{20130405:eq9} and \eqref{20130406:eq3}
for $\pi h(z)_1$ and $\pi h(z)_3$,
and applying at the end the map $\pi^{-1}$, using Corollary \ref{20130402:cor1}.
The results are as follows: $\tint g_0=\tint \tilde{L}$, and
$$
\begin{array}{l}
\displaystyle{
\tint g_1=
\int\Big(
-\frac1{8(x|x)}\tilde{L}^2
-\frac12\sum_{k\in J_1} \psi([f,v^k])\partial\psi([f,v_k])
} \\
\displaystyle{
-\frac12 \sum_{i\in J_0^f,k\in J_\frac12}
\varphi(a_i)\psi([a^i,[f,v^k]])\psi([f,v_k])
\Big)
\,.}
\end{array}
$$
%
The associated evolution equations are ($a\in\mf g_0^f,\,u\in\mf g_{-\frac12}$)
\begin{equation}\label{eq:6.19}
\frac{d\phi(a)}{dt_0}=0
\,,\,\,
\frac{d\psi(u)}{dt_0}=\psi(u)'
-\sum_{i\in J_0^f}\phi(a_i)\psi([a^i,u])
\,,\,\,
\frac{d L}{dt_0}=\widetilde L'
\,,
\end{equation}
and
\begin{equation}\label{eq:6.20}
\begin{array}{rcl}
\displaystyle{
\frac{d\phi(b)}{dt_1}
}
&=&0,
\\
\displaystyle{
\frac{d\psi(u)}{dt_1}
}
&=&
\displaystyle{
\psi(u)'''
-2\sum_{i\in J_0^f}\phi(a_i)\psi([a^i,u])''
-\sum_{i\in J_0^f}\phi(a_i)'\psi([a^i,u])'
}\\
&-&\displaystyle{
\sum_{i\in J_0^f}\!\!\left(
\phi(a_i)\psi([a^i,u])
\right)''
+\sum_{i,j\in J_0^f}\phi(a_i)\phi(a_j)\psi([a^i,[a^j,u])'}\\
&+&\displaystyle{
2\!\!\sum_{i,j\in J_0^f}\phi(a_j)\left(\phi(a_i)\psi([a^i[a^j,u]])\right)'
+\!\!\sum_{i,j\in J_0^f}
\phi(a_i)\psi([a^i[a^j,u]])\phi(a_j)'
}\\
&-&\displaystyle{
\frac3{4(x|x)}\widetilde L \psi(u)'
-\frac3{8(x|x)}\psi(u)\widetilde L'
+\frac1{4(x|x)}\widetilde L\sum_{i\in J_0^f}\phi(a_i)\psi([a^i,u])}\\
&-&\displaystyle{
\frac12\sum_{i\in J_0^f,k\in J_{\frac12}}\psi([a_i,u])\psi([a^i,[f,v^k]])\psi([f,v_k])
}\\
&-&\displaystyle{
\sum_{i,j,l\in J_0^f}\phi(a_i)\phi(a_j)\phi(a_l)\psi([a^i,[a^j,[a^l,u]]]),
}\\
\displaystyle{
\frac{d L}{dt_1}
}
&=&
\displaystyle{
\frac14\widetilde L'''
-\frac3{4(x|x)}\widetilde L\widetilde L'
+\frac32\sum_{k\in J_1}\psi([f,v_k])\psi([f,v^k])''
}\\
&-&\displaystyle
{\frac{3}{2}\sum_{i\in J_0^f,k\in J_{\frac12}}
\left(\phi(a_i)\psi([a^i,[f,v^k]])\psi([f,v_k])\right)'\,.
}
\end{array}
\end{equation}
Recall that all equations \eqref{eq:6.17}--\eqref{eq:6.20} are compatible and have common integrals
of motion due to Remark \ref{20130404:rem}

\begin{remark}\label{20130416:rem}
With respect to the $K$-Poisson structure all the functions $\varphi(a)$, $a\in\mf g^f_{0}$,
are central, and therefore they generate a Poisson vertex algebra ideal $J_K$.
This is clear from Table \ref{table2}.
Therefore, we can consider the PVA $\mc W/J_K$,
generated by the elements $\psi(u)$, $u\in\mf g_{-\frac12}$, and $L$.
The corresponding $\lambda$-brackets induced by $\{\cdot\,_\lambda\,\cdot\}_{K,\rho}$
are given by Table \ref{table2}:
$$
\{\psi(u)_\lambda\psi(u_1)\}=-\omega_-(u,u_1)
\,\,,\,\,\,\,
\{L_\lambda L\}=-4(x|x)\lambda
\,\,,\,\,\,\,
\{L_\lambda\psi(u)\}=0\,.
$$
As a result we obtain the following integrable Hamiltonian equations on the functions
$\psi(u)$, $u\in\mf g_{-\frac12}$, and $L$:
\begin{equation}\label{20130501:eq1}
\begin{array}{rcl}
\displaystyle{
\frac{d\psi(u)}{dt_1}
}
&=&
\displaystyle{
\psi(u)'''
-\frac3{4(x|x)} L \psi(u)'
-\frac3{8(x|x)}\psi(u) L'
}\\
&-&\displaystyle{
\frac12\sum_{i\in J_0^f,k\in J_{\frac12}}\psi([a_i,u])\psi([a^i,[f,v^k]])\psi([f,v_k])
\,,}\\
\displaystyle{
\frac{d L}{dt_1}
}
&=&
\displaystyle{
\frac14 L'''
-\frac3{4(x|x)} LL'
+\frac32\sum_{k\in J_1}\psi([f,v_k])\psi([f,v^k])''
\,.
}
\end{array}
\end{equation}
In the case of $\mf g=\mf{sl}_n$, $n\geq3$, when there exists a non-zero
central element $c\in\mf g_0^f$, we have additional equations
\begin{equation}\label{eq:minimal-reduced}
\begin{array}{l}
\displaystyle{
\frac{d\psi(u)}{dt_{\tilde{1}}}
=\psi([c,u])''-\frac{1}{2(x|x)}L\psi([c,u])
}\\
\displaystyle{
\frac{dL}{dt_{\tilde{1}}}
=\sum_{k\in J_{\frac12}}\left(
\psi([f,v_k])\psi([c,[f,v^k]])
\right)^\prime
\,.
}
\end{array}
\end{equation}
The $H$-Poisson structure does not induce a PVA $\lambda$-bracket on $\mc W/J_K$,
and in order to obtain the second Poisson structure on it
we have to apply Dirac's reduction, developed in \cite{DSKV13}.
\end{remark}

\begin{example}[see Example 3.20 in \cite{DSKV12}]\label{ex:w_3minimal}

Let $\mf g=\mf{sl}_3$. Then $f=E_{31}$ is a minimal nilpotent element which
is embedded in the $\mf{sl}_2$-triple $\{f,h=2x,e\}\subset\mf{sl}_3$, where
$h=E_{11}-E_{33}$ and $e=E_{13}$.
It is easy to check that $\mf g_0^f=\mb Fa$, where $a=E_{11}-2E_{22}+E_{33}$
and $\mf g_{-\frac12}=\mb Fu_1\oplus\mb Fu_2$, where $u_1=E_{21}$ and
$u_2=E_{32}$.
Furthermore, let us write $\mf g_{\frac12}=\mb Fv_1\oplus\mb Fv_2$, where
$v_1=E_{12}$ and $v_2=E_{23}$.
Then $v^1=\frac1{2(x|x)}v_2$ and $v^2=-\frac1{2(x|x)}v_1$.

Let us assume $s=e$ and denote $\phi=\phi(a)$, $\psi_+=\psi(u_1)$ and
$\psi_-=\psi(u_2)$. Then, by Theorem \ref{prop:minimal},
we get the following $\lambda$-bracket on the
generators of $\mc W$:
\begin{align*}
\{L_\lambda L\}_{z,\rho}&=(\partial+2\lambda)L-(x|x)\lambda^3+4(x|x)z\lambda,\\
\{L{}_\lambda \psi_{\pm}\}_{z,\rho}&=(\partial+\frac32\lambda)\psi_{\pm},\\
\{L{}_\lambda \phi\}_{z,\rho}&=(\partial+\lambda)\phi,\\
\{\psi_{\pm}{}_\lambda \psi_{\pm}\}_{z,\rho}&=0,\\
\{\psi_+{}_\lambda \psi_-\}_{z,\rho}&=
-L+\frac1{6(x|x)}\phi^2
-\frac12(\partial+2\lambda)\phi+2(x|x)\lambda^2-2(x|x)z,\\
\{\psi_{\pm}{}_\lambda \phi\}_{z,\rho}&=\pm3\psi_{\pm},\\
\{\phi{}_\lambda \phi\}_{z,\rho}&=12(x|x)\lambda\,.
\end{align*}
According to the above discussion, choosing $a(z)=a$ we get $\tint g_0=\tint \phi$
and $\tint g_1=\frac{3}{2(x|x)}\tint \psi_+\psi_-$.
The corresponding Hamiltonian equations are
$$
\frac{d\phi}{dt_0}=0
\,\,,\,\,\,\,
\frac{d\psi_{\pm}}{dt_0}=\mp3\psi_{\pm}
\,\,,\,\,\,\,
\frac{d L}{dt_0}=0
\,,
$$
and
$$
\begin{array}{l}
\displaystyle{
\frac{d\phi}{dt_1}=0
\,\,,\,\,\,\,
\frac{d L}{dt_1}=\frac{3}{(x|x)}\left(\psi_+\psi_-\right)^\prime
\,,
}\\
\displaystyle{
\frac{d\psi_{\pm}}{dt_1}=
\pm\frac3{2(x|x)} L\psi_{\pm}
\mp\frac{1}{4(x|x)^2}\phi^2\psi_{\pm}
-\frac3{4(x|x)}\psi_{\pm}\phi'
-\frac3{2(x|x)}\phi\psi_{\pm}'
\mp3\psi_{\pm}''
\,.
}
\end{array}
$$
On the other hand, if we choose $a(z)=f+ze$ we get $\tint g_0=\tint \widetilde L$ and
$$
\tint g_1=\int\left(\frac1{4(x|x)}\left(\psi_+\partial\psi_--\psi_-\partial\psi_+\right)
-\frac{1}{8(x|x)^2}\phi\psi_+\psi_-
-\frac{1}{8(x|x)}\widetilde L^2\right)\,.
$$
The corresponding Hamiltonian equations are
$$
\frac{d\phi}{dt_0}=0
\,\,,\,\,\,\,
\frac{d L}{dt_0}=\partial\widetilde L
\,\,,\,\,\,\,
\frac{d\psi_{\pm}}{dt_0}=\psi_{\pm}'
\pm\frac{1}{4(x|x)}\phi\psi_{\pm},
$$
and
$$
\begin{array}{rcl}
\displaystyle{
\frac{d\phi}{dt_1}
}
&=&0,
\\
\displaystyle{
\frac{d\psi_\pm}{dt_1}
}
&=&
\displaystyle{
\psi_\pm'''
\pm\frac{3}{4(x|x)}\phi\psi_{\pm}''
\pm\frac{1}{4(x|x)}\psi_{\pm}\phi''
\pm\frac{3}{4(x|x)}\phi' \psi_{\pm}'
+\frac3{16(x|x)^2}\phi^2\psi_{\pm}'
}\\
&-&
\displaystyle{
\frac3{4(x|x)}\widetilde L \psi_\pm'
-\frac3{8(x|x)}\psi_\pm \widetilde L'
+\frac{3}{16(x|x)^2}\psi_{\pm}\phi \phi'
}\\
&\mp&
\displaystyle{
\frac3{16(x|x)^2}\phi\psi_\pm\widetilde L
\pm\frac{3}{8(x|x)^2}\psi_{\pm}\psi_+\psi_-
\pm\frac{1}{64(x|x)^3}\phi^3\psi_{\pm}\,,
}\\
\displaystyle{
\frac{d L}{dt_1}
}
&=&
\displaystyle{
\frac14\widetilde L'''-\frac3{4(x|x)}\widetilde L\widetilde L'+
\frac3{4(x|x)}\left(\psi_+\psi_-''-\psi_-\psi_+''\right)
-\frac{3}{8(x|x)^2}(\phi\psi_+\psi_-)'\,.
}
\end{array}
$$
As pointed out in Remark \ref{20130416:rem},
$\phi$ generates a central Poisson vertex algebra ideal
with respect to the $K$-Poisson structure,
and we can consider the induced PVA structure 
on the quotient algebra $\quot{\mc W}{(\phi)}$,
where $(\phi)$ denotes the differential algebra ideal generated by $\phi$.
The corresponding $\lambda$-brackets are
\begin{equation}\label{20130501:eq2}
{\{{\psi_{\pm}}_\lambda\psi_{\pm}\}}=0
\,\,,\,\,\,\,
{\{{\psi_+}_\lambda\psi_-\}}=-\frac3{2c}
\,\,,\,\,\,\,
{\{L_\lambda L\}}=\frac3c\lambda
\,\,,\,\,\,\,
{\{L_\lambda\psi_{\pm}\}}=0\,.
\end{equation}
Hence, for an arbitrary non-zero constant $c$, the following Hamiltonian equations on the functions $\psi_+$, $\psi_-$ and $L$ are integrable and compatible:
\begin{equation}\label{20130501:eq3-2}
\begin{array}{l}
\displaystyle{
\frac{d\psi_\pm}{dt_{\tilde1}}
=\pm cL\psi_{\pm}\mp3\psi_{\pm}^{\prime\prime}
\,,
}\\
\displaystyle{
\frac{dL}{dt_{\tilde1}}
=\frac c4\left(\psi_+\psi_-\right)^\prime
\,,}
\end{array}
\end{equation}
and
\begin{equation}\label{20130501:eq3}
\begin{array}{l}
\displaystyle{
\frac{d\psi_\pm}{dt_1}=
\psi_\pm'''
-c L\psi_\pm'
-\frac c2\psi_\pm L'
\pm\frac23c^2\psi_{\pm}\psi_+\psi_-\,,
}\\
\displaystyle{
\frac{d L}{dt_1}=
\frac14L'''-c L L'+
c\left(\psi_+\psi_-''-\psi_-\psi_+''\right)
\,.
}
\end{array}
\end{equation}
\end{example}

\section{Generalized Drinfeld-Sokolov hierarchies for a short nilpotent}
\label{sec:5b}

\subsection{Preliminary computations}

For simplicity we will assume, as in the case of a minimal nilpotent element, that $s=e$.
In this case, $f+ze$ is a semisimple element of $\mf g((z^{-1}))$,
and we can describe explicitly the direct sum decomposition \eqref{eq:dech}.
\begin{lemma}\label{lem:short_dec}
We have the decomposition
$\mf g((z^{-1}))=\mf h\oplus\mf h^{\perp}$, where
\begin{equation}\label{eq:h_short}
\mf h=\Ker\ad(f+ze)=
\mf g_0^f((z^{-1}))\oplus((\ad f)^2-2z)\mf g_1((z^{-1}))\,,
\end{equation}
and
\begin{equation}\label{eq:hperp_short}
\mf h^\perp=\im\ad(f+ze)=
[f,\mf g_1]((z^{-1}))\oplus((\ad f)^2+2z)\mf g_1((z^{-1}))\,.
\end{equation}
\end{lemma}
\begin{proof}
Since $\mf g_0^f=\mf g_0^e$, clearly $\mf g_0^f\subset\ker\ad(f+ze)$.
Moreover, for $v\in\mf g_1$, we have
$$
\ad(f+ze)((\ad f)^2+2z)(v)=2z[f,v]+z[e,[f,[f,v]]]=0\,.
$$
Hence, 
\begin{equation}\label{20130416:eq1}
\mf g_0^f((z^{-1}))\oplus((\ad f)^2-2z)\mf g_1((z^{-1}))\subset\ker\ad(f+ze)\,.
\end{equation}
On the other hand, we have
$\ad(f+ze)\mf g_1=[f,\mf g_1]$, and $\ad(f+ze)[f,\mf g_1]=((\ad f)^2+2z)\mf g_1$.
Therefore, 
\begin{equation}\label{20130416:eq2}
[f,\mf g_1]((z^{-1}))\oplus((\ad f)^2+2z)\mf g_1((z^{-1}))\subset\im\ad(f+ze)\,.
\end{equation}
The equalities \eqref{eq:h_short} and \eqref{eq:hperp_short}
immediately follow from the inclusions \eqref{20130416:eq1} and \eqref{20130416:eq2}.
\end{proof}
\begin{lemma}\label{20130417:lem2}
The decomposition $\mf g((z^{-1}))=\mf h\oplus\mf h^{\perp}$
is a $\mb Z/2\mb Z$-grading of the Lie algebra $\mf g((z^{-1}))$,
namely
$[\mf h,\mf h]\subset\mf h$,
$[\mf h,\mf h^\perp]\subset\mf h^\perp$,
and
$[\mf h^\perp,\mf h^\perp]\subset\mf h$.
\end{lemma}
\begin{proof}
The first two inclusions are immediate, by the definitions of $\mf h$ and $\mf h^\perp$
and the Jacobi identity.
We are left to prove that $[\mf h^\perp,\mf h^\perp]\subset\mf h$.
By Lemma \ref{lem:short_dec}, it is enough to prove that,
for $v,v_1\in\mf g_1$, we have
\begin{enumerate}[(i)]
\item
$[[f,v],[f,v_1]]\in\mf g_0^f$,
\item
$[[f,v],((\ad f)^2+2z)v_1]\in((\ad f)^2-2z)\mf g_1$,
\item
$[((\ad f)^2+2z)v,((\ad f)^2+2z)v_1]\subset\mf g_0^f$.
\end{enumerate}
Part (i) is stated in Lemma \ref{20130417:lem1}.
By the Jacobi identity and part (i) we also have
$$
\begin{array}{l}
\displaystyle{
\vphantom{\Big(}
[[f,v],((\ad f)^2+2z)v_1]
=
[[f,v],[f,[f,v_1]]]+2z[[f,v],v_1]
=
[f,[v,[f,[f,v_1]]]]
} \\
\displaystyle{
\vphantom{\Big(}
-2z[v,[f,v_1]]
=
[f,[f,[v,[f,v_1]]]]
-2z[v,[f,v_1]]
=
((\ad f)^2-2z)[v,[f,v_1]]
\,,}
\end{array}
$$
proving part (ii).
We are left to prove part (iii). Again by the Jacobi identity, we have
$$
\begin{array}{l}
\displaystyle{
\vphantom{\Big(}
[((\ad f)^2+2z)v,((\ad f)^2+2z)v_1]
=
2z[[f,[f,v]],v_1]
+2z[v,[f,[f,v_1]]]
} \\
\displaystyle{
\vphantom{\Big(}
=
2z[f,[[f,v],v_1]]
+2z[f,[v,[f,v_1]]
-4z[[f,v],[f,v_1]]
=
-4z[[f,v],[f,v_1]]
\,,}
\end{array}
$$
and this lies in $\mf g_0^f$ by part (i).
\end{proof}
\begin{remark}\label{20130417:rem2}
As an alternative proof of Lemma \ref{20130417:lem2},
we can just note that the element $f+ze$ is a semisimple element of $\mf{sl}_2((z^{-1}))$,
and therefore it is conjugate to a multiple of $x$
over some field extension of $\mb F((z^{-1}))$.
Its eigenvalues on $\mf g((z^{-1}))$ are $0,\pm\sqrt{z}$,
and $\mf h$ is the eigenspace with eigenvalue $0$,
while $\mf h^\perp$ is the sum of eigenspaces with eigenvalues $\pm\sqrt{z}$.
The claim of the lemma follows immediately from this observation.
\end{remark}

Since $d=1$, the degree of $z$ equals $-2$.
It is then easy to find each piece $\mf h_i$ and $\mf h^\perp_i$, $i\in\mb Z$, 
of the decompositions \eqref{20130404:eq6}
(note that, since $f$ is an even nilpotent element, the degrees are all integers).
We have
\begin{enumerate}[(i)]
\item
$\mf h_i=\mf g_0^fz^{-\frac{i}{2}}$ for $i\in2\mb Z$,
\item
$\mf h_i=((\ad f)^2-2z)\mf g_1 z^{-\frac{i+1}{2}}$ for $i\in2\mb Z-1$,
\item
$\mf h^\perp_i=[f,\mf g_1]z^{-\frac{i}{2}}$ for $i\in2\mb Z$,
\item
$\mf h^\perp_i=((\ad f)^2+2z)\mf g_1 z^{-\frac{i+1}{2}}$ for $i\in2\mb Z-1$.
\end{enumerate}

We let, as in Section \ref{sec:3.5},
$\{a_j\}_{j\in J_0^f}$ and $\{a^j\}_{j\in J_0^f}$ be dual bases of $\mf g_0^f$
with respect to the inner product $(\cdot\,|\,\cdot)$,
and $\{u_k\}_{k\in J_1}$ and $\{u^k\}_{k\in J_1}$
be dual bases of $\mf g_{-1}$ and $\mf g_1$ respectively.
Note also that, then,
$\{[e,u_k]\}_{k\in J_1}$ and $\big\{-\frac12[f,u^k]\big\}_{k\in J_1}$
are dual bases of $[f,\mf g_1]\subset\mf g_0$.
Hence, the element $q\in\mf g_{\geq0}\otimes\mc V(\mf g_{\leq0})$ defined
in \eqref{20130404:eq1} is the following
\begin{equation}\label{20130404:eq7b}
q=
\sum_{j\in J_0^f}a^j\otimes a_j
-\frac12\sum_{k\in J_1}[f,u^k]\otimes [e,u_k]
+\sum_{k\in J_1}u^k\otimes u_k
\,.
\end{equation}

As in Section \ref{sec:5a},
we will solve equation \eqref{L0_dsr}
for $U(z)\in\mf h^\perp_{\geq1}\otimes\mc V(\mf g_{\leq0})$ 
and $h(z)\in\mf h_{\geq0}\otimes\mc V(\mf g_{\leq0})$
degree by degree, applying 
the quotient map $\pi:\,\mc V(\mf g_{\leq0})\to\mc V(\mf g^f)$.
Indeed, by Corollary \ref{20130402:cor2} we have maps
$$
\pi:\,\mf g((z^{-1}))\otimes\mc V(\mf g_{\leq0})\twoheadrightarrow\mf g((z^{-1}))\otimes\mc V(\mf g^f)
\,\,,\,\,\,\,
\pi^{-1}:\,\mf g((z^{-1}))\mc V(\mf g^f)\stackrel{\sim}{\to}\mf g((z^{-1}))\otimes\mc W\,,
$$
acting as identity on the first factors,
which restrict to bijections of $\mf g((z^{-1}))\otimes\mc W$.
Applying $\pi$ to both sides of equation \eqref{L0_dsr},
we get
\begin{equation}\label{20130405:eq6b}
\begin{array}{l}
\displaystyle{
e^{\ad(\pi U(z)_1+ \pi U(z)_2+\dots)}
(\partial+(f+ze)\otimes1+\pi q_0+\pi q_1)
} \\
\displaystyle{
=\partial+(f+ze)\otimes1+\pi h(z)_0+\pi h(z)_1+\dots
\,,}
\end{array}
\end{equation}
and, by the formula \eqref{20130404:eq7b} for $q$ we get
$\pi q=\pi q_0+\pi q_1$, where
\begin{equation}\label{20130405:eq7b}
\pi q_{0}=
\sum_{i\in J_0^f}a_i\otimes a^i
\,\,,\,\,\,\,
\pi q_{1}=
q_{1}=
\sum_{k\in J_1}u^k\otimes u_k
\,.
\end{equation}

%
Equating the homogeneous components of degree $0$
in both sides of equation \eqref{20130405:eq6b}, we get
the equation
$$
\pi h(z)_0+[(f+ze)\otimes1,\pi U(z)_1]=\pi q_0\,,
$$
which implies 
$$
\pi h(z)_0
=\pi q_0=
\sum_{i\in J_0^f}a_i\otimes a^i
\,,\,\,
\pi U(z)_1=0
\,.
$$
%
Next, equating the homogeneous components of degree $1$ in both sides 
of equation \eqref{20130405:eq6b}, we get
$$
\pi h(z)_1+[(f+ze)\otimes1,\pi U(z)_2]=\pi q_1\,.
$$
Recalling the expression \eqref{20130405:eq7b} for $\pi q_1$,
and using the obvious decomposition
\begin{equation}\label{20130417:eq1}
u^k=-\frac14((\ad f)^2-2z)u^k z^{-1}+\frac14[f+ze,[f,u^k]] z^{-1}
\,\in\mf h\oplus\mf h^\perp
\,,
\end{equation}
we deduce that
\begin{equation}\label{20130405:eq8b}
\pi h(z)_1
=
-\frac14
\sum_{k\in J_1} 
((\ad f)^2-2z)u^k z^{-1}
\otimes u_k
\,,\,\,
\pi U(z)_2
=
\frac14
\sum_{k\in J_1} 
[f,u^k] z^{-1}
\otimes u_k
\,.
\end{equation}
Similarly, 
equating the homogeneous components of degree $2$ in both sides 
of equation \eqref{20130405:eq6b}, we get
$$
\pi h(z)_2+[(f+ze)\otimes1,\pi U(z)_3]=
-\partial\pi U(z)_2+[\pi U(z)_2,\pi q_0]
\,,
$$
which implies, after a straightforward computation,
$$
\pi h(z)_2
=0
\,,\,\,
\pi U(z)_3
=
\frac1{16}
\sum_{k\in J_1}
((\ad f)^2+2z)u^k z^{-2}
\otimes
\Big(
-\partial u_k
+\sum_{j\in J_0^f}a_j[a^i,u_k]
\Big)
\,.
$$
%
Finally, we compute $\pi h(z)_3$ by equating 
the homogeneous components of degree $2$ in both sides 
of equation \eqref{20130405:eq6b}.
We get
$$
\begin{array}{l}
\displaystyle{
\vphantom{\Big(}
\pi h(z)_3+[(f+ze)\otimes1,\pi U(z)_4]=
-\partial\pi U(z)_3
+[\pi U(z)_3,\pi q_0]
} \\
\displaystyle{
+[\pi U(z)_2,\pi q_1]
+\frac12[\pi U(z)_2,[\pi U(z)_2,(f+ze)\otimes1]]
\,.}
\end{array}
$$
Note that $\partial\pi U(z)_3\in\mf h^\perp\otimes\mc V(\mf g^f)$.
Moreover, for $j\in J_0^f$, $\ad a^j$ commutes with $(\ad f)^2+2z$.
Therefore, 
$[\pi U(z)_3,\pi q_0]\in\mf h^\perp\otimes\mc V(\mf g^f)$.
It follows that the only 
contributions to $\pi h(z)_3$
are the components in $((\ad f)^2-2z)\mf g_1z^{-2}\otimes\mc V(\mf g^f)$ of
$$
\begin{array}{l}
\displaystyle{
[\pi U(z)_2,\pi q_1]
+\frac12[\pi U(z)_2,[\pi U(z)_2,(f+ze)\otimes1]]
} \\
\displaystyle{
=
\frac14
\!
\sum_{h,k\in J_1} 
\!
[[f,u^h],u^k] z^{-1}
\otimes u_h u_k
-\frac1{32}
\!
\sum_{h,k\in J_1} 
\!
\big[
[f,u^h] ,((\ad f)^2+2z)u^k
\big]
z^{-2}
\otimes u_h u_k
.
}
\end{array}
$$
By Lemma \ref{20130417:lem2}, we have
$\big[[f,u^h] ,((\ad f)^2+2z)u^k\big]\in\mf h$,
and according to Lemma \ref{20130417:lem2} 
and the decomposition \eqref{20130417:eq1},
the component in $\mf h$ of $[[f,u^h],u^k]$ can be
expressed in the following two alternative ways:
$$
\pi_{\mf h}[[f,u^h],u^k]
=
\frac14[[f,u^h],
((\ad f)^2+2z)u^k] z^{-1}
=
-\frac14((\ad f)^2-2z)[[f,u^h],u^k] z^{-1}
\,.
$$
Therefore,
\begin{equation}\label{20130406:eq3b}
\begin{array}{l}
\displaystyle{
\pi h(z)_3
=
\frac1{32}
\sum_{h,k\in J_1} 
[[f,u^h],
((\ad f)^2+2z)u^k] z^{-2}
\otimes u_h u_k
} \\
\displaystyle{
=
-\frac1{32}
\sum_{h,k\in J_1} 
((\ad f)^2-2z)[[f,u^h],u^k] z^{-2}
\otimes u_h u_k
} \\
\displaystyle{
=
\frac1{32}
\sum_{k\in J_1} 
((\ad f)^2-2z)u^k z^{-2}
\otimes
\sum_{h\in J_1} 
[[f,u^h],u_k]u_h
\,.
}
\end{array}
\end{equation}
In the last identity we used the completeness relations \eqref{complete2}.

\subsection{First few equations of the hierarchies}

We make the choice $a(z)=f+ze$.
According to Theorem \ref{final},
there is an associated integrable bi-Hamiltonian hierarchy
corresponding to the Laurent series $\tint g(z)\in\quot{\mc W}{\partial\mc W}((z^{-1}))$
defined by 
$$
\tint g(z)=\tint \big((f+ze)\otimes1|h(z)\big)\,.
$$
We write the integrals of motion $\tint g_n$, $n\in\mb Z_+$,
in terms of the various components of $h(z)\in\mf h\otimes\mc V(\mf g_{\leq0})$.
Recall that $\mf h_{2n}=\mf g_0^fz^{-n}$, which is orthogonal to $e$ and $f$ w.r.t. $(\cdot\,|\,\cdot)$.
Therefore the even components of $h(z)$ do not contribute to $\tint g(z)$.
Furthermore, 
we have $\mf h_{2n-1}=((\ad f)^2-2z)\mf g_1 z^{-n}$,
so we can write 
$$
h(z)_{2n-1}=
\sum_{k\in J_1}
((\ad f)^2-2z)u^k z^{-n}\otimes H_{2n-1,k}
\,,
$$
for some $H_{2n-1,k}\in\mc V(\mf g_{\leq0})$.
For $n=1,2$ we get, from 
the formulas \eqref{20130405:eq8b} and \eqref{20130406:eq3b}
for $\pi h(z)_1$ and $\pi h(z)_3$, that
\begin{equation}\label{20130417:eq5}
\pi H_{1,k}
=
-\frac14 u_k
\,\,,\,\,\,\,
\pi H_{3,k}
=
\frac1{32}
\sum_{h\in J_1} 
[[f,u^h],u_k]u_h
\,,
\end{equation}
for $k\in J_1$.
Note also that
$$
(f+ze|((\ad f)^2-2z)u^k)
=-2z(f|u^k)+z(e|(\ad f)^2u^k)
=-4z(f|u^k)
\,.
$$
Therefore, $\tint g(z)=\sum_{n\in\mb Z_+}\tint g_nz^{-n}$, where
$$
\tint g_n
=
-4\sum_{k\in J_1}
\tint (f|u^k) H_{2n+1,k}
\,.
$$
%
As in Section \ref{sec:6.2}, 
we obtain the values of $\tint g_n$ for $n=0,1$,
from the formulas \eqref{20130417:eq5} for $\pi H_{1,k}$ and $\pi H_{3,k}$,
and applying at the end the map $\pi^{-1}$, using Corollary \ref{20130402:cor2}.
The results are as follows: 
$$
\tint g_0
=
\tint \psi(f)
\,,
$$
and
$$
\tint g_1
=
\frac1{8}
\int 
\sum_{k\in J_1} 
\psi([f,[f,u^k]])\psi(u_k)
\,.
$$
%
Recalling Table \ref{table4}, we get the corresponding hierarchy of Hamiltonian equations
$\frac{dw}{dt_n}=\{{g_n}_\lambda w\}_{H,\rho}\big|_{\lambda=0}$.
They are ($a\in\mf g_0^f$, $u\in\mf g_{-1}$):
$$
\frac{d a}{dt_0}=0
\,\,,\,\,\,\,
\frac{d\psi(u)}{dt_0}=\psi(u)'
-\sum_{i\in J_0^f}\psi([a^i,u])a_i
\,,
$$
and
$$
\begin{array}{l}
\displaystyle{
\frac{d a}{dt_1}=0
\,\,,\,\,\,\,
\frac{d\psi(u)}{dt_1}=
\frac14\psi(u)'''
-\frac14\sum_{i\in J_0^f}\psi([a^i,u])a_i''
-\frac34\sum_{i\in J_0^f}a_i' \psi([a^i,u])'
}\\
\displaystyle{
-\frac34\sum_{i\in J_0^f}a_i\psi([a^i,u])''
+\frac12\sum_{i,j\in J_0^f}a_j\psi([a^i,[a^j,u]]) a_i'}\\
\displaystyle{
+\frac34\sum_{i,j\in J_0^f}a_ia_j \psi([a^i,[a^j,u]])'
+\frac14\sum_{i,j\in J_0^f}a_j\psi([a^j,[a^i,u]]) a_i'}\\
\displaystyle{
-\frac34\sum_{k\in J_1}\psi([[f,u^k],u]) \psi(u_k)'
-\frac14\sum_{i,j,l\in J_0^f} a_i a_j a_l\psi([a^i,[a^j,[a^l,u]]])
}\\
\displaystyle{
+\sum_{i\in J_0^f,k\in J_1}\left(
\psi(u_k\circ[a_i,u])-\psi([a_i,u_k]\circ u)
\right)a_i\psi([f,[f,u^k]])\,.
}
\end{array}
$$
%
\begin{remark}\label{20130416:remb}
As noted in Remark \ref{20130416:remb}, 
we easily see from Table \ref{table4} that
all the functions $a\in\mf g^f_{0}$ are central, 
and therefore they generate a Poisson vertex algebra ideal $J_K$.
Therefore, we can consider the PVA $\mc W/J_K$,
generated by elements $\psi(u)$, $u\in\mf g_{-1}$.
The corresponding $\lambda$-brackets 
induced by $\{\cdot\,_\lambda\,\cdot\}_{K,\rho}$
are given by equation \eqref{20130320:eq2b}:
\begin{equation}\label{20130417:eq6}
\begin{array}{l}
\displaystyle{
\{\psi(u)_\lambda\psi(u_1)\}
=
(e|u\circ u_1)\lambda
\,.
}
\end{array}
\end{equation}
As a result we obtain the following integrable equations on the functions
$\psi(u)$, $u\in\mf g_{-1}$:
$$
\frac{d\psi(u)}{dt_0}=\psi(u)'
\,,
$$
and
$$
\frac{d\psi(u)}{dt_1}=
\frac14\psi(u)'''
+\frac34\sum_{h,k\in J_1}
(u^k*u^h
|u)
\psi(u_h) \psi(u_k)'
\,,
$$
where $*$ is the Jordan product on $\mf g_1$ defined in \eqref{20130320:eq1}.
The last equation is, after a rescaling of the variables, the Svinolupov equation
associated to this Jordan product, \cite{Svi91}.
Hence, we proved that the Svinolupov equation is Hamiltonian
with Poisson structure given by \eqref{20130417:eq6}.
The second Poisson structure for this equation can be obtained via 
the Dirac reduction, which will be explained in our forthcoming publication \cite{DSKV13}.
\end{remark}

\subsection{Another example: a generic choice for \texorpdfstring{$s$}{s} when
\texorpdfstring{$\mf g=\mf{sl}_{2n}$}{g=sl(2n)}}

In the case of $\mf g=\mf{sl}_{2n}$ there is a unique conjugacy class of short
nilpotent elements.
Let us write $X\in\mf{sl}_{2n}$ as a block matrix
$$
X=
\begin{pmatrix}
A & B\\
C & D
\end{pmatrix}\,,
$$
where $A,B,C,D\in\mf{gl}_n$ and $\tr A=-\tr D$. Then, the element
$f=\begin{pmatrix}0&0\\ \id_n & 0\end{pmatrix}$ is a short nilpotent element.
Indeed, it
is contained in the following $\mf{sl}_2$-triple: $\{f,h=2x,e\}\subset\mf{sl}_{2n}$,
where
$$
h=
\begin{pmatrix}
\id_n & 0\\
0 & -\id_n
\end{pmatrix}
\quad\text{and}\quad
e=
\begin{pmatrix}
0 & \id_n\\
0 & 0
\end{pmatrix}\,,
$$
and we have the following $\ad x$-eigenspace decomposition:
$\mf{sl}_{2n}=\mf g_{-1}\oplus\mf g_0\oplus\mf g_{1}$,
where
\begin{align*}
\mf g_{-1}= & \left\{
\begin{pmatrix}
0&0\\
A&0
\end{pmatrix}
\mid A\in\mf{gl}_n\right\},
\quad
\mf g_0=\left\{
\begin{pmatrix}
A&0\\
0&B
\end{pmatrix}
\mid A,B\in\mf{gl}_n, \tr A=-\tr B\right\}\,,\\
& \qquad\qquad\qquad\mf g_{1}=\left\{
\begin{pmatrix}
0&A\\
0&0
\end{pmatrix}
\mid A\in\mf{gl}_n\right\}\,.
\end{align*}
Hence, the adjoint orbit of $f$ consists of all short nilpotent elements of $\mf{sl}_{2n}$. 

We note that the direct sum decomposition \eqref{20130322:eq4} is, in this case,
$\mf g_0=\mf g_0^f\oplus[f,\mf g_1]$, where
$$
\mf g_0^f=\left\{
\begin{pmatrix}
A&0\\
0&A
\end{pmatrix}
\mid A\in\mf{sl}_n\right\}
\quad\text{and}\quad
[f,\mf g_1]=\left\{
\begin{pmatrix}
A&0\\
0&-A
\end{pmatrix}
\mid A\in\mf{gl}_n\right\}\,.
\\
$$
Let us consider
$$
s=\begin{pmatrix}
0&S\\
0&0
\end{pmatrix}\,,
$$
where $S\in\mf{gl}_n$ is a semisimple element with non-zero
eigenvalues. Then, $f+zs\in\mf g((z^{-1}))$ is semisimple and we can describe
explicitly the decomposition \eqref{eq:dech}, generalizing Lemma
\ref{lem:short_dec} (which corresponds to the choice $s=e$).
\begin{lemma}\label{lem:short_dec_sl_2n}
We have the decomposition
$\mf g((z^{-1}))=\mf h\oplus\mf h^{\perp}$, where
\begin{equation}\label{eq:h_short_sl_2n}
\mf h=\Ker\ad(f+zs)=
\left\{
\begin{pmatrix}
A & zSB\\
B & A
\end{pmatrix}
\Big|
\begin{array}{l}
A,B\in\mf{gl}_n((z^{-1})),\\
\left[S,A\right]=\left[S,B\right]=\tr A=0
\end{array}
\right\}\,,
\end{equation}
and
\begin{equation}
\begin{split}\label{eq:hperp_short_sl_2n}
&\mf h^\perp=\im\ad(f+zs)\\
&=
\left\{\!
\begin{pmatrix}
A+C&-zSB+D\\
B+E&-A+F
\end{pmatrix}
\Bigg|
\begin{array}{l}
A,B,C,D,E,F\in\mf{gl}_n((z^{-1})),\\
\left[S,A\right]=\left[S,B\right]=0,\\
C,D,E,F\in\im\{\ad S:\mf{gl}_n\to\mf{gl}_n\}((z^{-1}))
\end{array}\!\!
\right\}\,.
\end{split}
\end{equation}
\end{lemma}
\begin{proof}
Let $X=\begin{pmatrix}A&B\\C&D\end{pmatrix}\in\mf{sl}_{2n}((z^{-1}))$. We have
$$
[f+zs,X]=
\begin{pmatrix}
zSC-B& z(SD-AS)\\
A-D&B-zCS
\end{pmatrix}\,,
$$
from which follows that $\mf h$ given in \eqref{eq:h_short_sl_2n} is contained
in $\ker\ad(f+zs)$.
Moreover, if $X=\begin{pmatrix}A&-zSB\\B&-A\end{pmatrix}$, with $[S,A]=[S,B]=0$,
then
$$
X=\left[f+zs,\begin{pmatrix}\frac12B&-\frac12A\\
\frac1{2z}S^{-1}A&-\frac12B\end{pmatrix}\right]\,,
$$
and, if $X=\begin{pmatrix}C&D\\E&F\end{pmatrix}$, with
$C,D,E,F\in\im\{\ad S:\mf{gl}_n\to\mf{gl}_n\}((z^{-1}))$,
then
$$
X=\left[f+zs,
\begin{pmatrix}E+(\ad S)^{-1}(Dz^{-1}+ES) & F+(\ad S)^{-1}(C+F)S\\
(\ad S)^{-1}(C+F)z^{-1} & (\ad S)^{-1}(Dz^{-1}+ES)\end{pmatrix}\right]\,,
$$
from which follows that $\mf h^{\perp}$ given in \eqref{eq:hperp_short_sl_2n}
is contained in $\im\ad(f+zs)$.
Since, by assumption, $\ad S:\mf{gl}_n\to\mf{gl}_n$ is semisimple,
the equalities \eqref{eq:h_short_sl_2n} and \eqref{eq:hperp_short_sl_2n}
follow.
\end{proof}
By Theorem \ref{final} we can construct a generalized Drinfeld-Sokolov hierarchy of bi-Hamiltonian equations.
However the general formula of the integrals of motion and of the corresponding integrable hierarchy
is very complicated. For $n=1$ one gets the well known KdV hierarchy. In the next example we will
treat in detail the case corresponding to $n=2$.
\begin{example}
Let $\mf g=\mf{sl}_4$ be the Lie algebra of $4\times4$ traceless matrices,
and let us assume $S=\diag(s_1,s_2)$ is a diagonal matrix with distinct eigenvalues (the case $s_1=s_2$
was treated in the previous sections, since in this case $s$ becomes a scalar multiple of $e$).
A basis of $\mf g_0^f$ is given by the following matrices:
\begin{align*}
&A_{11}=E_{11}-E_{22}+E_{33}-E_{44},
\qquad A_{12}=E_{12}+E_{34}\,,\\
&A_{21}=E_{21}+E_{43},
\qquad E_{31},\qquad E_{32},\qquad E_{41},\qquad E_{42}\,.
\end{align*}
By Theorem \ref{20120727:thm1b}, $\mc W$ is the algebra of differential polynomials with
generators $A_{11}$, $A_{12}$, $A_{21}$ and $\psi(E_{ij})$, for $i=3,4$ and $j=1,2$, where
$\psi:\mf g_{-1}\to\mc W\{2\}$ is the map defined by \eqref{psiu}.
The $\lambda$-bracket among generators of $\mc W$ can be computed using Table \ref{table4}.

In order to construct the generalized Drinfeld-Sokolov integrable hierarchy we
need to compute $U(z)\in\mf g((z^{-1}))\otimes\mc V(\mf g_{\leq0})$ and
$h(z)=h(z)_0+h(z)_1+h(z)_2+\ldots\in\mf g((z^{-1}))\otimes\mc V(\mf g_{\leq0})$
satisfying \eqref{L0_dsr}. The first terms of the series $h(z)$ are
(we let $c=(E_{ij}|E_{ji})$, for any $i,j=1,\ldots,4$):
\begin{align}
\begin{split}\label{20130605:eq1}
&h(z)_0=A_{11}\otimes \frac{A_{11}}{4c}\,,
\qquad
h(z)_1=\\
&(E_{13}s_1+E_{32}z^{-1})\otimes\frac1{8c^2s_1(s_1-s_2)}
\left((3s_1+s_2)A_{12}A_{21}-4c(s_1-s_2)\psi(E_{31})\right)\\
&+(E_{24}s_2+E_{42}z^{-1})\otimes\frac1{8c^2s_2(s_2-s_1)}
\left((s_1+3s_2)A_{12}A_{21}-4c(s_1-s_2)\psi(E_{42})\right)\,,\\
&h(z)_2=A_{11}z^{-1}\otimes\frac{(s_1+s_2)}{4c^3(s_1-s_2)^2}
\left(cA_{12}'A_{21}-cA_{12}A_{21}'-A_{11}A_{12}A_{21}\right)\\
&\qquad\qquad\qquad\quad
+\frac{1}{2c^2(s_1-s_2)}\left(A_{12}\psi(E_{41})+A_{21}\psi(E_{32})
\right)\,.
\end{split}
\end{align}
By Theorem \ref{final} we get an integrable hierarchy of bi-Hamiltonian equations
for any $0\neq a(z)\in Z(\mf h)$, where $\mf h$ is defined in \eqref{eq:h_short_sl_2n},
and the integrals of motion are the coefficients of the power series 
$\tint g(z)=\tint (a(z)\otimes1|h(z))=\tint g_0+g_1z^{-1}+\ldots$, whose first terms 
can be computed using \eqref{20130605:eq1}.
We consider the following three choices of $a(z)$:
\begin{enumerate}[(a)]
\item $a(z)=f+zs$;
\item $a(z)=s^{-1}+ez$;
\item $a(z)=A_{11}$.
\end{enumerate}

In case (a) we get
$$
\tint g_0=\int\psi(f)+\frac1{4c}A_{12}A_{21}\,.
$$
The corresponding system of Hamiltonian equations is
\begin{align*}
&\frac{d A_{11}}{dt_0}=0\,,
\qquad
\frac{d A_{12}}{dt_0}=A_{12}'-\frac{A_{11}A_{12}}{2c}\,,
\qquad
\frac{d A_{21}}{dt_0}=A_{21}'-\frac{A_{11}A_{21}}{2c}\,,\\
&\frac{d \psi(E_{31})}{dt_0}=\psi(E_{31})'\,,
\qquad
\frac{d \psi(E_{32})}{dt_0}=\psi(E_{32})'-\frac{A_{11}\psi(E_{32})}{2c}\,,\\
&\frac{d \psi(E_{41})}{dt_0}=\psi(E_{41})'+\frac{A_{11}\psi(E_{41})}{2c}\,,
\qquad
\frac{d \psi(E_{42})}{dt_0}=\psi(E_{42})'\,.
\end{align*}

In case (b) we get
$$
\tint g_0=\int\psi(s^{-1})-\frac{s_1+s_2}{4cs_1s_2}A_{12}A_{21}\,.
$$
The corresponding system of Hamiltonian equations is
\begin{align*}
&\frac{d A_{11}}{dt_0}=0\,,
\qquad
\frac{d A_{12}}{dt_0}
=\frac{s_1+s_2}{4cs_1s_2}(A_{11}A_{12}-2cA_{12}')
-\frac{s_1-s_2}{s_1s_2}\psi(E_{32})\,,
\\
&\frac{d A_{21}}{dt_0}
=\frac{s_1+s_2}{4cs_1s_2}(-A_{11}A_{21}-2cA_{21}')
+\frac{s_1-s_2}{s_1s_2}\psi(E_{41})\,,
\\
&\frac{d \psi(E_{31})}{dt_0}
=\frac{\psi(E_{31})'}{s_1}
+\frac{3(s_1-s_2)}{8cs_1s_2}(A_{12}A_{21})'
+\frac{3s_1+s_2}{4cs_1s_2}\left(A_{21}\psi(E_{32})-A_{12}\psi(E_{41})
\right)\,,
\\
&\frac{d \psi(E_{32})}{dt_0}
=\frac{s_1+s_2}{4cs_1s_2}
\left(2c\psi(E_{32})'-A_{11}\psi(E_{32})\right)
+\frac{A_{12}}{c}\left(\frac{\psi(E_{31})}{s_1}-\frac{\psi(E_{42})}{s_2}\right)\\
&\qquad\qquad
+\frac{s_1-s_2}{16cs_1s_2}\left(4cA_{12}''-3A_{11}A_{12}'-A_{11}'A_{12}
+A_{11}^2A_{12}+4A_{12}^2A_{21}\right)\,,
\\
&\frac{d \psi(E_{41})}{dt_0}=\frac{s_1+s_2}{4cs_1s_2}
\left(2c\psi(E_{41})'+A_{11}\psi(E_{41})\right)
+\frac{A_{21}}{c}\left(\frac{\psi(E_{41})}{s_2}-\frac{\psi(E_{31})}{s_1}\right)\\
&\qquad\qquad
-\frac{s_1-s_2}{16cs_1s_2}\left(4cA_{21}''+3A_{11}A_{21}'+A_{11}'A_{21}
+A_{11}^2A_{21}+4A_{12}A_{21}^2\right)\,,
\\
&\frac{d \psi(E_{42})}{dt_0}
=\frac{\psi(E_{42})'}{s_2}
-\frac{3(s_1-s_2)}{8cs_1s_2}(A_{12}A_{21})'
+\frac{s_1+3s_2}{4cs_1s_2}\left(A_{12}\psi(E_{41})-A_{21}\psi(E_{32})
\right)\,.
\end{align*}

Finally, in case (c) we get $\tint g_0=\tint A_{11}$ and
$$
\tint g_1\!\!=\!\int\!\! \frac{(s_1+s_2)}{c^2(s_1-s_2)^2}
\!\left(2cA_{12}'A_{21}\!-\!A_{11}A_{12}A_{21}\right)
+\frac{2}{c(s_1-s_2)}\!\left(A_{12}\psi(E_{41})
\!+\!A_{21}\psi(E_{32})\right).
$$
The system of Hamiltonian equations corresponding to $\tint g_0$ is
\begin{gather*}
\frac{d A_{11}}{dt_0}=0\,,
\qquad
\frac{d A_{12}}{dt_0}=2A_{12}\,,
\qquad
\frac{d A_{21}}{dt_0}=-2A_{21}\,,\\
\frac{d \psi(E_{31})}{dt_0}=0\,,
\quad
\frac{d \psi(E_{32})}{dt_0}=2\psi(E_{32})\,,
\quad
\frac{d \psi(E_{41})}{dt_0}=-2\psi(E_{41})\,,
\quad
\frac{d \psi(E_{42})}{dt_0}=0\,.
\end{gather*}
The system of Hamiltonian equations corresponding to $\tint g_1$ has a much
more complicated expression.
However, since $A_{11}\in\mf g_0^s$, we can see from Table \ref{table4}
that it is a central element for the $K$ Poisson structure. Hence it generates
a central PVA ideal $J_K$. Then, we can consider the quotient PVA $\quot{\mc W}{J_K}$
and the corresponding reduced Hamiltonian equations.
The reduced Hamiltonian equation corresponding to $\tint g_1$ becomes
\begin{align*}
\frac{d A_{12}}{dt_1}&
=\frac{s_1+s_2}{c^2(s_1-s_2)^2}\left(
4c^2A_{12}''-2A_{12}^2A_{21}\right)\\
&+\frac{2}{c(s_1-s_2)}\left(2c\psi(E_{32})'-A_{12}(\psi(E_{31})-\psi(E_{42}))
\right)\,,
\\
\frac{d A_{21}}{dt_1}&
=\frac{s_1+s_2}{c^2(s_1-s_2)^2}
\left(-4c^2A_{21}''+2A_{12}A_{21}^2\right)\\
&+\frac{2}{c(s_1-s_2)}\left(2c\psi(E_{41})'+A_{21}(\psi(E_{31})-\psi(E_{42}))
\right)\,,
\\
\frac{d\psi(E_{31})}{dt_1}&
=\frac{1}{c^2(s_1-s_2)}
\left(cA_{12}\psi(E_{41})'+cA_{21}\psi(E_{32})'
+cA_{12}A_{21}''-cA_{12}''A_{21}\right)\\
&+\frac{2s_1}{c(s_1-s_2)^2}\left(A_{12}'\psi(E_{41})+A_{21}'\psi(E_{32})\right)\,,
\\
\frac{d\psi(E_{32})}{dt_1}&
=\frac{s_1+s_2}{c^2(s_1-s_2)^2}
\left(2cA_{12}'(\psi(E_{42})-\psi(E_{31}))-2A_{12}A_{21}\psi(E_{32})\right)\\
&
+\frac{1}{8c^3(s_1-s_2)}\left(
-8c^3A_{12}'''
+6cA_{12}^2A_{21}'+16c^2A_{12}'\psi(f)+8c^2A_{12}\psi(f)'
\right.\\
&\left.-6cA_{12}A_{12}'A_{21}
+16c^2\psi(E_{32})(\psi(E_{42})-\psi(E_{31}))
\right)\,,
\\
\frac{d\psi(E_{41})}{dt_1}&
=\frac{s_1+s_2}{c^2(s_1-s_2)^2}
\left(2cA_{21}'(\psi(E_{42})-\psi(E_{31}))+2A_{12}A_{21}\psi(E_{41})\right)\\
&
+\frac{1}{8c^3(s_1-s_2)}
\left(-8c^3A_{21}'''
+6cA_{12}'A_{21}^2+16c^2A_{21}'\psi(f)+8c^2A_{21}\psi(f)'\right.\\
&\left.-6cA_{12}A_{21}A_{21}'
+16c^2\psi(E_{41})(\psi(E_{42})-\psi(E_{31}))
\right)\,,
\\
\frac{d\psi(E_{42})}{dt_1}&
=\frac{1}{c^2(s_1-s_2)}
\left(cA_{12}\psi(E_{41})'+cA_{21}\psi(E_{32})'+cA_{12}''A_{21}-cA_{12}A_{21}''\right)\\
&
-\frac{2s_2}{c(s_1-s_2)^2}\left(A_{12}'\psi(E_{41})+A_{21}'\psi(E_{32})\right)\,.
\end{align*}
\end{example}

\section{Generalized Drinfeld-Sokolov hierarchies for a minimal nilpotent
and a choice of a maximal isotropic subspace}\label{sec:8}

\subsection{The classical \texorpdfstring{$\mc W$}{W}-algebra $\mc W(\mf l)$}
\label{sec:8.1}

In the previous publication \cite{DSKV12} we considered a slightly more general
construction of the classical $\mc W$-algebras,
associated to the nilpotent element $f$ and a choice of an isotropic subspace
$\mf l\subset\mf g_{\frac12}$,
with respect to the skew-symmetric bilinear form $\omega_+$ defined in \eqref{20130201:eq5}.
The construction for $\mf l=0$ is described in Section \ref{sec:2.2}.

In the present section we consider the case when $f$ is a minimal nilpotent element,
and $\mf l\subset\mf g_{\frac12}$ is a maximal isotropic subspace,
that is $\mf l=\mf l^{\perp\omega_+}$.
We review the construction of the $\mc W$-algebra in this case,
and we study the associated generalized Drinfeld-Sokolov 
integrable bi-Hamiltonian hierarchies.

We fix throughout the section a maximal isotropic subspace $\mf l\subset\mf g_{\frac12}$,
and a maximal isotropic subspace $\mf l'\subset\mf g_{\frac12}$ complementary to $\mf l$.
Let $\{v_k\}_{k\in L}$ be a basis of $\mf l$,
and let $\{v^k\}_{k\in L}$ be the $\omega_+$-dual basis of $\mf l'$: 
\begin{equation}\label{20130507:eq3}
\omega_+(v_h,v^k)=\delta_{h,k}
\,.
\end{equation}
Then, clearly, a basis for $\mf g_{\frac12}$ is
$$
\{v_k\}_{k\in J_{\frac12}}
=\{v_k\}_{k\in L}\cup\{v^k\}_{k\in L}
\,,
$$
and the $\omega_+$-dual basis, again of $\mf g_{\frac12}$, is
$$
\{v^k\}_{k\in J_{\frac12}}
=\{v^k\}_{k\in L}\cup\{-v_k\}_{k\in L}\,.
$$

Using the same notation as in \cite{DSKV12}, we consider the direct sum decomposition
$\mf g=\mf n\oplus\mf p$,
where
$$
\mf n=\mf l\oplus\mf g_{1}
\,\,\text{ and }\,\,
\mf p=\mf l'\oplus\mf g_{\leq0}
\,.
$$
The orthogonal complement to $\mf n$ w.r.t. $(\cdot\,|\,\cdot)$ is 
$$
\mf n^\perp=[f,\mf l]\oplus\mf g_{\geq0}
\,.
$$
Note that, since $\mf l$ is isotropic, we have $[\mf n,\mf n]=0$,
thus we can choose $s$ to be any homogeneous 
(with respect to the decomposition \eqref{dec}) element $s\in\mf n$.

The \emph{classical} $\mc W$-\emph{algebra} $\mc W(\mf l)$ is, by definition,
the differential algebra
$$
\mc W(\mf l)
=\big\{g\in\mc V(\mf p)\,\big|\,a\,^\rho_\lambda\,g=0\,\text{ for all }a\in\mf n\}\,,
$$
endowed with the following PVA $\lambda$-bracket
$$
\{g_\lambda h\}_{z,\rho}^{\mf l}=\rho_{\mf l}\{g_\lambda h\}_z,
\qquad g,h\in\mc W\,,
$$
where the $\lambda$-bracket $\{\cdot_{\lambda}\cdot\}_z$ is defined in \eqref{lambda}
and $\rho_{\mf l}:\,\mc V(\mf g)\twoheadrightarrow\mc V(\mf p)$ is
the differential algebra homomorphism defined on generators by
$$
\rho_\mf l(a)=\pi_{\mf p}(a)+(f| a),
\qquad a\in\mf g\,,
$$
where $\pi_{\mf p}:\mf g\to\mf p$ denotes the projection with kernel $\mf n$.

It is proved in \cite{DSKV12} that, for $z=0$,
the $\mc W$-algebra $\mc W(\mf l)$ associated to $\mf l$
is isomorphic to the $\mc W$-algebra $\mc W$
(associated to the choice $\mf l=0$).
Moreover, for $s=e$, we get a PVA isomorphism $\mc W\stackrel{\sim}{\longrightarrow}\mc W(\mf l)$ for arbitrary $z$.
We can describe explicitly the map
$\mc W\stackrel{\sim}{\longrightarrow}\mc W(\mf l)$,
considered as a differential algebra isomorphism.
It is given by the restriction to $\mc W\subset\mc V(\mf g_{\leq\frac12})$ of the differential
algebra homomorphism $\pi_{\mf p}:\mc V(\mf g_{\leq\frac12})\to\mc V(\mf p)$,
extending the projection map $\pi_{\mf p}:\mf g_{\leq\frac12}\to\mf p$
(with kernel $\mf l$).
It follows by this observation and Theorem \ref{20120727:thm1}
that $\mc W(\mf l)$
is the algebra of differential polynomials with the following generators:
the energy momentum $L_{\mf l}$, defined as
(cf. equation \eqref{virL})
$$
L_{\mf l}=f+x'+\frac12\sum_{j\in J_0}a_j a^j+\sum_{k\in L}[f,v_k]v^k\,,
$$
and elements of conformal weight $1$ and $\frac32$,
given by the bijective maps
(cf. equations \eqref{phi} and \eqref{psi})
$\varphi_{\mf l}:\,\mf g_{0}^f\to\mc W(\mf l)\{1\}$, given by
$$
\varphi_{\mf l}(a)=a+\frac12\sum_{k\in L}\pi_{\mf p}([a,v_k])v^k
\,,
$$
and $\psi_{\mf l}:\,\mf g_{-\frac12}\to\mc W(\mf l)\{\frac32\}$, given by
$$
\psi_{\mf l}(u)=
u+\frac13\sum_{h,k\in L}\pi_{\mf p}([[u,v_h],v_k])v^hv^k
+\sum_{k\in L}[u,v_k]v^k+\partial\pi_{\mf p}[e,u]\,.
$$

Consider the quotient map $\pi_{\mf l}:\,\mf p\to\mf g^f$ 
with kernel $[e,\mf g_{\leq\frac12}]\cap\mf p=\mb Fx\oplus\mf l'$,
and extend it to a surjective differential algebra homomorphism
$\pi_{\mf l}:\,\mc V(\mf p)\twoheadrightarrow\mc V(\mf g^f)$.
We have an analogue of Corollary \ref{20130402:cor1}:
\begin{corollary}\label{20130402:cor1-bis}
The quotient map 
$\pi_{\mf l}:\,\mc V(\mf p)\twoheadrightarrow\mc V(\mf g^f)$
restricts to a differential algebra isomorphism
$\pi_{\mf l}:\,\mc W_{\mf l}\stackrel{\sim}{\longrightarrow}\mc V(\mf g^f)$,
and the inverse map $\pi_{\mf l}^{-1}:\,\mc V(\mf g^f)\stackrel{\sim}{\longrightarrow}\mc W_{\mf l}$
is defined on generators by 
$$
\begin{array}{l}
\displaystyle{
\vphantom{\Big(}
\pi_{\mf l}^{-1}(a)=\varphi_{\mf l}(a)
\,\,\,\, \text{ for } a\in\mf g_0^f
\,,} \\
\displaystyle{
\vphantom{\Big(}
\pi_{\mf l}^{-1}(u)=\psi_{\mf l}(u)
\,\,\,\, \text{ for } u\in\mf g_{-\frac12}
\,,} \\
\displaystyle{
\pi_{\mf l}^{-1}(f)=L_{\mf l}-\frac12\sum_{i\in J_0^f}\varphi_{\mf l}(a_i)\varphi_{\mf l}(a^i)
=:\tilde{L}_{\mf l}
\,.}
\end{array}
$$
\end{corollary}

If we take $s=e$,
then the $\lambda$-brackets among the generators are the same as the ones 
given by Table \ref{table2}.
If we take $s\in\mf l$, then the $\{\cdot_{\lambda}\cdot\}_{H,\rho}^{\mf l}$
$\lambda$-bracket is the same as in Table \ref{table2},
and it is not hard to check, using Corollary \ref{20130402:cor1-bis}, that 
the $\{\cdot_{\lambda}\cdot\}_{K,\rho}^{\mf l}$ $\lambda$-bracket is given by Table \ref{table5}:

\begin {table}[H]
\caption{$K$-$\lambda$-brackets among generators of $\mc W(\mf l)$ 
for minimal nilpotent $f$ and maximal isotropic $\mf l$} 
\label{table5} 
\begin{center}
\begin{tabular}{c||c|c|c}
\phantom{$\Bigg($} 
$\{\cdot\,_\lambda\,\cdot\}_{K,\rho}^{\mf l}$
& $L_{\mf l}$ & $\varphi_{\mf l}(b)$ & $\psi_{\mf l}(u_1)$ \\
\hline
\hline \phantom{$\Bigg($} 
$L_{\mf l}$ & $0$ & $0$ & $-\frac32(s|u_1)\lambda$ \\
\hline \phantom{$\Bigg($}
$\varphi_{\mf l}(a)$ & $0$ & $0$ & $(s|[u_1,a])$ \\
\hline \phantom{$\Bigg($}
$\psi_{\mf l}(u)$ & $-\frac32(s|u)\lambda$
& $(s|[b,u])$ & $ 0 $
\end{tabular}
\end{center}
\end{table}

\subsection{Embeddable elements $s\in\mf g_{\frac12}$}

Clearly, if we chose $s=e$ we get the same generalized Drinfeld-Sokolov hierarchies
of bi-Hamiltonian equations as for the case when $\mf l=0$ (cf. Section \ref{sec:5a}). 
Hence, let us assume $s\in\mf l$.
As usual, we need to require that $f+zs$ is a semisimple element of $\mf g((z^{-1}))$.
This is guaranteed (cf. Proposition \ref{20130419:lem1} below) if we assume that $s$ satisfies the following property:
\begin{definition}
An element $s\in\mf g_{\frac12}$ is called \emph{embeddable} (with respect to $x$)
if there exists $s^*\in\mf g_{-\frac12}$ such that $[s,s^*]=2x$.
\end{definition}
\begin{lemma}\label{20130503:lem}
If $s\in\mf g_{\frac12}$ is embeddable, then $2s^*,4x,s$ is an $\mf{sl}_2$-triple and:
\begin{enumerate}[(a)]
\item
$\mf g_0^s\subset\mf g_0^f$;
\item
the maps $\ad(f):\,\mf g^s_{\frac12}\to[f,\mf g^s_{\frac12}]\subset\mf g_{-\frac12}$,
$\ad(s)\circ\ad(f):\,\mf g^s_{\frac12}\to[s,[f,\mf g^s_{\frac12}]]\subset\mf g_{0}$,
and $\ad(s)\circ\ad(s)\circ\ad(f):\,\mf g^s_{\frac12}\to[s,[s,[f,\mf g^s_{\frac12}]]]\subset\mf g_{\frac12}$,
are bijective;
\item
we have the orthogonal (w.r.t. $(\cdot\,|\,\cdot)$) decomposition
$\mf g_0=\mf g_0^s\oplus[s,\mf g_{-\frac12}]$;
\item
we have the dual (w.r.t. $(\cdot\,|\,\cdot)$) decompositions
$\mf g_{\frac12}=\mf g_{\frac12}^s\oplus\mb F[s^*,e]$,
and
$\mf g_{-\frac12}=\mf g_{-\frac12}^{s^*}\oplus\mb F[s,f]$;
\item
we have the orthogonal decompositions
$\mf g_0=\mf g_0^f\oplus\mb Fx$,
and $\mf g_0^f=\mf g_0^s\oplus[s,[f,\mf g^s_{\frac12}]]$;
\item
we have the dual (w.r.t. $(\cdot\,|\,\cdot)$) decompositions
$\mf g_{\frac12}=[s,[s,[f,\mf g_{\frac12}^s]]]\oplus\mb Fs$,
and $\mf g_{-\frac12}=[f,\mf g_{\frac12}^s]\oplus\mb Fs^*$;
\item
we have the dual (w.r.t. $(\cdot\,|\,\cdot)$) decompositions
$\mf g_{\frac12}=[s,\mf g_0^f]\oplus\mb Fs$,
and $\mf g_{-\frac12}=[s^*,\mf g_0^f]\oplus\mb Fs^*$;
\item
we have a non-degenerate, symmetric, $\mf g_0^s$-invariant bilinear 
form $\chi$ on $\mf g^s_{\frac12}$ given by
$\chi(u,v)=([s,[f,u]]|[s,[f,v]])$.
\end{enumerate}
\end{lemma}
\begin{proof}
The fact that $2s^*,4x,s$ is an $\mf{sl}_2$-triple follows by the definition
of embeddable element.
Recall from Section \ref{sec:3.1} that $\mf g_0^f$ is the orthocomplement to $x$ in $\mf g_0$.
Hence, in order to prove part (a) we only need to show that $(x|a)=0$ for every $a\in\mf g_0^s$.
This follows by the identity $x=\frac12[s,s^*]$ and by the invariance of the bilinear form.

Part (b) is immediate from representation theory of $\mf{sl}_2$.

For part (c), the fact that $\mf g_0$ admits 
the direct sum decomposition $\mf g_0^s\oplus[s,\mf g_{-\frac12}]$ 
follows by representation theory of $\mf{sl}_2$,
and the fact that this decomposition is orthogonal 
follows by invariance of the bilinear form $(\cdot\,|\,\cdot)$.

Similarly, for part (d), By representation theory of $\mf{sl}_2$
we immediately have the direct sum decompositions
$\mf g_{\frac12}=\mf g_{\frac12}^s\oplus\mb F[s^*,e]$,
and
$\mf g_{-\frac12}=\mf g_{-\frac12}^{s^*}\oplus\mb F[s,f]$.
These decompositions are dual since, by invariance,
$([s^*,e]|\mf g_{-\frac12}^{s^*})=0$ and $([s,f]|\mf g_{\frac12}^s)=0$.

Next, we prove part (e). By part (d) we have that
$\mf g_{-\frac12}=[f,\mf g_{\frac12}]=[f,\mf g^s_{\frac12}]\oplus\mb F s^*$.
Here we used Lemma \ref{20130315:lem2}(a).
Hence, by part (c), we get that
$\mf g_0=\mf g^s_0\oplus[s,[f,\mf g^s_{\frac12}]]\oplus\mb F x$.
We already know from part (a) that $\mf g^s_0\subset\mf g^f_0$.
Moreover, we claim that $[s,[f,\mf g^s_{\frac12}]]\subset\mf g^f_0$.
Indeed, if $u\in\mf g^s_{\frac12}$, then $[f,[s,[f,u]]]\in\mf g_{-1}=\mb Ff$,
and $(e|[f,[s,[f,u]]])=(2x|[s,[f,u]])=(s|[f,u])=0$.
Hence, we have the direct sum decomposition 
$\mf g^f_0=\mf g^s_0\oplus[s,[f,\mf g^s_{\frac12}]]$.
This decomposition is obviously orthogonal, by invariance of the bilinear form $(\cdot\,|\,\cdot)$.

We already proved that $\mf g_{-\frac12}=[f,\mf g^s_{\frac12}]\oplus\mb F s^*$.
Moreover, by part (b) $[s,[s,[f,\mf g^s_{\frac12}]]]$ is a subspace of $\mf g_{\frac12}$
of codimension 1.
For $u\in\mf g^s_{\frac12}$, we have $(s^*|[s,[s,[f,u]]])=-2(x|[s,[f,u]])=-(s|[f,u])=0$.
In other words, $[s,[s,[f,\mf g^s_{\frac12}]]]$ is orthogonal to $s^*$,
and therefore it does not contain $s$.
To conclude the proof of part (f) we just need to observe that, obviously, $s$
is orthogonal to $[f,\mf g^s_{\frac12}]$.

Next we prove part (g). The direct sum decompositions
$\mf g_{\frac12}=[s,\mf g_0^f]\oplus\mb Fs$
and $\mf g_{-\frac12}=[s^*,\mf g_0^f]\oplus\mb Fs^*$
follow from representation theory of $\mf{sl}_2$.
These decompositions are dual with respect to $(\cdot\,|\,\cdot)$
by invariance and since $\mf g^f_0$ is orthogonal to $x$.

Finally, we prove part (h).
The inner product $\chi$ is clearly symmetric, and it is non-degenerate by part (f).
The fact that it is $\mf g^s_0$-invariant 
follows by invariance of $(\cdot\,|\,\cdot)$ and the fact that $\mf g^s_0\subset\mf g^f_0$.
\end{proof}

According to Lemma \ref{20130503:lem},
$(\cdot\,|\,\cdot)$ restricts to a non-degenerate symmetric bilinear form on $\mf g^s_0$,
and we fix an orthonormal basis
\begin{equation}\label{20130507:eq1}
\{a_i\}_{i\in J^s_0}\subset\mf g^s_0
\,\,\text{ such that }\,\,
(a_i|a_i)=\delta_{i,j}\,.
\end{equation}
Moreover, we have the non-degenerate symmetric bilinear form $\chi$ on $\mf g^s_{\frac12}$,
and we fix an orthonormal basis
\begin{equation}\label{20130507:eq2}
\{u_k\}_{k\in J^s_{\frac12}}\subset\mf g^s_{\frac12}
\,\,\text{ such that }\,\,
\chi(u_h|u_k)=([s,[f,u_h]]|[s,[f,u_k]])=\delta_{h,k}\,.
\end{equation}
Then, by part (e) in Lemma \ref{20130503:lem},
dual (w.r.t. $(\cdot\,|\,\cdot)$) bases of $\mf g_0$ are
$$
\begin{array}{l}
\vphantom{\Big(}
\displaystyle{
\{a_i\}_{i\in J_0}=\{x\}\cup\{a_i\}_{i\in J^s_0}\cup\{[s,[f,u_k]]\}_{k\in J^s_{\frac12}}
\,,} \\
\vphantom{\Big)}
\displaystyle{
\{a^i\}_{i\in J_0}=\Big\{\frac1{(x|x)}x\Big\}\cup\{a_i\}_{i\in J^s_0}\cup\{[s,[f,u_k]]\}_{k\in J^s_{\frac12}}
\,\,\subset\mf g_0
\,,}
\end{array}
$$
while by part (f) in Lemma \ref{20130503:lem},
we have the following basis of $\mf g_{\frac12}$
$$
\{v_k\}_{k\in J_{\frac12}}=\{[s,[s,[f,u_k]]]\}_{k\in J^s_{\frac12}}\cup\{s\}
\,\subset\mf g_{\frac12}
\,,
$$
and the dual (w.r.t. $(\cdot\,|\,\cdot)$) basis of $\mf g_{-\frac12}$ is
$$
\{-[f,u_k]\}_{k\in J^s_{\frac12}}\cup\Big\{\frac1{4(x|x)}s^*\Big\}
\,\subset\mf g_{-\frac12}
\,.
$$

\begin{proposition}\label{20130503:prop}
Let $\mf g$ be  a simple Lie algebra not of type $C_n$.
Then the set of embeddable elements $s\in\mf g_{\frac12}$
is a non empty Zariski open subset of $\mf g_{\frac12}$.
Moreover, any two embeddable elements $s\in\mf g_{\frac12}$
can be obtained from one another by the adjoint action of $G_0^f$ 
up to a non-zero constant factor.
\end{proposition}
\begin{proof}
%
We identify $\mf g$ with $\mf g^*$ using the invariant bilinear form
normalized so that $(\theta|\theta)=2$.
Choose root vectors $e_{\alpha_1}$, $e_{-\alpha_1}$, $e_{\theta-\alpha_1}$,
$e_{-\theta+\alpha_1}$,
such that $(e_{\alpha_1}|e_{-\alpha_1})=1$,
and $(e_{\theta-\alpha_1}|e_{-\theta+\alpha_1})=1$.
Let $s=e_{\alpha_1}+e_{\theta-\alpha_1}$,
and $s^*=e_{-\alpha_1}+e_{-\theta+\alpha_1}$.
Then, $[s,s^*]=\theta=2x$.
This proves that the set of embeddable elements in $\mf g_{\frac12}$ is not empty.
By Kostant's Theorem \cite[Thm.4.2]{Kos59},
the $G_0$-orbit of any embeddable element $s$ is Zariski open in $\mf g_{\frac12}$.
Since $\mf g_0=\mf g^f_0\oplus\mb Fx$, the last statement of the proposition follows.
\end{proof}

\subsection{Preliminary computations}\label{sec:8.3}

\begin{proposition}\label{20130419:lem1}
If $s$ is embeddable,
then $f+zs\in\mf g((z^{-1}))$ is semisimple,
and we have the decomposition 
$\mf g((z^{-1}))=\mf h\oplus\mf h^{\perp}$, where
\begin{equation}\label{20130419:eq1}
\mf h:=\mb F(f+zs)((z^{-1}))\oplus\mf g_0^s((z^{-1}))
\oplus\mb F(e+z^{-1}s^*)((z^{-1}))
=\ker\ad(f+zs)
\end{equation}
and
\begin{equation}\label{20130419:eq2}
\begin{array}{c}
\vphantom{\Big(}
\displaystyle{
\mf h^{\perp}
:=\mb F(2f-zs)((z^{-1}))
\oplus\mb Fx((z^{-1}))
\oplus\mb F(2e-z^{-1}s^*)((z^{-1}))
\oplus[f,\mf g^s_{\frac12}]((z^{-1}))
} \\
\vphantom{\Big(}
\displaystyle{
\oplus[s,[f,\mf g^s_{\frac12}]]((z^{-1}))
\oplus[s,[s,[f,\mf g^s_{\frac12}]]]((z^{-1}))
=\im\ad(f+zs)
\,.
}
\end{array}
\end{equation}
In fact, we have
\begin{enumerate}[(i)]
\item
$[f+zs,2x]=2f-zs$;
\item
$[f+zs,-\frac16(2e-z^{-1}s^*)]=x$;
\item
$[f+zs,[s^*,e]z^{-1}]=2e-z^{-1}s^*$;
\item
$[f+zs,u]=[f,u]\in[f,g^s_{\frac12}]$,
for all $u\in\mf g^s_{\frac12}$;
\item
$[f+zs,[f,u]z^{-1}]=[s,[f,u]]\in[s,[f,g^s_{\frac12}]]$,
for all $u\in\mf g^s_{\frac12}$;
\item
$[f+zs,[s,[f,u]]z^{-1}]=[s,[s,[f,u]]]\in[s,[s,[f,g^s_{\frac12}]]]$,
for all $u\in\mf g^s_{\frac12}$.
\end{enumerate}
\end{proposition}
\begin{proof}
Clearly, $f+zs\in\ker\ad(f+zs)$.
By Lemma \ref{20130503:lem}(a), $\mf g_0^s\subset\mf g_0^f$,
and therefore $\mf g_0^s\subset\ker\ad(f+zs)$.
Next, we have
$[f+zs,e+z^{-1}s^*]=[f,e]+[s,s^*]=0$,
by definition of admissible element.
Therefore, we have the inclusion $\mf h\subset\ker(\ad(f+zs)$.

Furthermore, all identities (i)--(v) are immediate.
For identity (vi) we just have to show that $[f,[s,[f,u]]]=0$ for all $u\in\mf g^s_{\frac12}$.
Since, clearly, $[f,[s,[f,u]]]\in\mf g_{-1}=\mb Ff$,
the claim follows from the following identity:
$(e|[f,[s,[f,u]]])=2(x|[s,[f,u]])=(s|[f,u])=0$.

The identities (i)--(vi) imply the inclusion $\mf h^\perp\subset\im(f+zs)$.
In order to conclude the proof of the proposition,
we are left to prove, recalling Lemma \ref{20130506:lem},
that $\mf h\oplus\mf h^\perp=\mf g((z^{-1}))$.
This immediately follows from parts (b), (e) and (f) in Lemma \ref{20130503:lem}.
\end{proof}
\begin{remark}\label{20130506:rem1}
If $\mf g$ is of type $C_n$,
then there is no element $s\in\mf g_{\frac12}$ such that $f+s$ is semisimple.
Indeed, in this case $\mf g_0^f$ is a simple Lie algebra of type $C_{n-1}$,
and its representation on $\mf g_{\frac12}$ is the standard $2n-2$-dimensional representation.
Therefore, the set of non-zero elements of $\mf g_{\frac12}$ form a single $G_0^f$-orbit.
Hence, the $G_0^f$-orbit of $f+s$ is $f+\mf g_{\frac12}\backslash\{0\}$,
and its closure contains $f$.
A fortiori, $f$ lies in the closure of the $G$-orbit of $f+s$.
But $f$ is nilpotent, so it does not lie in the $G$ orbit of the semisimple element $f+s$.
This is a contradiction since, by Lemma \ref{20130506:lem}(ii),
the $G$-orbit of $f+s$ is closed.
\end{remark}
\begin{remark}\label{20130507:rem}
It is not difficult to show that the converse to Proposition \ref{20130419:lem1} holds:
the set of elements $s\in\mf g_{\frac12}$ such that $f+s$ is semisimple in $\mf g$
coincides with the set of embeddable elements.
\end{remark}

We generalize the argument in Section \ref{sec:4} to this case.
Let us assume that $s\in\mf g_{\frac12}$ is embeddable, 
so that, by Proposition \ref{20130419:lem1},
we have the direct sum decomposition $\mf g((z^{-1}))=\mf h\oplus\mf h^\perp$,
given by \eqref{20130419:eq1} and \eqref{20130419:eq2}.
Since $s\in\mf g_{\frac12}$, we let $z$ have degree $-\frac32$
and we get the corresponding decompositions for $\mf h$  and $\mf h^{\perp}$ 
given by \eqref{20130404:eq6}.
To find the homogeneous subspaces, we use the obvious set identity:
$$
\frac12\mb Z=\Big(\frac32\mb Z-1\Big)\sqcup\frac32\mb Z\sqcup\big(\frac32\mb Z+1)\,.
$$
Hence, we get:
\begin{enumerate}[(i)]
\item
$\mf h_i=\mb F (f+zs)z^{-\frac23(i+1)}$ for $i\in\frac32\mb Z-1$,
\item
$\mf h_i=\mf g_0^sz^{-\frac23 i}$ for $i\in\frac32\mb Z$,
\item
$\mf h_i=\mb F (e+z^{-1}s^*)z^{-\frac23 (i-1)}$ for $i\in\frac32\mb Z+1$,
\item
$\mf h^\perp_i=
\mb F (2f-zs)z^{-\frac23(i+1)}
\oplus[s,[s,[f,\mf g^s_{\frac12}]]]z^{-\frac23\big(i-\frac12\big)}$ 
for $i\in\frac32\mb Z-1$,
\item
$\mf h^\perp_i=\mb Fxz^{-\frac23 i}\oplus[s,[f,\mf g^s_{\frac12}]]z^{-\frac23 i}$ for $i\in\frac32\mb Z$,
\item
$\mf h^\perp_i=
[f,\mf g^s_{\frac12}]z^{-\frac23 \big(i+\frac12\big)}\oplus\mb F(2e-z^{-1}s^*)z^{-\frac23 (i-1)}$
for $i\in\frac32\mb Z+1$.
\end{enumerate}

\subsection{Generalized Drinfeld-Sokolov hierarchies}

Consider the element 
$$
q=\sum_{i\in P}q^i\otimes q_i\in\mf n^{\perp}\otimes\mc V(\mf p)\,,
$$
where $\{q_i\}_{i\in P}$ is a basis of $\mf p$,
and $\{q^i\}_{i\in P}$ is the dual (w.r.t. $(\cdot\,|\,\cdot)$) basis of $\mf n^\perp$.
In terms of the bases \eqref{20130507:eq3}, \eqref{20130507:eq1}, and \eqref{20130507:eq2}, 
we can write $q$ as follows:
$$
\begin{array}{l}
\displaystyle{
q=
\frac 1{2(x|x)} e\otimes f
+\frac1{4(x|x)} s\otimes s^*
+\frac 1{(x|x)} x\otimes x
+\sum_{i\in J_0^s}a_i\otimes a_i
} \\
\displaystyle{
+\sum_{k\in J^s_{\frac12}} [s,[f,u_k]]\otimes[s,[f,u_k]]
-\sum_{k\in J^s_{\frac12}} [s,[s,[f,u_k]]]\otimes[f,u_k]
+\sum_{k\in L} [f,v_k]\otimes v^k
\,.}
\end{array}
$$
Its image under the map $\pi_{\mf l}:\,\mc V(\mf p)\to\mc V(\mf g^f)$
defined in Section \ref{sec:8.1} is
$$
\begin{array}{l}
\displaystyle{
\pi_{\mf l}q=
\frac 1{2(x|x)} e\otimes f
+\frac1{4(x|x)}s\otimes s^*
+\sum_{i\in J_0^s}a_i\otimes a_i
} \\
\displaystyle{
+\sum_{k\in J^s_{\frac12}} [s,[f,u_k]]\otimes[s,[f,u_k]]
-\sum_{k\in J^s_{\frac12}} [s,[s,[f,u_k]]]\otimes[f,u_k]
\,,}
\end{array}
$$
and its homogeneous components $\pi_{\mf l}q_i\in\mf n^\perp_i\otimes\mc V(\mf g^f)$
are
\begin{equation}\label{20130507:eq4}
\begin{array}{l}
\displaystyle{
\pi_{\mf l}q_0=
\sum_{i\in J_0^s}a_i\otimes a_i
+\sum_{k\in J^s_{\frac12}} [s,[f,u_k]]\otimes[s,[f,u_k]]
\,,} \\
\displaystyle{
\pi_{\mf l}q_{\frac12}=
\frac1{4(x|x)}s\otimes s^*
-\sum_{k\in J^s_{\frac12}} [s,[s,[f,u_k]]]\otimes[f,u_k]
\,,} \\
\displaystyle{
\pi_{\mf l}q_1=
\frac 1{2(x|x)} e\otimes f
\,.}
\end{array}
\end{equation}

Using the results in \cite[Sec.4]{DSKV12} and the same trick as in Section \ref{sec:5a},
we construct the generalized Drinfeld-Sokolov integrable bi-Hamiltonian hierarchy
as follows:
\begin{enumerate}[1.]
\item
We solve equation \eqref{L0_dsr} for
$U(z)\in\mf h^{\perp}_{\geq\frac12}\otimes\mc V(\mf p)$ 
and $h(z)\in\mf h_{\geq-\frac12}\otimes\mc V(\mf p)$,
after applying the map $\pi_{\mf l}$ to both sides of the equation.
\item
We fix an element $a(z)$ in the center of $\mf h$,
and we compute the Laurent series $\tint g(z)=\sum_{n\in\mb Z_+}\tint g_nz^{N-n}$ 
defined by \eqref{gz}.
The corresponding bi-Hamiltonian integrable hierarchy is
$\frac{dw}{dt_n}=\{{g_n}_0 w\}_{H,\rho}$, $n\in\mb Z_+$.
\end{enumerate}

\subsection*{Step 1}

By \cite[Prop.4.5]{DSKV12}
there are unique 
$U(z)=U(z)_{\frac12}+U(z)_{1}+U(z)_{\frac32}+\dots\in\mf h^{\perp}_{\geq\frac12}\otimes\mc V(\mf p)$ 
and $h(z)=h(z)_{-\frac12}+h(z)_{0}+h(z)_{\frac12}+\dots\in\mf h_{\geq-\frac12}\otimes\mc V(\mf p)$
solving equation \eqref{L0_dsr}.
Using the same trick as in Section \ref{sec:5a},
we apply the map $\pi_{\mf l}$ to both sides of the equation \eqref{L0_dsr},
and we solve it degree by degree,
for $\pi_{\mf l} U(z)_{i+1}\in\mf h^{\perp}_{i+1}\otimes\mc V(\mf p)$ 
and $\pi_{\mf l} h(z)_i\in\mf h_{i}\otimes\mc V(\mf p)$.

If we apply the map $\pi_{\mf l}$ to both sides of equation \eqref{L0_dsr}, we get
\begin{equation}\label{20130422:eq2}
\begin{array}{l}
\displaystyle{
e^{\ad(\pi_{\mf l} U(z)_{\frac12}+\pi_{\mf l} U(z)_1+ \pi_{\mf l} U(z)_{\frac32}+\dots)}
(\partial+(f+zs)\otimes1+\pi_{\mf l} q_0+\pi_{\mf l} q_{\frac12}+\pi_{\mf l} q_1)
} \\
\displaystyle{
=\partial+(f+zs)\otimes1+\pi_{\mf l} h(z)_{-\frac12}+\pi_{\mf l} h(z)_0+\pi_{\mf l} h(z)_{\frac12}+\dots
\,.}
\end{array}
\end{equation}

If we take the homogeneous components of degree $-\frac12$
in both sides of equation \eqref{20130422:eq2}, we get
$$
\pi_{\mf l} h(z)_{-\frac12}+[(f+zs)\otimes1,\pi_{\mf l} U(z)_{\frac12}]=0\,,
$$
which implies $\pi_{\mf l} U(z)_{\frac12}=\pi_{\mf l} h(z)_{-\frac12}=0$.

Similarly, if we take the homogeneous components of degree $0$ in both sides 
of equation \eqref{20130405:eq6}, we get
$$
\pi_{\mf l} h(z)_0+[(f+zs)\otimes1,\pi_{\mf l} U(z)_1]
=
\pi_{\mf l} q_0
\,.
$$
By Proposition \ref{20130419:lem1} we have
$\mf h_0=\mf g_0^s$ and $\mf h_0^\perp=\mb Fx\oplus[s,[f,\mf g^s_{\frac12}]]$.
It follows by the expression \eqref{20130507:eq4} for $\pi_{\mf l}q_0$
and Proposition \ref{20130419:lem1}(v) that
\begin{equation}\label{20130422:eq3}
\pi_{\mf l} h(z)_0
=\sum_{i\in J_0^s}a_i\otimes a_i
\,\,,\,\,\,\,
\pi_{\mf l} U(z)_1
=\sum_{k\in J_{\frac12}^s}[f,u_k]z^{-1}\otimes [s,[f,u_k]]
\,.
\end{equation}

Next, we take the homogeneous components of degree $\frac12$ in both sides 
of equation \eqref{20130422:eq2}:
$$
\pi_{\mf l} h(z)_{\frac12}+[(f+zs)\otimes1,\pi_{\mf l} U(z)_{\frac32}]
=
\pi_{\mf l} q_{\frac12}
\,.
$$
By Proposition \ref{20130419:lem1} we have
$\mf h_{\frac12}=\mb F(f+zs)z^{-1}$
and $\mf h^\perp_{\frac12}=\mb F(2f-zs)z^{-1}\oplus[s,[s,[f,\mf g^s_{\frac12}]]]
=\ad(f+zs)\big(
\mb F(2x)z^{-1}\oplus[s,[f,\mf g^s_{\frac12}]]z^{-1}\big)$.
It follows from the expression \eqref{20130507:eq4} for $\pi_{\mf l}q_{\frac12}$,
and the obvious decomposition $s=\frac23(f+zs)z^{-1}-\frac13(2f-zs)z^{-1}$,
that
\begin{equation}\label{20130422:eq4}
\begin{array}{l}
\displaystyle{
\pi_{\mf l} h(z)_{\frac12}
=
\frac1{6(x|x)}(f+zs)z^{-1}\otimes s^*
\,}\\
\displaystyle{
\pi_{\mf l} U(z)_{\frac32}
=
-\frac1{6(x|x)}xz^{-1}\otimes s^*
-\sum_{k\in J_{\frac12}^s}[s,[f,u_k]]z^{-1}\otimes [f,u_k]
\,.}
\end{array}
\end{equation}

Next, we take the homogeneous components of degree $1$ in both sides of equation \eqref{20130422:eq2}:
$$
\begin{array}{c}
\vphantom{\Big(}
\displaystyle{
\pi_{\mf l} h(z)_1
+[(f+zs)\otimes1,\pi_{\mf l} U(z)_2]
=
\pi_{\mf l} q_1
-\partial\pi_{\mf l}U(z)_1+[\pi_{\mf l}U(z)_1,\pi_{\mf l} q_0]
} \\
\vphantom{\Big(}
\displaystyle{
+\frac12[\pi_{\mf l} U(z)_1,[\pi_{\mf l} U(z)_1,(f+zs)\otimes1]]
\,.}
\end{array}
$$
It follows from the formulas \eqref{20130507:eq4} for $\pi_{\mf l}q_{0}$ and $\pi_{\mf l}q_{1}$,
and \eqref{20130422:eq3} for $\pi_{\mf l} U(z)_1$,
that
$$
\begin{array}{l}
\vphantom{\Big(}
\displaystyle{
\pi_{\mf l} h(z)_1
+[(f+zs)\otimes1,\pi_{\mf l} U(z)_2]
=
\frac 1{2(x|x)} e\otimes f
-\sum_{k\in J_{\frac12}^s}[f,u_k]z^{-1}\otimes \partial [s,[f,u_k]]
}\\
\displaystyle{
- \sum_{i\in J_0^s} \sum_{k\in J_{\frac12}^s}
[f,[a_i,u_k]] z^{-1}\otimes a_i [s,[f,u_k]]
}\\
\displaystyle{
+\sum_{h,k\in J_{\frac12}^s}
[ [f,u_h] , [s,[f,u_k]] ]
z^{-1}\otimes [s,[f,u_h]] [s,[f,u_k]]
}\\
\displaystyle{
-\frac12 \sum_{h,k\in J_{\frac12}^s}
[[f,u_h],[s,[f,u_k]]]
z^{-1}\otimes [s,[f,u_h]] [s,[f,u_k]]
\,.}
\end{array}
$$
By Proposition \ref{20130419:lem1} we have
$\mf h_{1}=\mb F(e+z^{-1}s^*)$
and $\mf h^\perp_{1}=\mb F(2e-z^{-1}s^*)\oplus[f,\mf g^s_{\frac12}]z^{-1}$.
We thus get,
by the orthogonality condition \eqref{20130507:eq2},
and the obvious identities 
$e=\frac13(e+z^{-1}s^*)+\frac13(2e-z^{-1}s^*)$
and $s^*=\frac23(e+z^{-1}s^*)z-\frac13(2e-z^{-1}s^*)z$,
\begin{equation}\label{20130507:eq5}
\pi_{\mf l} h(z)_1
=
\frac 1{6(x|x)} (e+z^{-1}s^*) \otimes f
+\frac1{12(x|x)} (e+z^{-1}s^*) \otimes \sum_{k\in J_{\frac12}^s} [s,[f,u_k]] [s,[f,u_k]]
\end{equation}

Finally, we take the homogeneous components of degree $\frac32$ 
in both sides of equation \eqref{20130422:eq2}:
$$
\begin{array}{c}
\vphantom{\Big(}
\displaystyle{
\pi_{\mf l} h(z)_{\frac32}
\!+\![(f+zs)\!\otimes\!1,\pi_{\mf l} U(z)_{\frac52}]
=
-\partial\pi_{\mf l}U(z)_{\frac32}
\!+[\pi_{\mf l}U(z)_{\frac32},\pi_{\mf l} q_0]
+[\pi_{\mf l}U(z)_1,\pi_{\mf l} q_{\frac12}]
} \\
\vphantom{\Big(}
\displaystyle{
+\frac12[\pi_{\mf l} U(z)_1,[\pi_{\mf l} U(z)_{\frac32},(f+zs)\otimes1]]
+\frac12[\pi_{\mf l} U(z)_{\frac32},[\pi_{\mf l} U(z)_1,(f+zs)\otimes1]]
\,.}
\end{array}
$$
We can use equations \eqref{20130507:eq4}, \eqref{20130422:eq3} and \eqref{20130422:eq4},
to rewrite the above identity as follows
$$
\begin{array}{l}
\displaystyle{
\pi_{\mf l} h(z)_{\frac32}
+[(f+zs)\otimes1,\pi_{\mf l} U(z)_{\frac52}]
=
\frac1{6(x|x)}xz^{-1}\otimes \partial s^*
}\\
\displaystyle{
+\sum_{k\in J_{\frac12}^s}[s,[f,u_k]]z^{-1}\otimes \partial [f,u_k]
+\sum_{i\in J_0^s} \sum_{k\in J_{\frac12}^s}
[s,[f,[a_i,u_k]]] z^{-1} \otimes [f,u_k]a_i
}\\
\displaystyle{
-\sum_{h,k\in J_{\frac12}^s}
[[s,[f,u_h]],[s,[f,u_k]]] z^{-1}\otimes [f,u_h] [s,[f,u_k]]
}\\
\displaystyle{
-\frac1{4(x|x)} \sum_{k\in J_{\frac12}^s}
[s,[f,u_k]] z^{-1} \otimes s^* [s,[f,u_k]]
}\\
\displaystyle{
-\sum_{h,k\in J_{\frac12}^s}
[[f,u_h],[s,[s,[f,u_k]]]]
z^{-1}\otimes [s,[f,u_h]] [f,u_k]
}\\
\displaystyle{
+\frac1{24(x|x)} \sum_{k\in J_{\frac12}^s}
[s,[f,u_k]]
z^{-1}\otimes s^* [s,[f,u_k]]
}\\
\displaystyle{
+\frac12 \sum_{h,k\in J_{\frac12}^s}
[[f,u_h],[s,[s,[f,u_k]]]]
z^{-1}\otimes [s,[f,u_h]] [f,u_k]
}\\
\displaystyle{
+\frac12 \sum_{h,k\in J_{\frac12}^s}
[[s,[f,u_h]],[s,[f,u_k]]]
z^{-1}\otimes [f,u_h] [s,[f,u_k]]
\,.}\\
\end{array}
$$
By Proposition \ref{20130419:lem1} we have
$\mf h_{\frac32}=\mf g^s_0z^{-1}$
and $\mf h^\perp_{\frac32}=\mb Fxz^{-1}\oplus[s,[f,\mf g^s_{\frac12}]]z^{-1}$.
It follows by the above equation and
the orthogonality relations \eqref{20130507:eq1} and \eqref{20130507:eq2},
that
\begin{equation}\label{20130507:eq5b}
\pi_{\mf l} h(z)_{\frac32}
=
-\sum_{i\in J^s_0} \sum_{k\in J_{\frac12}^s}
a_i z^{-1}\otimes[f,u_k] [s,[f,[a_i,u_k]]]
\,.
\end{equation}
Here we used the fact that, for $u,v\in\mf g^s_{\frac12}$, 
the component in $\mf g^s_0$
of $[[s,[f,u]],[s,[f,v]]\in\mf g^f_0$ is equal to $\sum_{i\in J^s_0}\chi([a_i,u],v)a_i$.

\subsection*{Step 2}

Given an element $0\neq a(z)\in Z(\mf h)$,
consider the Laurent series $\tint g(z)=\sum_{n\in\mb Z_+}\tint g_nz^{N-n}$ 
defined by \eqref{gz}.
According to \cite[Thm.4.18]{DSKV12},
the coefficients $\tint g_n,\,n\in\mb Z_+$ of this series 
span an infinite-dimensional subspace of $\quot{\mc W}{\partial\mc W}$
and they satisfy the Lenard-Magri recursion conditions \eqref{eq:lenard2}. 

In order to compute the coefficient of $\tint g(z)$,
we first apply the surjective map
$\pi_{\mf l}:\,\mf g((z^{-1}))\otimes\mc V(\mf p)\twoheadrightarrow\mf g((z^{-1}))\otimes\mc V(\mf g^f)$
(acting as the identity on the first factor, and as a differential algebra homomorphism 
on the second factor) to $h(z)$ inside equation \eqref{gz}, 
and then we apply the inverse map
$\pi_{\mf l}^{-1}:\,\mc V(\mf g^f)
\stackrel{\sim}{\rightarrow}\mf g((z^{-1}))\otimes\mc W_{\mf l}$,
using Corollary \ref{20130402:cor1-bis}.
The resulting equation is
$$
\tint g(z)=\tint \pi_{\mf l}^{-1}(a(z)\otimes1 | \pi_{\mf l}h(z))
\,.
$$

We will consider three different choices of $a(z)\in\ Z(\mf h)$:
\begin{enumerate}[(a)]
\item $a(z)=c\in Z(\mf g_0^s)$,
\item $a(z)=e+z^{-1}s^*$,
\item $a(z)=f+zs$.
\end{enumerate}

\subsubsection*{Case (a): $a(z)=c\in Z(\mf g_0^s)$}

Note that $c\in\mf g^s_0$ is orthogonal to $f,e,s,s^*$.
Hence, 
according to the description of the
subspaces $\mf h_i$, $i\in\frac12\mb Z_+$, in Section \ref{sec:8.3},
we obtain 
$\tint g(z)=\sum_{i\in\frac32\mb Z_+}\tint \pi_{\mf l}^{-1}(a(z)\otimes1 | \pi_{\mf l}h(z)_i)$.
We then write, for $n\in\mb Z_+$,
\begin{equation}\label{20130507:eq7}
\pi_{\mf l}h(z)_{\frac32n}=
\sum_{i\in J^s_0} a_iz^{-n}\otimes h_{n,i}
\,,
\end{equation}
with $h_{n,i}\in\mc V(\mf g^f)$.
It follows that, for every $n\in\mb Z_+$, we have
\begin{equation}\label{20130507:eq8}
\tint g_n=
\sum_{i\in J^s_0}
\tint \pi_{\mf l}^{-1}
(c|a_i) h_{n,i}
\,,
\end{equation}
From equations \eqref{20130422:eq3}, \eqref{20130507:eq7} and \eqref{20130507:eq8}
we thus get, using Corollary \ref{20130402:cor1-bis},
$$
\tint g_0=
\tint \varphi_{\mf l}(c)
\,,
$$
while from equation \eqref{20130507:eq5b} we get
$$
\tint g_1=
-\sum_{k\in J^s_{\frac12}}
\tint
\psi_{\mf l}([f,u_k])\varphi_{\mf l}([s,[f,[c,u_k]]])
\,.
$$
The corresponding Hamiltonian equations of the hierarchy
$\frac{dw}{dt_n}=\{{g_n}_0w\}^{\mf l}_{H,\rho}$ are
$$
\frac{d\phi_{\mf l}(a)}{dt_0}=\phi([c,a])
\,\,,\,\,\,\,
\frac{d\psi_{\mf l}(u)}{dt_0}=\psi([c,u])
\,\,,\,\,\,\,
\frac{dL_{\mf l}}{dt_0}=0\,,
$$
and
$$
\begin{array}{l}
\vphantom{\Big(}
\displaystyle{
\frac{d\phi_{\mf l}(a)}{dt_1}=
\sum_{k\in J^s_{\frac12}}
\varphi_{\mf l}([s,[f,[c,u_k]]]) \psi_{\mf l}([f,[a,u_k]])
} \\
\vphantom{\Big(}
\displaystyle{
+\sum_{k\in J^s_{\frac12}}
\varphi_{\mf l}([a,[s,[f,[c,u_k]]]]) \psi_{\mf l}([f,u_k])
-\sum_{k\in J^s_{\frac12}}
(a|[s,[f,[c,u_k]]]) \psi_{\mf l}([f,u_k])'
} \\
\vphantom{\Big(}
\displaystyle{
\frac{d\psi_{\mf l}(u)}{dt_1}=
\sum_{k\in J^s_{\frac12}}
\psi_{\mf l}([u,[s,[f,[c,u_k]]]])
\psi_{\mf l}([f,u_k])
} \\
\vphantom{\Big(}
\displaystyle{
-\sum_{k\in J^s_{\frac12}}
\sum_{r\in J_{\frac12}}
\varphi_{\mf l}([[f,u_k],v^r]^\sharp)
\phi_{\mf l}([u,v_r]^\sharp)
\varphi_{\mf l}([s,[f,[c,u_k]]])
} \\
\vphantom{\Big(}
\displaystyle{
-\sum_{k\in J^s_{\frac12}}
\phi_{\mf l}([[f,u_k],[e,u]]^\sharp)'
\varphi_{\mf l}([s,[f,[c,u_k]]])
} \\
\vphantom{\Big(}
\displaystyle{
-2\sum_{k\in J^s_{\frac12}}
\phi_{\mf l}([[f,u_k],[e,u]]^\sharp)
\varphi_{\mf l}([s,[f,[c,u_k]]])'
} \\
\vphantom{\Big(}
\displaystyle{
-\sum_{k\in J^s_{\frac12}}
\frac{(u_k|u)}{2(x|x)}
\tilde{L}_{\mf l}
\varphi_{\mf l}([s,[f,[c,u_k]]])
+\sum_{k\in J^s_{\frac12}}
(u_k|u)
\varphi_{\mf l}([s,[f,[c,u_k]]])''
\,,} \\
\vphantom{\Big(}
\displaystyle{
\frac{dL_{\mf l}}{dt_1}=
-\frac32
\sum_{k\in J^s_{\frac12}}
\Big(\psi_{\mf l}([f,u_k])\varphi_{\mf l}([s,[f,[c,u_k]]])\Big)'
\,.}
\end{array}
$$

\subsubsection*{Case (b): $a(z)=e+z^{-1}s^*$}

Note that $e+z^{-1}s^*$ is orthogonal to $\mf g^s_0$ and to $e+z^{-1}s^*$.
Hence, 
according to the description of the
subspaces $\mf h_i$, $i\in\frac12\mb Z_+$, in Section \ref{sec:8.3},
we obtain 
$\tint g(z)=\sum_{i\in\frac32\mb Z_++\frac12}\tint \pi_{\mf l}^{-1}(a(z)\otimes1 | \pi_{\mf l}h(z)_i)$.
We then write, for $n\in\mb Z_+$,
\begin{equation}\label{20130507:eq7b}
\pi_{\mf l}h(z)_{\frac32n+\frac12}=
(f+zs)z^{-(n+1)}\otimes h_{n,a}
\,,
\end{equation}
with $h_{n,a}\in\mc V(\mf g^f)$.
It follows that
\begin{equation}\label{20130507:eq8b}
\tint g(z)
=\sum_{n\in\mb Z_+}
\tint g_nz^{-(n+1)}
=6(x|x)
\sum_{n\in\mb Z_+}
\tint \pi_{\mf l}^{-1}
h_{n,a}
z^{-(n+1)}
\,.
\end{equation}
From equations \eqref{20130422:eq4}, \eqref{20130507:eq7b} and \eqref{20130507:eq8b}
we thus get, using Corollary \ref{20130402:cor1-bis},
$$
\tint g_0
=
\tint \psi_{\mf l}(s^*)
\,.
$$
The corresponding Hamiltonian equation is
$$
\begin{array}{l}
\displaystyle{
\frac{d\phi_{\mf l}(a)}{dt_0}=\psi_{\mf l}([s^*,a])
\,\,,\,\,\,\,
\frac{dL_{\mf l}}{dt_0}=\frac12\psi_{\mf l}(s^*)'
\,,}\\
\displaystyle{
\frac{d\psi_{\mf l}(u)}{dt_0}
=\sum_{k\in J_{\frac12}}\phi_{\mf l}([s^*,v^k]^{\sharp})\phi_{\mf l}([u,v_k]^{\sharp})
+\frac{([e,s^*]|u)}{2(x|x)}\widetilde L_{\mf l}+\phi_{\mf l}([s^*,[e,u]]^{\sharp})'
\,.}
\end{array}
$$

\subsubsection*{Case (c): $a(z)=f+zs$}

Note that $f+zs$ is orthogonal to $\mf g^s_0$ and to $f+zs$.
Hence, 
according to the description of the
subspaces $\mf h_i$, $i\in\frac12\mb Z_+$, in Section \ref{sec:8.3},
we obtain 
$\tint g(z)=\sum_{i\in\frac32\mb Z_++1}\tint \pi_{\mf l}^{-1}(a(z)\otimes1 | \pi_{\mf l}h(z)_i)$.
We then write, for $n\in\mb Z_+$,
\begin{equation}\label{20130507:eq7c}
\pi_{\mf l}h(z)_{\frac32n+1}=
(e+z^{-1}s^*)z^{-n}\otimes h_{n,b}
\,,
\end{equation}
with $h_{n,b}\in\mc V(\mf g^f)$.
It follows that
\begin{equation}\label{20130507:eq8c}
\tint g(z)
=\sum_{n\in\mb Z_+}
\tint g_nz^{-n}
=6(x|x)
\sum_{n\in\mb Z_+}
\tint \pi_{\mf l}^{-1}
h_{n,b}
z^{-n}
\,.
\end{equation}
From equations \eqref{20130507:eq5}, \eqref{20130507:eq7c} and \eqref{20130507:eq8c}
we thus get, using Corollary \ref{20130402:cor1-bis},
$$
\begin{array}{l}
\displaystyle{
\tint g_0
=
\int L_{\mf l}
-\frac12\sum_{i\in J_0^f}\varphi_{\mf l}(a_i)\varphi_{\mf l}(a^i)
+\frac1{2}
\sum_{k\in J_{\frac12}^s} 
\varphi_{\mf l}([s,[f,u_k]]) \varphi_{\mf l}([s,[f,u_k]])
} \\
\displaystyle{
=
\int L_{\mf l}
-\frac12\sum_{i\in J_0^s}\varphi_{\mf l}(a_i)\varphi_{\mf l}(a^i)
\,.}
\end{array}
$$
In the last identity we used the orthogonal decomposition 
$\mf g^f_0=\mf g^s_0\oplus[s,[f,\mf g^s_{\frac12}]]$.
The corresponding Hamiltonian equation is
$$
\begin{array}{l}
\displaystyle{
\frac{d\phi_{\mf l}(a)}{dt_0}=\phi_{\mf l}(a)'
-\sum_{i\in J_0^s}\left(\phi_{\mf l}([a^i,a])\phi_{\mf l}(a_i)+(a|a^i)\phi_{\mf l}(a_i)'\right)\,,
}\\
\displaystyle{
\frac{d\psi_{\mf l}(u)}{dt_0}=\psi_{\mf l}(u)'+\sum_{i\in J_0^s}\phi_{\mf l}(a_i)\psi_{\mf l}([u,a^i])
\,,} \\
\displaystyle{
\frac{dL_{\mf l}}{dt_0}=\left(L_{\mf l}-\frac12\sum_{i\in J_0^s}\phi_{\mf l}(a_i)\phi_{\mf l}(a^i)\right)'
\,.}
\end{array}
$$


%
\end{document}